\documentclass[12pt, oneside]{article}
\usepackage[utf8]{inputenc}
\usepackage[english]{babel}  
\usepackage{amsmath,amssymb,amsthm,algorithmicx,algorithm,setspace,url}
\usepackage{graphicx,bbm,bm,mathrsfs,subfig,transparent}
\usepackage[top=1in,bottom=1in,left=1in,right=1in]{geometry}
\usepackage[authoryear, sort]{natbib}
\usepackage{authblk,booktabs,epstopdf,color,lscape,float}
\usepackage[space]{grffile}
\usepackage{authblk}
\usepackage{mathtools}
\usepackage{caption}
\date{}
\newtheorem{theorem}{Theorem}[section]
\newtheorem{lemma}[theorem]{Lemma}

\newtheorem{corollary}[theorem]{Corollary}
\newtheorem{condition}[theorem]{Condition}
\newtheorem{definition}[theorem]{Definition}
\newtheorem{remark}[theorem]{Remark}

\allowdisplaybreaks

\title{\bf Bayesian Variable Selection in Multivariate Nonlinear Regression with Graph Structures}
\author[1]{Yabo Niu \thanks{ybniu@stat.tamu.edu}}
\author[2]{Nilabja Guha \thanks{nilabja\_guha@uml.edu}}
\author[1]{Debkumar De}
\author[3]{Anindya Bhadra}
\author[4]{Veerabhadran Baladandayuthapani \thanks{veera@mdanderson.org}}
\author[1]{Bani K. Mallick \thanks{bmallick@stat.tamu.edu}}
\affil[1]{\normalsize{Department of Statistics, Texas A\&M University}}
\affil[2]{\normalsize{Department of Mathematical Sciences, University of Massachusetts Lowell}}
\affil[3]{\normalsize{Department of Statistics, Purdue University}}
\affil[4]{\normalsize{Department of Biostatistics, The University of Texas MD Anderson Cancer Center}}
\providecommand{\keywords}[1]{\textbf{\textsc{Keywords:}} #1}
\begin{document}

\maketitle

\begin{abstract}
\baselineskip=28pt
	Gaussian graphical models (GGMs) are well-established tools for probabilistic exploration of dependence structures using precision matrices. We develop a Bayesian method to incorporate covariate information in this GGMs setup in a nonlinear seemingly unrelated regression framework.  We propose a joint predictor and graph selection model and develop an efficient collapsed Gibbs sampler algorithm to search the joint model space. Furthermore, we investigate its theoretical variable selection properties. We demonstrate our method on a variety of simulated data, concluding with a real data set from the TCPA project. 
\end{abstract}

\baselineskip=16pt
\vfill\keywords{Bayesian variable selection; Decomposable graph; Gaussian graphical model; Hyper-inverse Wishart prior; Zellner's g-prior.}\vskip 20pt

\newpage
\baselineskip=28pt 
\section{Introduction}

Probabilistic graphical models provide a helpful tool to describe and visualize the  dependence  structures among random variables. Describing  the dependence structures, such as conditional dependence, graphical models can  provide insights into the model properties and interdependence between random variables. A graph comprises of vertices (nodes) connected by edges (links or arcs). In a probabilistic graphical model, the  random variables (single or vector) are represented by the vertices and  probabilistic relationships between these variables are expressed by the edges. An edge may or may not carry directional information. In this paper we concentrate on undirected Gaussian graphical models (GGMs) where the edges do not carry any directional information. Furthermore in this model, the variables follow a multivariate normal distribution with a particular structure on the inverse of the covariance matrix, called the precision or the concentration matrix.

More specifically, we will consider the Bayesian approaches to develop Gaussian graphical models. In a Bayesian setting, modeling  is based on hierarchical specifications for the covariance matrix  (or precision matrix) using  global priors on the space of positive-definite matrices, such as inverse Wishart prior or its equivalence like hyper-inverse Wishart (HIW) distribution (\cite{dawid1993hyper}).   In this paper, we only consider HIW distribution for decomposable graphs as this construction enjoys many advantages, such as computational efficiency due to its conjugate formulation and exact calculation of marginal likelihoods (\cite{scott2008feature}). Furthermore, this prior can impose sparsity on the graph (\cite{giudici1996learning}), which is perfect for high dimensional graphs. The use of HIW prior within Bayesian framework for GGMs has been well studied for the past decade, see \cite{giudici1999decomposable}, \cite{carvalho2007simulation}, \cite{carvalho2009objective}. Indeed, the HIW distributions for non-decomposable graphs exist and many studies have been done in that scenario (\cite{roverato2002hyper}, \cite{atay2005monte}, \cite{wang2010simulation}). Stochastic search for graph space is the crucial part of Bayesian computation in these models. For detailed description and comparison of various Bayesian computation methods in this scenario, see \cite{jones2005experiments}.

Most of the existing GGMs (as described above) are used to infer the conditional dependency structure of stochastic variables ignoring any covariate effects on the variables, when multiple sets of variables are assessed simultaneously. Hence, in this paper we consider that after conditioning on covariate effects $\mathrm{X}$ the responses $\mathrm{Y}$ is a GGM. An example of such a data structure include various types of genomic, epigenomic, transcriptomic and proteomic data have become available using array and sequencing based techniques. The variables in these biological systems contain enormous numbers of genetic markers have been collected at various levels such as mRNA, DNA, microRNA and protein expressions from a common set of samples. The interrelations within and among these markers provide key insights into the disease etiology. One of the crucial questions is to integrate these diverse data-types to obtain more informative and interpretable graphs representing the interdependencies between the variables. For example, in our particular case study used in this paper we consider protein and mRNA expression levels from the same patient samples have been collected extensively under The Cancer Genome Atlas (TCGA) project. As the protein expression levels are correlated due to presence of complex biological pathways and interactions, hence we are interested to develop conditional dependence model for them. However, in addition to other proteins, it is well-established that transcriptomic-level mRNA expressions modify the downstream proteomic expressions. This integrating the mRNA expressions as covariates or predictors in the model will produce more precise and refined estimates of the protein-level graphical structure.


From a modeling standpoint, to incorporate covariates in this graphical modeling framework, we adopt seemingly unrelated regression (SUR) (\cite{zellner1962efficient}, \cite{holmes2002accounting}) models where multiple predictors affect multiple responses, just as in the case of a multivariate regression, with the additional complication that the response variables exhibit an unspecified correlation structure. Similar SUR model has been proposed in \cite{wang2010sparse} which allows different responses to have different predictors. Our model can be seen as its special case which assumes that all responses have the same predictors. The model we propose has both theoretical and computational advantages over the SUR model in \cite{wang2010sparse}. In stead of using approximation for marginal likelihood, regression parameters and error covariance matrices can be marginalized. Combining with MCMC sampling technique, we have the exact posterior sampling other than approximation. Meanwhile, the closed form marginal likelihood enables us to explore the theoretical result for variable selection. By forcing all responses to have the same set of covariates, we strengthen signals in the likelihood functions which in return give us faster convergence time in searching of the true covariates. Second, our model is more suitable to the problem of interest which is to identify the influential gene expression based drivers for the entire network. Furthermore, we propose a joint sparse modeling approach for the responses (e.g. protein expressions) as well as covariates (e.g. mRNA expressions). This joint model simultaneously performs a Bayesian selection of significant covariates (\cite{george1993variable}, \cite{kuo1998variable}) as well as the significant entries in the adjacency matrix of the correlated responses (\cite{carvalho2009objective}).  In the frequentist setting, similar joint modeling has been recently attempted by \cite{yin2011sparse}, \cite{lee2012simultaneous} and \cite{cai2012covariate} for linear models.

To our best knowledge, the literature on Bayesian estimation of joint covariate-dependent graphical models is sparse, with the possible exception of \cite{bhadra2013joint}. Our proposed method differs from \cite{bhadra2013joint} in many aspects. In their paper, they used a linear model (for the covariate effects) with  independent priors on the regression coefficients.  On the other hand, we develop a nonlinear spline based model and propose a multivariate version of the well-known Zellner's $g$-prior (\cite{zellner1986assessing}) for the regression parameters. This is a natural extension of the original  $g$-prior from multiple linear regression models to have a matrix normal structure. In fact, it is  also a conjugate prior in this multivariate setup, hence drastically reduces the computational burden. Moreover, we investigate the Bayesian variable selection consistency of this multivariate regression model with graphical structures. Indeed, there are a few papers which have considered Bayesian variable selection consistency in a multiple linear regression framework with univariate responses (\cite{fernandez2001benchmark}, \cite{liang2008mixtures}). However, to our best knowledge, none of the existing papers investigated these results for multivariate regression with or without graphical structures.

To demonstrate our joint model for both variable and graph selection, we conduct a simulation study on a synthetic data set. The spline based regression captures the nonlinear structure precisely and the graph estimation has identified the underlying true graph. We also illustrate the necessity of incorporating the correct covariate structure by comparing the graph selection with respect to the null (no covariate) model and the linear regression model. As a result, we discover that the graph estimator is highly dependent on specifying the correct covariate structure. At the end, we analyze a data set from the TCPA project to identify the mRNA driver based protein networks. 

The rest of this article is organized as follows. We introduce our model and prior specification in the next section. In section 3, we present the stochastic search algorithm for the joint variable and graph selection. The variable selection consistency results have been presented  in section 4. Some simulation experiments are conducted in section 5. Finally, we apply our method on the TCPA data in section 6 and section 7 concludes.

\section{The Model}
\subsection{Undirected graph and hyper-inverse Wishart distribution}
An undirected graph G can be represented by the pair $(V, E)$, where $V$ is a finite set of vertices and the set of edges $E$ is a subset of the set $V \times V$ of unordered pairs of vertices. A graph is complete if all vertices are joint by edges, which means the graph is fully connected. Every node in a complete graph is a neighbor of every other such node. A complete vertex set $C$ in $G$ that is maximal (with respect to $\subseteq$) is a clique, maximally complete subgraph. That is, $C$ is complete and we cannot add a further node that shares an edge with each node in $C$. A subset $C \subseteq V$ is said to be an uv-separator if all paths from $u$ to $v$ intersect $C$. The subset $C \subseteq V$ is called a separator of $A$ and $B$ if it is an uv-separator for every $u \in A, v \in B$ (\cite{lauritzen1996graphical}, \cite{dobra2000bounds}). A graph is decomposable if and only if it can be split into a set of cliques and separators. By ordering them properly, it forms the junction tree representation which has the running intersection property. Let $J_{G}=(C_1,S_2,C_2,S_3, \dots, C_{k-1},S_k,C_k)$ be the junction tree representation of a decomposable graph $G$, where $C_i$ and $S_j$ are cliques and separators respectively, $i = 1, \dots, k, j=2, \dots,k$. Then $S_j = H_{j-1} \cap C_j$, where $H_{j-1} = C_1 \cup \ldots \cup C_{j-1}$, $j = 2,\dots, k$. Or in other words, $S_i$ is the intersection of $C_i$ with all the previous components $( C_1, C_2, \dots, C_{i-1} )$, so that $S_i$ separates the next component from the previous set. This running intersection property is quite useful when decomposing a graph into cliques and separators. For more details of decomposable graphs, see \cite{dobra2000bounds}. In this paper we only consider decomposable graphs.

Let $\mathrm{y}=(y_1, y_2, \dots, y_q)\sim N_q(0,\Sigma)$ and $\Sigma^{-1} = (\sigma^{ij})_{q\times q}$. The conditional dependencies lie in the precision matrix which is the inverse of the covariance matrix. Therefore, $y_i$ and $y_j$ are conditionally independent given the rest of the variables if and only if $\sigma^{ij}=0$, where $i\neq j$. This property induces a unique undirected graph corresponding to each multivariate Gaussian distribution. Thus, $q$ random variables represent $q$ nodes and if $G$ is the adjacency graph pairing to the precision matrix, then the presence of an off-diagonal edge between two nodes implies non-zero partial correlation (i.e., conditional dependence) and the absence of an edge implies conditional independence.

The inverse Wishart distribution which is a class of conjugate priors for positive definite matrices does not have the conditional independencies using to impose graphs. By imposing the conditional independencies on the inverse Wishart distribution, \cite{giudici1996learning} derived two classes of distributions -- ``local'' and ``global''. But only the ``local'' one induces sparse graphs. This is known as hyper-inverse Wishart distribution, proposed by \cite{dawid1993hyper}. It is the general set of conjugate priors for positive definite matrices which satisfies the hyper Markov law. Its definition is based on the junction tree representation. Let $J_{G}=(C_1,S_2,C_2,S_3, \dots, C_{k-1},S_k,C_k)$ be the junction tree representation of a decomposable graph $G$, then the HIW prior for the corresponding covariance matrix $\Sigma_G$ can be written as a ratio of products of cliques over products of separators (\cite{carvalho2009objective})
\begin{equation}
p(\Sigma | G) = \frac{\prod_{C\in\mathscr{C}} p(\Sigma_C | b, D_C)}{\prod_{S\in\mathscr{S}} p(\Sigma_S | b, D_S)},
\end{equation}
where $\mathscr{C}$ and $\mathscr{S}$ are the sets of all cliques and all separators respectively. For each clique $C$ (and separator $S$), $\Sigma_C \sim \mathrm{IW}(b, D_C)$ with density
\begin{equation}
p(\Sigma_C | b ,D_C) = \frac{ |D_C|^{\frac{b+|C|-1}{2}} }{ 2^{\frac{(b+|C|-1)|C|}{2}} \Gamma_{|C|}\big(\frac{b+|C|-1}{2}\big)} |\Sigma_C|^{-\frac{b+2|C|}{2}} exp\bigg\{  -\frac{1}{2} tr\big(\Sigma_C^{-1}D_C\big)   \bigg \},
\end{equation}
where $\Gamma_p(\cdot)$ is the multivariate gamma function.

For a given graph $G$, let $\mathrm{y}_i \sim N_q (0, \Sigma_G), i = 1, \dots, n$ and $\mathrm{Y} = (\mathrm{y}_1, \mathrm{y}_2, \dots, \mathrm{y}_n)^T$. If $\Sigma_G | G \sim \mathrm{HIW}_G(b, D)$, for some positive integer $b>3$ and positive definite matrix $D$, we have $\Sigma_G | \mathrm{Y}, G \sim \mathrm{HIW}_G (b+n, D + \mathrm{Y}^T\mathrm{Y})$. Therefore, the posterior of $\Sigma_G$ is still a HIW distribution. In the next section, we will incorporate covariate information in this model in a nonlinear regression framework.

\subsection{Covariate adjusted GGMs}
We consider the following covariate adjusted Gaussian distribution $\mathrm{y} \sim N_q( f(\mathrm{x}), \Sigma_G  )$, where $\mathrm{y} = (y_1, y_2, \dots, y_q)^T$, $\mathrm{x} = (x_1, x_2, \dots, x_p)^T$ and the function $f: \mathbb{R}^p \rightarrow \mathbb{R}^q$ performs a smooth, nonlinear mapping from the $p$-dimensional predictor space to the $q$-dimensional response space. $\Sigma_G$ is the covariance structure of $\mathrm{y}$ corresponding to the graph $G$. Linear model developed by \cite{bhadra2013joint} is a particular case of this where $f(\mathrm{x})=\mathrm{x}^T\bm{\beta}$. In the nonlinear setup, we choose to use spline to approximate the nonlinear function $f(\cdot)$. Without loss of generality, we assume all components of $\mathrm{x}$ share the same range, which means we can use the same knot points for all variables which simplifies the notations. And we also assume all covariates are centered so that the intercept terms are zero here. Given $k$ knot points, $\mathrm{w}= (w_1, w_2, \dots, w_k)^T$, the spline basis for $x_i$ is $\{ (x_i-w_1)_+, (x_i-w_2)_+, \dots, (x_i-w_q)_+  \}$. So $f(\mathrm{x})$ can be approximated by the linear form $\mathrm{uB}$, where $\mathrm{u}_{1 \times p(k+1)} = \{\mathrm{x}^T, (\mathrm{x}-w_1)_+^T, (\mathrm{x}-w_2)_+^T, \dots,  (\mathrm{x}-w_k)_+^T\}$ and $(\mathrm{x}-w_i)_+^T = \{(x_1-w_i)_+, (x_2-w_i)_+, \dots, (x_p-w_i)_+\}$ and $\mathrm{B}$ is the coefficient matrix, which has the structure below,
\begin{eqnarray*}
\mathrm{B}_{p(k+1) \times q} &=&
\begin{bmatrix}
\beta_{110} & \beta_{210} &\ldots &\beta_{q10}\\
 \vdots & \vdots &\ddots & \vdots\\
\beta_{1p0} & \beta_{2p0} &\ldots &\beta_{qp0}\\
\beta_{111} & \beta_{211} &\ldots &\beta_{q11}\\
 \vdots & \vdots &\ddots & \vdots\\
\beta_{1p1} & \beta_{2p1} &\ldots &\beta_{qp1}\\
\beta_{1pk} & \beta_{2pk} &\ldots &\beta_{qpk}\\
 \vdots & \vdots &\ddots & \vdots\\
\beta_{1pk} & \beta_{2pk} &\ldots &\beta_{qpk}\\
\end{bmatrix}.\\
\end{eqnarray*}
We assume the knot points $w$'s to be known and prespecified. That way, we have spline-adjusted model $\mathrm{y} \sim N_q( \mathrm{uB}, \Sigma_G )$ which has a linear model structure. Therefore, any variable selection method for linear regression can be used for the mean structure.

\subsection{The Bayesian Hierarchical Model}
Assuming we have a set of $n$ independent samples $\mathrm{Y}=(\mathrm{y}_1, \mathrm{y}_2, \dots, \mathrm{y}_n)^T$, where $\mathrm{y}_i \sim N_q(f(\mathrm{x}_i), \Sigma_G)$ and let $f(\mathrm{X}) = (f(\mathrm{x}_1), f(\mathrm{x}_2), \dots, f(\mathrm{x}_n))^T$. We have
\begin{equation}
\mathrm{Y} \sim \mathrm{MN}_{n\times q} (f(\mathrm{X}), I_n, \Sigma_G),
\end{equation}
where $\mathrm{MN}_{n\times q}(f(\mathrm{X}), I_n, \Sigma_G)$ is the matrix normal distribution with mean $f(\mathrm{X})$, and $I_n$ as the covariance matrix between $n$ rows and $\Sigma_G$ as the covariance matrix between $q$ columns. We approximate $f(\cdot)$ by $f(\mathrm{X}) = \mathrm{UB}$, where $\mathrm{U}$ is the spline basis matrix which has the structure below. And it is equivalent to write out the model as multivariate linear regression, $\mathrm{Y} = \mathrm{UB} + \mathrm{E}$, where $\mathrm{E} \sim \mathrm{MN}_{n\times q} (\mathrm{0}, I_n, \Sigma_G)$.
\begin{eqnarray*}
\mathrm{U}_{n\times p(k+1)} &=&
\begin{bmatrix}
\mathrm{x}_1 & \mathrm{x}_2 & \ldots &\mathrm{x}_n\\
(\mathrm{x}_1 -w_1)_+ & (\mathrm{x}_2 -w_1)_+ & \ldots & (\mathrm{x}_n-w_k)_+\\
\vdots & \vdots & \vdots & \ddots \\
(\mathrm{x}_1 -w_k)_+ & (\mathrm{x}_2 -w_k)_+ & \ldots & (\mathrm{x}_n-w_k)_+\\
\end{bmatrix}^T\\
\end{eqnarray*}


To introduce the notion of redundant variables for the variable selection in the mean structure, we define a binary vector $\bm\gamma = (\gamma_1, \ldots, \gamma_p)^T$, where $\gamma_i = 0$ if and only if $\beta_{jis} = 0$, for all $j = 1,\ldots,q, \ \ s = 0,1,\ldots,k$. By following this rule, the spline basis functions are related to each variable when performing the model selection. It means selecting one variable is equivalent to select all its related basis functions. Similarly, to introduce the notion of sparsity in the precision matrix, we define a binary variable $G_l$, where $l=1,\ldots,\frac{q(q-1)}{2}$, the $l$th off diagonal element in the adjacency matrix corresponding to the graph $G$. Diagonal elements of the adjacency matrix are restricted to one. The number of edges in the graph $G$ is denoted as $|E|=\sum_lG_l$. The Bayesian hierarchical model is given by
\begin{eqnarray}
(\mathrm{Y}-\mathrm{U}_{\bm\gamma} \mathrm{B}_{\bm{\gamma}, G})|  \mathrm{B}_{\bm{\gamma}, G}, \Sigma_G & \sim & \mathrm{MN}_{n \times q}(0,  I_n, \Sigma_G), \label{eq:Y}\\
\mathrm{B}_{\bm{\gamma}, G}| \bm\gamma, \Sigma_G & \sim & \mathrm{MN}_{p_{\bm\gamma}(k+1) \times q}\big(0,g (\mathrm{U}_{\bm\gamma}^T \mathrm{U}_{\bm\gamma})^{-1}_{p_{\bm\gamma}(k+1)}, \Sigma_G\big), \label{eq:B}\\
\Sigma_G|G & \sim & \mathrm{HIW}_G (b,d I_q),\label{eq:sigma}\\
\gamma_i &\overset{i.i.d.}{\sim}& \mathrm{Bernoulli}(\alpha_{\bm\gamma}) \text{  for } i=1,\ldots,p_{\bm\gamma},\label{eq:gamma}\\
G_l &\overset{i.i.d.}{\sim} & \mathrm{Bernoulli}(\alpha_{G}) \text{  for } l=1,\ldots,\frac{q(q-1)}{2},\label{eq:graph}\\
\alpha_{\bm\gamma} & \sim & \mathrm{U}(0,1),\label{eq:weight} \\
\alpha_G & = & 2/(q-1),\label{eq:gweight}
\end{eqnarray}
where $\mathrm{U}_{\bm\gamma}$ is the spline basis matrix with regressors corresponding to $\bm{\gamma}$ and $b>3$, $g$, $d$ are fixed positive hyper parameters. $\alpha_{\bm\gamma}$ is used to control the sparsity of variable selection and $\alpha_G$ is responsible for the complexity of graph selection. Also, denote $p_{\bm\gamma}=\sum_i\gamma_i$.


Equation (\ref{eq:B}) is the extended version of Zellner's g-prior (\cite{zellner1986assessing}) for multivariate regression. $g$-prior in this matrix normal form requires one more parameter than the usual multivariate normal form to allow the covariance structure between columns. Here, we  use $\Sigma_G$ as that parameter. There are a couple of reasons for this choice. First, it drastically decreases the complexity of marginalization. By using the same structure as the graph, it gives us the ability to integrate out the coefficient matrix $\mathrm{B}_{\bm{\gamma}, G}$. That way, we derive the marginal of $\mathrm{Y}$ explicitly. Moreover, it allows the variable selection and the graph selection to borrow strength from each other. Next, we derive the marginal density of data $\mathrm{Y}$ given only $\bm\gamma$ and graph $G$ in this modeling framework.

By using equation (\ref{eq:Y}) and (\ref{eq:B}), we have
\begin{equation*}
\mathrm{Y} | \bm{\gamma}, \Sigma_G \sim \mathrm{MN}_{n \times q} (0, I_n+g P_{\bm\gamma}, \Sigma_G),
\end{equation*}
where $P_{\bm\gamma} = \mathrm{U}_{\bm\gamma}(\mathrm{U}_{\bm\gamma}^T\mathrm{U}_{\bm\gamma})^{-1}\mathrm{U}_{\bm\gamma}^T$. In order to calculate the marginal of $\mathrm{Y}$, we need to vectorize $\mathrm{Y}$ as follows,
\begin{equation*}
vec(\mathrm{Y}^T) | \bm{\gamma}, \Sigma_G \sim N_{nq} (0, (I_n+g P_{\bm\gamma})\otimes \Sigma_G),
\end{equation*}
where $\otimes$ is the Kronecker product. Next, using the equation (\ref{eq:sigma}), we integrate out the $\Sigma_G$ to derive the marginal distribution of $\mathrm{Y}$.  The detailed calculation is in Appendix 1. Let $\mathscr{C}$ and $\mathscr{S}$ be the sets of all cliques and all separators for the given graph $G$ then 
\begin{equation}
f(\mathrm{Y}|\bm\gamma, G) = M_{n,G} \times {(g+1)}^{-\frac{p_{\bm\gamma}(k+1)q}{2}}
\frac
{\prod_{C\in\mathscr{C}} {|dI_C+S_C(\bm\gamma)|}^{-\frac{b+n+|C|-1}{2}}}
{\prod_{S\in\mathscr{S}} {|dI_S+S_S(\bm\gamma)|}^{-\frac{b+n+|S|-1}{2}}},
\end{equation}
where $S(\bm\gamma) = \mathrm{Y}^T(I_n-\frac{g}{g+1}P_{\bm\gamma})\mathrm{Y}$, $S_C(\bm\gamma)$ and $S_S(\bm\gamma)$ denote the quadratic forms restricted to the clique $C \in \mathscr{C}$ and the separator $S\in \mathscr{S}$. 

The normalizing constant $M_{n,G}$ has the following factorization which depends only  on $n$ and $G$, but it is the same for all $\bm\gamma$ under the same graph $G$. The advantage of this is when updating $\bm\gamma$ in the stochastic search, this term cancels out reducing the computational complexity. 
\begin{equation*}
M_{n,G}=
{(2\pi)}^{-\frac{nq}{2}}
\frac
{\prod_{C\in\mathscr{C}}
\frac{{|dI_C|}^{\frac{b+|C|-1}{2}}}
{2^{-\frac{n|C|}{2}}\Gamma_{|C|}\big(\frac{b+|C|-1}{2}\big)\Gamma_{|C|}^{-1}\big(\frac{b+n+|C|-1}{2}\big)}
}
{\prod_{S\in\mathscr{S}}
\frac{{|dI_S|}^{\frac{b+|S|-1}{2}}}
{2^{-\frac{n|S|}{2}}\Gamma_{|S|}\big(\frac{b+|S|-1}{2}\big)\Gamma_{|S|}^{-1}\big(\frac{b+n+|S|-1}{2}\big)}
}
\end{equation*}

\subsection{Prior specification for $\bm\gamma$ and $G$}
We use beta-binomial priors (\cite{george1993variable}) for both variable and graph selection. 
We control the sparsity by fixing $\alpha_G$. \cite{jones2005experiments} suggested to use $\frac{2}{|V|-1}$ as the hyper parameter for the Bernoulli distribution. For an undirected graph, it has peak around $|V|$ edges and it will be lower when applying to decomposable graphs. Additionally, we have other ways to control the number of edges which will be stated in the next section.

\section{The Stochastic Search Algorithm}
\subsection{Searching for $\gamma$}
From equation (\ref{eq:gamma}), we obtain the prior $p(\bm{\gamma}|\alpha_{\bm\gamma})=\prod_{i=1}^p p(\gamma_i|\alpha_{\bm\gamma})=\alpha_{\bm\gamma}^{p_{\bm\gamma}}(1-\alpha_{\bm\gamma})^{p-p_{\bm\gamma}}$.
Next, using equation (\ref{eq:weight}), we  integrate out $\alpha_{\bm\gamma}$, so that the marginal prior for $\bm{\gamma}$ is $p(\bm{\gamma}) \propto p_{\bm\gamma}!(p-p_{\bm\gamma})!$.
The searching for $\bm{\gamma}$ proceeds as follows:
\begin{itemize}
\item
Given $\bm\gamma$, propose $\bm\gamma^*$ by the following procedure. With equal probabilities, randomly choose one entry in $\bm\gamma$, say $\gamma_{s^*}$. If $\gamma_{s^*}=0$, then with probability $\delta$ change it to $1$ and with probability $1-\delta$ remain the same; if $\gamma_{s^*}=1$, then with probability $1-\delta$ change it to $0$ and with probability $\delta$ remain the same. Under this setting, $\delta$ is the probability of adding one variable when $\gamma_{s^*}=0$ and $1-\delta$ is the probability of deleting one variable when $\gamma_{s^*}=1$. If $\bm\gamma^*=\bm\gamma$, then $\frac{q(\bm\gamma|{\bm\gamma}^*)}{q({\bm\gamma}^*|\bm\gamma)}$=1. If one variable has been added to the model, $\frac{q(\bm\gamma|{\bm\gamma}^*)}{q({\bm\gamma}^*|\bm\gamma)}=\frac{1-\delta}{\delta}$; if one variable has been deleted from the model, $\frac{q(\bm\gamma|{\bm\gamma}^*)}{q({\bm\gamma}^*|\bm\gamma)}=\frac{\delta}{1-\delta}$.
\item
Calculate the marginal densities under both models $p(\mathrm{Y}|\bm{\gamma}, G)$ and $p(\mathrm{Y}|{\bm\gamma}^*, G)$.
\item
Accept ${\bm\gamma}^*$ with probability
\begin{equation*}
r(\bm{\gamma}, {\bm\gamma}^*) = min \bigg\{ 1, \frac{p(\mathrm{Y}|{\bm\gamma}^*, G)p({\bm\gamma}^*)q(\bm\gamma|{\bm\gamma}^*)}{p(\mathrm{Y}|\bm{\gamma}, G)p(\bm{\gamma})q({\bm\gamma}^*|\bm\gamma)} \bigg\}.
\end{equation*}
\end{itemize}
Notice, under the same graph the normalizing constant $M_{n,G}$ cancels out. Another thing is by using the parameter $\delta$, we can further control the sparsity of variable selection.

\subsection{Searching for $G$}
Similar to the calculation for $\bm{\gamma}$, the prior over the graph space is $p(G|\alpha_G)=\prod_{l=1}^{q(q-1)/2}p(G_l|\alpha_G)=\alpha_G^{|E|}(1-\alpha_G)^{q(q-1)/2-|E|}$, where $|E|$ is the total number of edges in the graph $G$ and $\alpha_G=2/(|V|-1)$. The searching for $G$ works as follows:
\begin{itemize}
\item
Given the current decomposable graph $G$, propose a new decomposable graph $G^*$ by the following procedure. With equal probabilities, randomly select an off-diagonal entry from the adjacency matrix of graph $G$, say $G_{s^*}$. If $G_{s^*}=0$, then with probability $\eta$ change it to $1$ and with probability $1-\eta$ remain the same; if $G_{s^*}=1$, with probability $1-\eta$ change it to $0$ and with probability $\eta$ remain the same. So the probability of adding an edge is $\eta$ when $G_{s^*}=0$ and $1-\eta$ is the probability of deleting an edge when $G_{s^*}=1$. We discard all proposed graphs which are non-decomposable. In those cases, the chain remains in the same graph for that iteration. If an edge has been added to the graph, $\frac{p(G|G^*)}{p(G^*|G)}=\frac{1-\eta}{\eta}$; if an edge has been removed from the graph, $\frac{p(G|G^*)}{p(G^*|G)}=\frac{\eta}{1-\eta}$.
\item
Calculate the marginal densities under both graphs $p(\mathrm{Y}|\bm{\gamma}, G)$ and $p(\mathrm{Y}|\bm{\gamma}, G^*)$.
\item
Accept $G^*$ with probability
\begin{equation*}
r(G, G^*) = min \bigg\{ 1, \frac{p(\mathrm{Y}|\bm{\gamma}, G^*)p(G^*)q(G|G^*)}{p(\mathrm{Y}|\bm{\gamma}, G)p(G)q(G^*|G)} \bigg\}.
\end{equation*}
\end{itemize}
This procedure is called add-delete Metropolis–Hastings sampler (\cite{jones2005experiments}). Another tool for sparsity is $\eta$. By choosing its value to be less than 0.5, it can reinforce sparsity on the graph.

\subsection{Conditional distributions of $\mathrm{B}_{\bm{\gamma}, G}$ and $\Sigma_G$}
We integrate out $\mathrm{B}_{\bm{\gamma}, G}$ and $\Sigma_G$ to make the stochastic search more efficient. But the conditional distributions of them both have simple closed forms. In \cite{bhadra2013joint}, the conditional distribution of $\Sigma_G$ depends on the coefficient matrix $\mathrm{B}_{\bm{\gamma}, G}$, but by using Zellner's $g$-prior it only requires $\bm\gamma$ and $G$,
\begin{equation*}
\Sigma_G | \mathrm{Y}, \bm{\gamma}, G  \sim  \mathrm{HIW}_G (b+n, dI_q+S(\bm{\gamma})).
\end{equation*}
At each iteration, given $\bm\gamma$ and $\Sigma_G$, using the following conditional distribution we can simulate $\mathrm{B}_{\bm{\gamma}, G}$,
\begin{equation*}
\mathrm{B}_{\bm{\gamma}, G} | \mathrm{Y}, \bm{\gamma}, \Sigma_G  \sim  \mathrm{MN}_{p_{\bm\gamma}(k+1) \times q} \bigg( \frac{g}{g+1}\big(\mathrm{U}_{\bm\gamma}^T\mathrm{U}_{\bm\gamma}\big)^{-1}\mathrm{U}_{\bm\gamma}^T\mathrm{Y}, \frac{g}{g+1}\big(\mathrm{U}_{\bm\gamma}^T\mathrm{U}_{\bm\gamma}\big)^{-1}, \Sigma_G \bigg).
\end{equation*}

\subsection{Choices of hyperparameters}
For choosing hyperparameters, we need to specify $g$, $b$, $d$. \cite{liang2008mixtures} summarized some choices for $g$ in the $g$-prior, like  $g=n$ (\cite{kass1995reference}),  $g=p^2$ (\cite{foster1994risk}), $g=max(n,p^2)$ (\cite{fernandez2001benchmark}) and other empirical Bayes methods to choose $g$. Based on simulations the choice of $g$ is not very critical in our approach. As long as $g$ satisfies the basic condition $g=O(n)$, there is no significant effect on the results. This condition is to keep the variances of the prior of coefficients not to be too small as $n$ goes to infinity. But when the dimension of the predictor space $p$ is large, one can consider to use $g=max(n, p^2)$.

The hyperparameter $b$ and $d$ are the two constants which control the hyper-inverse Wishart distribution. The common choice for the degree of freedom $b$ is 3 which provides a finite moment for the HIW prior (\cite{jones2005experiments}, \cite{carvalho2009objective}). Based on our experiments $d$ has a big impact on the graph selection results. Large $d$ results in more sparse graphs. On the other hand, large $d$ also contributes to large variances for coefficients. After standardizing the variances of responses to be $1$, \cite{jones2005experiments} suggested to use $1/(b+1)$ as a default choice of $d$, since the marginal prior mode for each variance term is $d(b+1)$. In our approach we are basically using the residuals after variable selection to fit the graphical model, hence it is impossible to know the variances. But we find $d = 1$ works well in our simulations.

The common choice for $\delta$ and $\eta$ in the stochastic search is $0.5$. Unless strong parsimony is required, we suggest to use this value. On the other hand, $\alpha_G$ can be set to $1/(|V|-1)$ or $0.5/(|V|-1)$ to achieve more sparsity on graph selection for noisy data.

\section{Variable Selection Consistency}

In this section, we first show the  Bayes Factor consistency of the variable selection method for a given graphical structure. To our best knowledge, there are no results on Bayesian variable selection consistency for this case. We first define the pairwise Bayes factor consistency for a given graph, subsequently under moderate conditions, we prove the pairwise Bayes factor consistency.  For some related development in multiple linear regression model with univariate response, see \cite{liang2008mixtures} and \cite{fernandez2001benchmark}. For simplicity, from now on we refer  the multivariate regression model as the regression model or just the model. Without further specification, the model we refer implies the regression model, not the graphical model.

Let binary vector $\bm{t}=(t_1,\dots,t_p)^T$ denote the regression model with respect to the true set of covariates of size $p_{\bm{t}}=\sum_{i=1}^{p}t_i$ and binary vector $\bm{a}=(a_1,\dots,a_p)^T$ denote an alternative regression model of size $p_{\bm{a}}=\sum_{i=1}^{p}a_i$. We use $\bm\gamma$ to represent any subset of the regression model space for being consistent with the notation in the early section. Next, we introduce the definition of pairwise Bayes factor consistency with graph structures.

\begin{definition}
{\normalfont{\textbf{(pairwise Bayes factor consistency under a given graph)}}}
Let $\mathrm{BF}(\bm{a}; \bm{t}|G)$ be the Bayes factor in favor of an alternative model $\bm{a}$ for a given graph $G$, such that $\mathrm{BF}(\bm{a}; \bm{t}|G) = \frac{P(\mathrm{Y}|\bm{a}, G)}{P(\mathrm{Y}|\bm{t}, G)}$. If $\mathrm{p}\lim_{n\rightarrow\infty} \mathrm{BF}(\bm{a}; \bm{t}|G)=0$, for any $\bm{a}\neq\bm{t}$, then we have pairwise Bayes factor consistency with respect to the true regression model $\bm{t}$ and the graph $G$.
\end{definition}

Here, ``$\mathrm{p}\lim_{n\rightarrow\infty}$" denotes convergence in probability and the probability measure is the sampling distribution under the true data generating model (\cite{liang2008mixtures}). Notice that the alternative model and the true model used in the Bayes factor calculation have the same graph $G$ which may not be the true underlying  graph. To clarify, the Bayes factor in the definition above is for a given graph $G$, where as the actual Bayes factor for the joint model is defined as $\mathrm{BF}(\bm{a};\bm{t})=\frac{P(\mathrm{Y}|\bm{a})}{P(\mathrm{Y}|\bm{t})}=\frac{\int P(\mathrm{Y}|\bm{a},G)\pi(G)dG}{\int P(\mathrm{Y}|\bm{t},G)\pi(G)dG}$, where $\pi(G)$ is the prior on the graph space. In this paper, the graph $G$ is restricted to the set of decomposable graphs and the number of nodes $q$ in the graph is finite. Before the main result, some regularization conditions need to be introduced.

\begin{condition} \label{cond1}
The set of graphs we consider is restricted to decomposable graphs with the number of nodes $q$ is finite. The number of knots $k$ for the spline basis is also finite.  
\end{condition}

\begin{condition}\label{cond2}
Let $\lambda_{min}\leq\dots\leq\lambda_i\leq\dots\leq\lambda_{max}$ be the eigenvalues of $\big({\mathrm{U}}^T_{\bm\gamma}{\mathrm{U}_{\bm\gamma}}\big)_{p_{\bm\gamma}(k+1)\times p_{\bm\gamma}(k+1)}$. Assume $0<c_\mathrm{U}<\frac{\lambda_{min}}{n} \leq \frac{\lambda_{max}}{n}<d_\mathrm{U}<\infty$, where $c_\mathrm{U}$ and $d_\mathrm{U}$ are two positive finite constants.
\end{condition}

\begin{condition}\label{cond3}
Let $\mathrm{E}_y = \mathrm{U}_{\bm{t}}\mathrm{B}_{\bm{t},G}$. Assume $\inf_{\bm{a}\neq \bm{t}} tr\{\mathrm{E}^T_y(I_n-P_{\bm{a}})\mathrm{E}_y\} > C_0n$, where $P_{\bm{a}}$ is the projection matrix of an alternative model $\bm{a}$ and $C_0$ is some fixed constant.
\end{condition}

\begin{condition}\label{cond4}
For the $g$-prior $\mathrm{B}_{\bm{\gamma}, G}| \bm\gamma, \Sigma_G  \sim  \mathrm{MN} \big(0,g\big({\mathrm{U}}^T_{\bm\gamma}{\mathrm{U}_{\bm\gamma}}\big)^{-1}, \Sigma_G\big)$ as in equation (\ref{eq:sigma}), assume $g = O(n)$.
\end{condition}

\begin{condition}\label{cond5}
Let $\hat{\Sigma}_{\bm{\gamma}, G}^{-1}$ be the MLE of $\Sigma_G^{-1}$ under any regression model $\bm\gamma$ and any decomposable graph $G$. Assume $\hat{\Sigma}_{\bm{\gamma}, G}^{-1}$ converges to some positive definite matrix ${\Sigma_{\bm{\gamma}, G}^0}^{-1}$ which has all eigenvalues bounded away from zero and infinity. Later, we drop the subscript $\bm{\gamma}$ and only use $\hat{\Sigma}_{G}^{-1}$ and ${\Sigma_{G}^0}^{-1}$.
\end{condition}

\begin{condition} \label{cond6}
The number of total covariates satisfies $\lim_{n\rightarrow\infty}\frac{p}{n}=0$, i.e. $p=o(n)$.
\end{condition}





Condition \ref{cond2} is needed to avoid singularity when the dimension of the model space $p$ increases to infinity as $n$ goes to infinity. Condition \ref{cond3} indicates that no true covariate can be fully explained by the rest of the covariates, which implies that  regressing any true covariate on all others, the coefficient of determination $R^2$ is less than 1. Condition \ref{cond4} makes sure we assign a non-degenerated prior on the coefficient matrix. Condition \ref{cond5} imposes restriction on the limit of $\hat {\Sigma}_G^{-1}$ or equivalently on the corresponding quadratic forms. It is needed for the given clique and separator decomposition of the hyper-inverse Wishart prior. For inverse Wishart prior, this condition can be relaxed if the corresponding graph is complete. We assume that  the MLE converges to a positive definite matrix ${\Sigma_G^0}^{-1}$. For the true graph this statement holds trivially. The explicit calculation of the MLE can be done by calculating the MLEs for each clique and separator, then combining them to form $\hat {\Sigma}_G^{-1}$. Given any model $\bm\gamma$, for $C\in \mathscr{C}$, define $\hat{\Sigma}_C^{-1}=\big\{ \frac{1}{n}\mathrm{Y}_C^T(I_n-\frac{g}{g+1}P_{\bm{\gamma}})\mathrm{Y}_C\big\}^{-1}$ and $\hat{\Sigma}_S^{-1}$ is defined similarly. The MLE is given by $\hat {\Sigma}_G^{-1}=\sum_{C\in\mathscr{C}} \hat{\Sigma}_C^{-1}\big|_0-\sum_{\mathscr{S}}\hat{\Sigma}_S^{-1}\big|_0$ where the suffix `0' implies that the elements corresponding to the vertices which are not in that subgraph are filled with zeros to create a $q\times q$ matrix (\cite{lauritzen1996graphical}).

\begin{lemma}\label{lm1}
Let $(\bm{t}, G)$ be the model with true covariates and any finite dimensional graph $G$. Assume $(\bm{a}, G)$ is any alternative model $(\bm{a}\neq\bm{t})$ with the same graph $G$. Then, for the model given by (\ref{eq:Y})-(\ref{eq:sigma}), under Condition \ref{cond1}-\ref{cond6}, $\mathrm{p}\lim_{n\rightarrow\infty}\mathrm{BF}(\bm{a};\bm{t}|G) = 0$ for \textbf{any} graph $G$ and any model $\bm{a}\neq\bm{t}$.
\end{lemma}
\begin{proof}
See Appendix 2 for details.
\end{proof}

That is to say, in the Lemma \ref{lm1}, we conclude that the pairwise Bayes factor for variable selection is consistent for any given graph, which is a quite strong result considering the magnitude of the graph space (here it is restricted to decomposable graph space). Next we show that with finite dimension graph, the result we have from Lemma \ref{lm1} is equivalent to the traditional Bayes factor for regression models.

\begin{theorem}\label{th1}
{\normalfont{\textbf{(pairwise Bayes factor consistency)}}}
For model given by (\ref{eq:Y})-(\ref{eq:gweight}), under Condition \ref{cond1}-\ref{cond6}, $\mathrm{p}\lim_{n\rightarrow\infty}\mathrm{BF}(\bm{a};\bm{t}) = 0$ for any model $\bm{a}\neq\bm{t}$.
\end{theorem}
\begin{proof}
Since the number of nodes $q$ in the graph is finite, then let $N_G(q)<\infty$ denote the number of all possible graphs. Therefore, for any alternative model $\bm{a}\neq\bm{t}$, we have
	\begin{align*}
	\mathrm{BF}(\bm{a};\bm{t}) 
	& =\frac{P(\mathrm{Y}|\bm{a})}{P(\mathrm{Y}|\bm{t})}
	  =\frac{\int P(\mathrm{Y}|\bm{a},G)\pi(G)dG}
	        {\int P(\mathrm{Y}|\bm{t},G)\pi(G)dG} \\
	& =\frac{\sum_{i=1}^{N_G(q)} P(\mathrm{Y}|\bm{a},G_i)\pi(G_i)}
			{\sum_{i=1}^{N_G(q)} P(\mathrm{Y}|\bm{t},G_i)\pi(G_i)} \\
	& \leq \frac{N_G(q)P(\mathrm{Y}|\bm{a},G_M)\pi(G_M)}
				{P(\mathrm{Y}|\bm{t},G_M)\pi(G_M)}
	 = N_G(q)\frac{P(\mathrm{Y}|\bm{a},G_M)}{P(\mathrm{Y}|\bm{t},G_M)}\rightarrow 0,	
	\end{align*}
where $P(\mathrm{Y}|\bm{a},G_M) = max\{P(\mathrm{Y}|\bm{a},G_i), i=1,\dots, N_G(q)\}$ and $\pi(G_M) = max\{\pi(G_i),i=1,\dots, N_G(q)\}$. Then following the result in Lemma \ref{lm1}, we have a pairwise Bayes factor consistency for regression models.
\end{proof}

Notice that we do not need the number of covariates $p$ to be finite here. As long as the conditions are satisfied, the result of Lemma \ref{lm1} and Theorem \ref{th1} hold accordingly. Next, we discuss the variable selection consistency. For finite dimensional covariate space, the variable selection consistency is an immediate result.

\begin{corollary}\label{col1}
For model given by (\ref{eq:Y})-(\ref{eq:gweight}), under Condition \ref{cond1}-\ref{cond5}, if the  number of covariates $p$ is finite, then $\mathrm{p}\lim_{n\rightarrow\infty} P(\bm{t}|\mathrm{Y}) = 1$.
\end{corollary}
\begin{proof}
	\begin{equation*}
	P(\bm{t}|\mathrm{Y})
	= \frac{\pi(\bm{t})P(\mathrm{Y}|\bm{t})}{\sum_{\bm{a}} \pi(\bm{a})P(\mathrm{Y}|\bm{a})}
	= \bigg( \sum_{\bm{a}} \frac{\pi(\bm{a})p(\mathrm{Y}|\bm{a})}{\pi(\bm{t})p(\mathrm{Y}|\bm{t})} \bigg)^{-1}
	= \bigg(1+ \sum_{\bm{a}\neq\bm{t}} \frac{\pi(\bm{a})}{\pi(\bm{t})} \mathrm{BF}({\bm{a};\bm{t}}) \bigg)^{-1} \rightarrow 1,
\end{equation*}
where $\pi(\cdot)$ is the prior on the regression model. Notice, the last summation has finite number of terms, since the covariate space is finite. Then the rest follows directly form Theorem \ref{th1}.
\end{proof}

Therefore, the variable selection is consistent when the model space is finite. Corollary \ref{col1} does not depend on the graph, which means it holds even if we do not identify the true graph.

\section{Simulation Study}
In this section, we present the simulation study for our method considering the hierarchical model from Section 2. In order to justify the necessity for a covariate adjusted mean structure, we compare the graph estimation among spline regression, linear regression and no covariates (graph only) assuming an underlying nonlinear covariate structure as the true model.

Considering the hierarchical model in Section 2, we choose $p=30$, $q=40$, $n=700$. All predictors $x_{ij}$, $i=1, \dots, n$ and $j=1, \dots, p$ are simulated from uniform distribution $(-1, 1)$. For the nonlinear regression structure, we use a relatively simple and smooth function, similar to the function used in \cite{denison1998automatic},
\begin{equation}\label{sim}
f_i(\mathrm{x}_j) = h_{i1}\sin(x_{j5}) + h_{i2}\sin(x_{j11}) + h_{i3} x_{j17} + h_{i4}e^{x_{j24}}, i=1,2,\dots, q,
\end{equation}
where $\mathrm{x}_j = (x_{j1}, x_{j2}, \dots, x_{jp})$ is the $j$th sample of data $\mathrm{X} = (\mathrm{x}^T_1, \mathrm{x}^T_2, \dots, \mathrm{x}^T_q)^T$. The true set of predictors is $\{5, 11, 17, 24\}$. The coefficients in (\ref{sim}) are simulated from exponential distribution $\exp(1)$. Figure \ref{fig2by2}(a) shows the true adjacency matrix for the true graph $G$. The true covariance matrix $\Sigma_G$ is generated from $\mathrm{HIW}_G(3, I_q)$. And the columns of error matrix $\mathrm{E}$ are $n$ random draws from multivariate normal distribution $N_q(0, \Sigma_G)$. Thus $\mathrm{Y} = F(\mathrm{X}) + \mathrm{E}$, where $F(\mathrm{X}) =\big(f_j(\mathrm{x}_i)\big)_{ij}$. For hyperparameters, we use $g=n=700$, $b=3$, $d=1$ and $\delta=\eta=1/2$ in the stochastic search of graphs.

For spline basis functions, we use 10 fixed knot points which divide $(-1,1)$ evenly into 11 intervals. 100,000 MCMC iterations are performed after 10,000 burn-in steps. We use a similar true graph as in \cite{bhadra2013joint} in the simulation. The results are quite fascinating. For variable selection, after burn-in iterations, it quickly converges to the true set of predictors. Furthermore for the variable selection, if we only use the linear regression model to estimate the nonlinear structure, it misses some important predictors. As shown in the simulation study, the exponential term has not been identified without the spline regression. Although the linear case selects most of the correct variables, estimates of mean functions are completely wrong, which misleads the graph estimation completely as we show next.

\begin{figure}[h]
	\centering
	\subfloat[][]{\includegraphics[width=.4\textwidth]{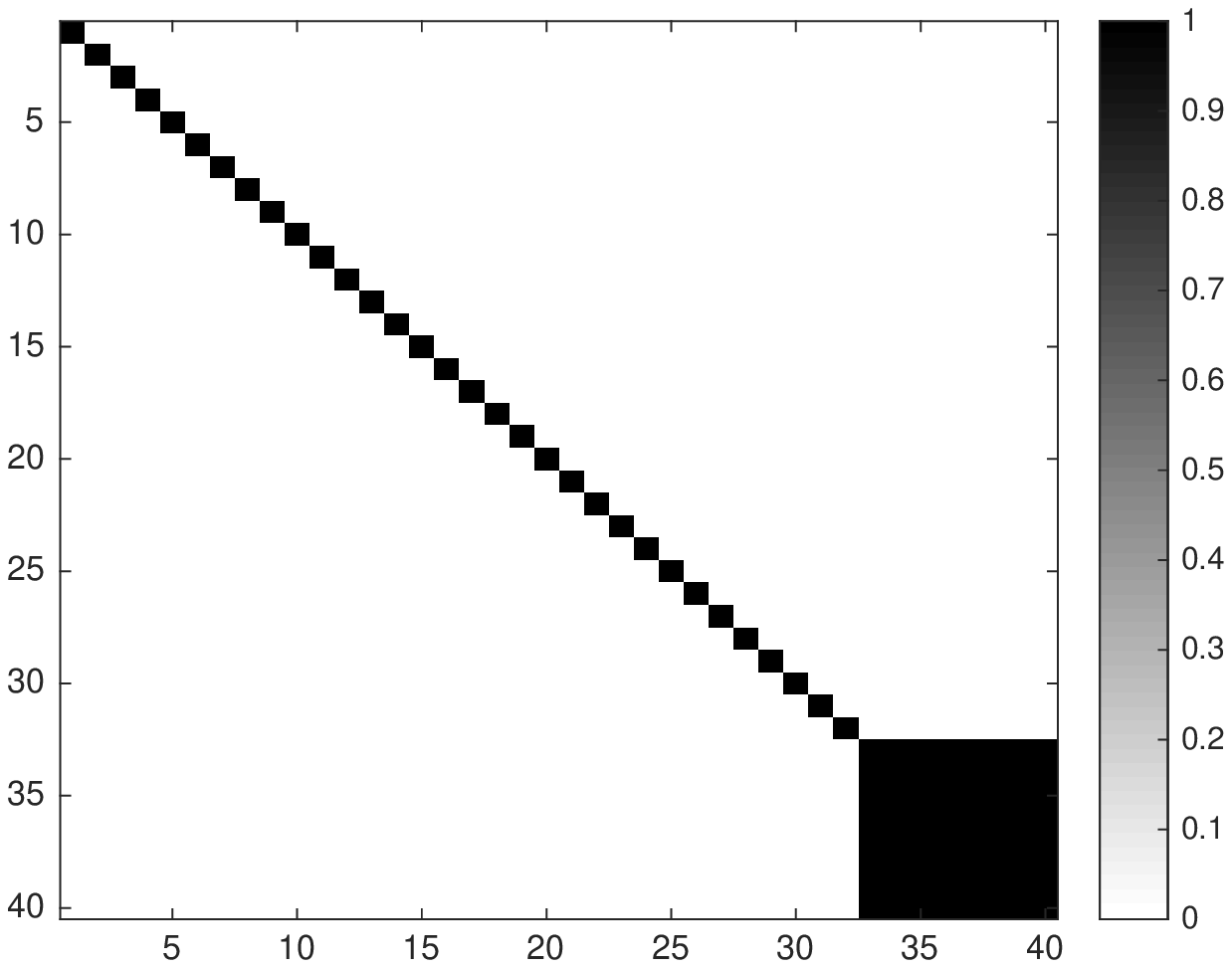}}\quad
	\subfloat[][]{\includegraphics[width=.4\textwidth]{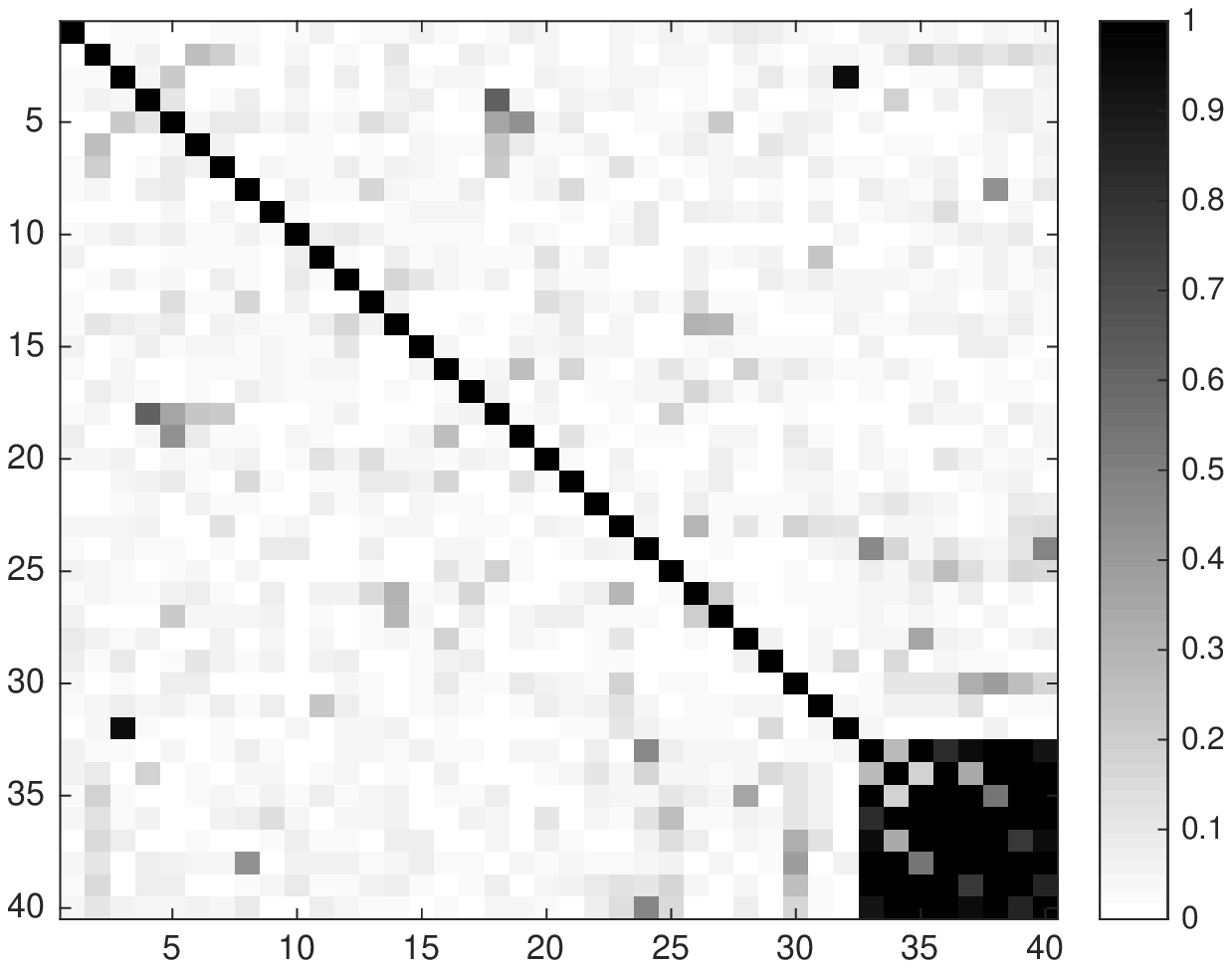}}\\
	\subfloat[][]{\includegraphics[width=.4\textwidth]{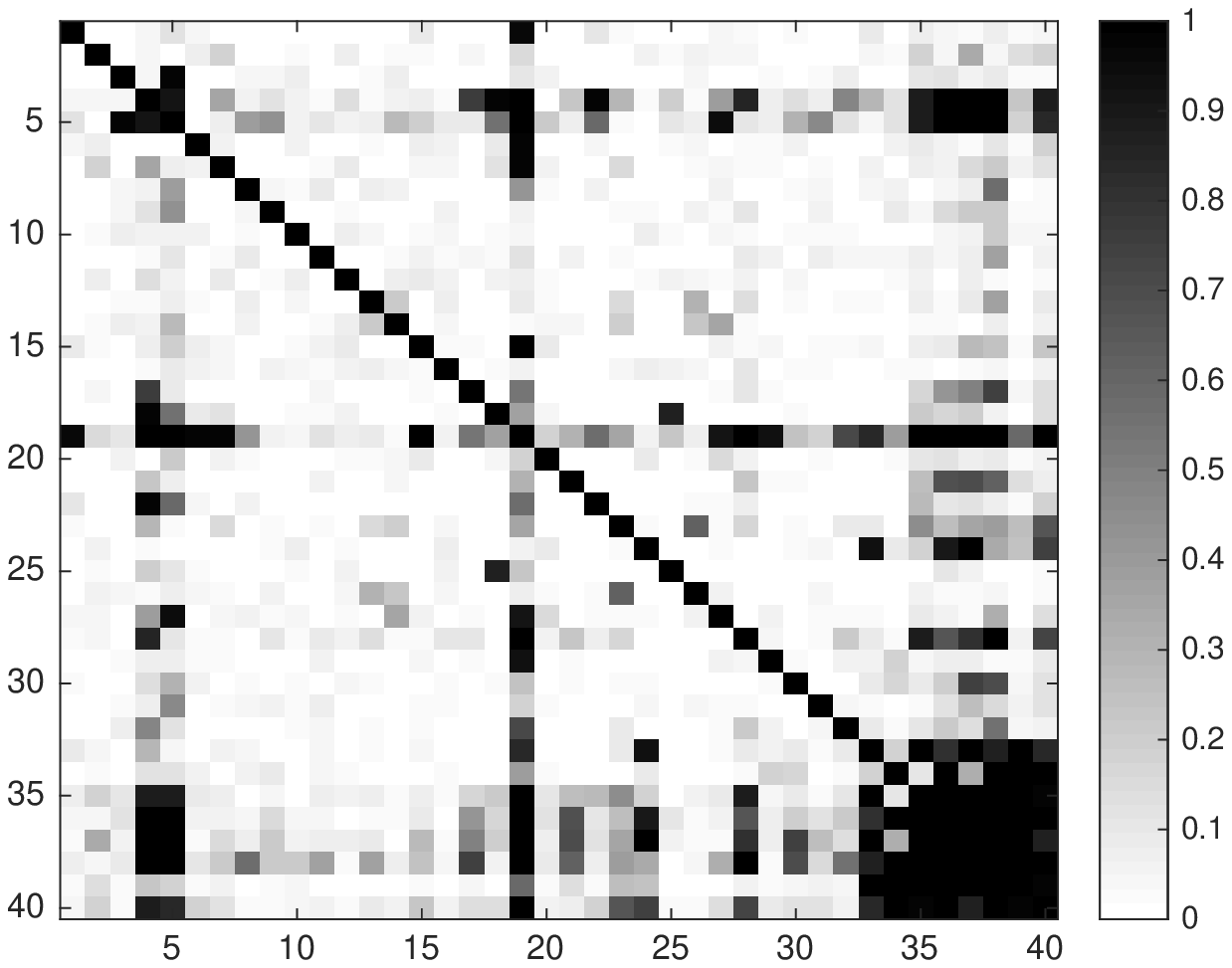}}\quad
	\subfloat[][]{\includegraphics[width=.4\textwidth]{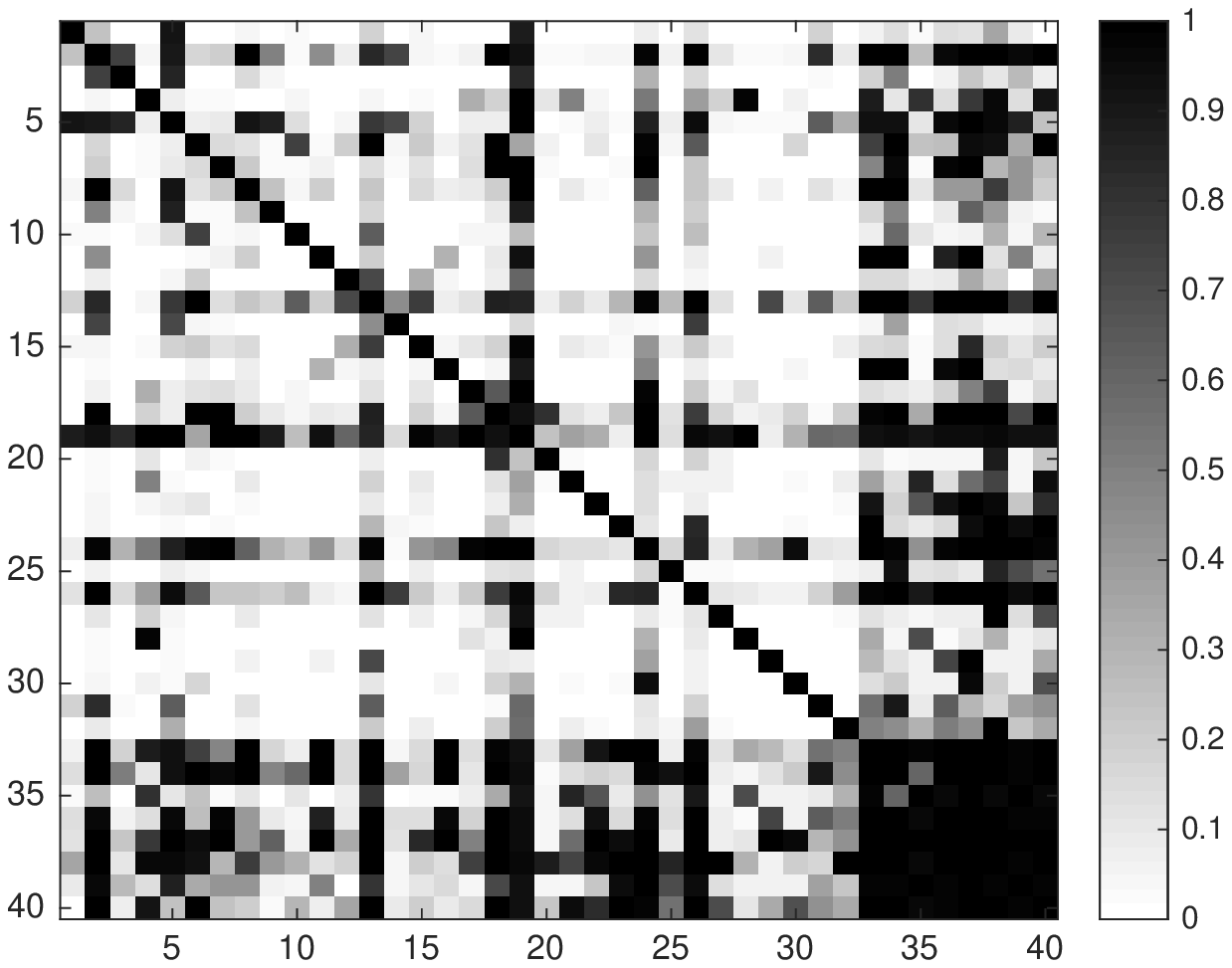}}
	\caption{(a) The adjacency matrix of the true graph $G$. In the adjacency matrix, 1 indicates an edge and 0 indicates no edge. So edges are represented by black bricks. The diagonal entries are 1 by default. (b), (c) and (d) are the marginal posterior probabilities of each edge in the estimated graph on a gray scale under spline regression, linear regression and without covariates, respectively.}
	\label{fig2by2}
\end{figure}

Here, we use marginal posterior probability for each edge to choose our final estimation of the graph. Marginal probabilities are calculated by averaging all adjacency matrices in the MCMC chain. The cut-off point is 0.5, which means we only select the edge with posterior probability more than half. The cut-off point can also be varied to accomplish different degrees of sparsity for the estimated graph.

Figure \ref{fig2by2}(b) shows when using spline regression to capture the nonlinear mean structure, the major parts of the true graph can be recovered. On the other hand, from Figure \ref{fig2by2}(c), we can see that the linear regression model fails to estimate the residual terms properly, hence the estimated graph is completely wrong. It may still capture a few true edges, however a large number of false edges have been added. Thus, modeling the true mean structure is essential for estimating the graph. That way, specification of an incorrect mean structure (e.g. using linear function to estimate nonlinear function, choosing a wrong set of covariates) always leads to a wrong graph estimation. This can also be illustrated by ignoring the covariates to estimate the graph. Figure \ref{fig2by2}(d) shows in this scenario the estimated graph is again completely wrong. We plot the receiver operating characteristic (ROC) curves for the above three cases in Figure 2. As we can see the joint models (i.e. spline and linear regression model) perform  better than no covariates. Furthermore, the ROC curve of spline regression model is nearly perfect.   

\begin{figure}[H] \label{ROC}
	\centering
	\includegraphics[width=0.5\textwidth]{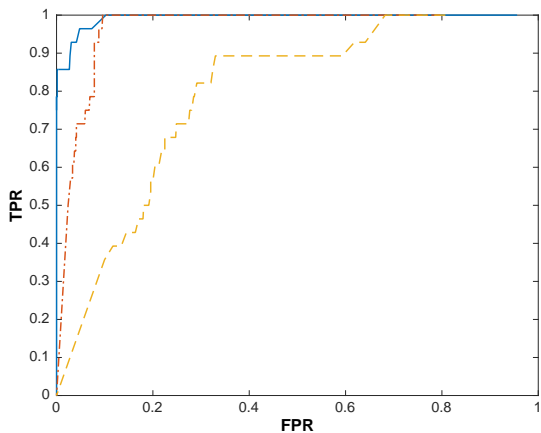}
	\caption{Plot of ROC curves for graph selection. The blue, red and yellow lines are spline regression model, linear regression model and no covariates (graph only) model, respectively.}
\end{figure}

\section{Protein-mRNA data}

In this section, we apply our method to a protein-mRNA data set from The Cancer Proteome Atlas (TCPA) project. The major goals of this analysis are (i)  to identify the influential gene expression based drivers, i.e. mRNAs which  influence the protein activities and (ii) to estimate the protein network, simultaneously. The central dogma to our model is the well-known fact that mRNA which is the messenger of DNA from transcription plays a critical role in proteomics by a process called translation. Consequently, the protein expressions play a crucial role for the development of tumor cells. Therefore, to identify which mRNAs dominate this process is the key component in this oncology study. This also motivates us to regress the protein level data on the mRNA based data. Multivariate regression is a powerful tool for combining information across all regressions. One can use an univariate regression model on each protein. However, there are multiple advantages of our model to apply in this scenario. First, it combines the information across all responses, i.e. protein expressions. As we know, proteins tend to work together as a network. Performing single regression separately on each of them will lose information which lies in this protein network. Second, by jointly modeling all proteins, we obtain a graph of all proteins after correcting the mRNA effects from the mean structure. In our analysis, we choose the breast cancer data which has the largest sample size of 844 among other tumors.  Based on different DNA functions, proteins are categorized into 12 pathways with their corresponding mRNAs. For more details about these pathways, see \cite{akbani2014pan} and \cite{ha2016personalized}. We apply our model for each pathway since proteins from the same pathway exhibit similar behaviors.  We use spline based regression to further investigate the different nonlinear relationships among proteins and mRNAs. The results of the covariate selection and the graph estimation are summarized below.

For the standardized design matrix, we use ten evenly distributed knot points for spline basis to capture the nonlinearities between proteins and mRNAs. Since the observations of mRNAs are not uniformly distributed across their ranges, we use the penalized spline regression proposed by \cite{crainiceanu2005bayesian} to solve the rank deficiency problem in the $g$-prior. The selection results along with the number of proteins and mRNAs used in each pathway are summarized in the table below. Four of those pathways don't have any influential mRNA which controls the proteins. The rest seven pathways all have one or more related mRNAs according to the results. For example, the pathway of Apoptosis is about programmed cell death. The model selects only the mRNA corresponding to gene BCL2. BCL2 is an anti-cell death gene (\cite{tsujimoto1998role}). Proteins in BCL2 family play an important role in control of apoptosis. Studies have found that they constitute a life or death decision point for cells in a common pathway involved in various forms of Apoptosis (\cite{tsujimoto1998role}). In this sense, our model identifies the correct mRNA that dominates this Apoptosis process. BCL2 also contributes to the pathway about hormone receptor and signaling. CCNE1 has been selected in the pathway of cell cycle. It has been found that there is an association between CCNE1 amplification and breast cancer treatment (\cite{etemadmoghadam2013synthetic}). CCNE1 is also related to the endometrioid endometrial carcinomas (\cite{nakayama2016ccne1}, \cite{cocco2016dual}). CDH1 has been selected in the pathway of core reactive and EMT. Mutations in CDH1 have been observed to be associated with increased susceptibility to develop lobular breast cancer (\cite{schrader2008hereditary}, \cite{polyak2009transitions}). From the selection results, INPP4B is related to PI3K/AKT and hormone receptor and signaling pathways. Interestingly, there are emerging evidences that INPP4B identified as a tumor suppressor regulates PI3K/AKT signaling in breast cancer cell lines (\cite{bertucci2013phosphoinositide}). For more information about the relationship between INPP4B and PI3K/ATK pathway, see \cite{mcauliffe2010deciphering}, \cite{miller2011mutations} and \cite{gasser2014sgk3}. Both ERBB2 and GATA3 genes have strong influence in breast cancer, see \cite{harari2000molecular}, \cite{revillion1998erbb2}, \cite{kallioniemi1992erbb2}, \cite{dydensborg2009gata3} and \cite{yan2010gata3}. For a full list of gene names related to each of 12 pathways, see supplementary table S2 in \cite{ha2016personalized}. For the plots of estimated nonlinear functions of each node in all seven pathways, see Appendix 3.







\begin{table}[H] \label{tab1}
	\caption{12 pathways and mRNA selection results} 
	\centering 
	\scalebox{0.9}{
	\begin{tabular}{c c c c c} 
		\hline\hline 
		\# & pathway names & \# of proteins & \# of mRNAs & mRNA selected\\ [0.5ex] 
		\hline 
		1  & Apoptosis           &  9 &  9  &  BCL2  \\ 
		2  & Breast reactive     &  6 &  7  &   -    \\
		3  & Cell cycle          &  8 &  7  & CCNE1  \\
		4  & Core reactive       &  5 &  5  &  CDH1  \\
		5  & DNA damage response & 11 & 11  &   -    \\
		6  & EMT                 &  7 &  7  &  CDH1  \\
		7  & PI3K/AKT            & 10 & 10  &INPP4B  \\
		8  & RAS/MAPK            &  9 & 10  &   -    \\
		9  & RTK                 &  7 &  5  &  ERBB2 \\
		10 & TSC/mTOR            &  8 &  5  &   -    \\
		11\&12 & Hormone receptor\&signaling & 7 & 4 & INPP4B, GATA3, BCL2 \\ [1ex] 
		\hline 
	\end{tabular}
	}
	\label{table:nonlin} 
\end{table}

The protein networks for all 12 pathways are shown in Figure 3. The number on each edge is the estimated partial correlation between two proteins it connects. Green edge means positively partial correlated; red edge means negatively partial correlated. The thickness of the edge represents the magnitude of the absolute value of partial correlation. Within each network, majority of the proteins tends to be positively correlated, which means most of the proteins are working together within each pathway. Proteins which are related to the same gene family have high positive correlation in the graph. For example, AKTPS and AKTPT (AKT gene family) in PI3K/AKT pathway, GSK3A and GSK3P (GSK gene family) in PI3K/AKT pathway, and S6PS24 and S6PS23 (RPS6 gene family) in TSC/mTOR pathway. We define the degree of freedom for nodes as the number of edges connected to them. Then hub nodes are the nodes which have the largest degree of freedom in each pathway. These are the proteins which have the maximum connectivities and interact heavily with other proteins. The summary of hub nodes are shown in the table below (the number in the bracket is the degree of freedom of the hub node).

\begin{table}[H] \label{tab2}
	\caption{hub nodes in each pathway}
	\centering 
	\scalebox{0.9}{
	\begin{tabular}{c c c c} 
		\hline\hline 
		pathway & hub nodes & pathway & hub nodes\\ [0.5ex]
		\hline
		Apoptosis           & BAX(7)             & Breast reactive & RAB(5), RBM(5) \\
		Cell cycle          & CYCLINE1(7)        & Core reactive & CAV(4), ECA(4)   \\
		DNA damage response & XRC(9)             & EMT & ECA(6), COL(6)             \\
		PI3K/AKT            & ATKPT(9), GSK3P(9) & RAS/MAPK & CJU(8)                \\
		RTK                 & EGFRPY10(6), HER3(6), SHC(6) & TSC/mTOR & S6PS23(6)   \\ 
		Hormone receptor\&signaling  & INPP4B(6) \\ [1ex]		
		\hline
	\end{tabular}
	}
\end{table}

\newcommand*{\addheight}[2][.5ex]{%
  	\raisebox{0pt}[\dimexpr\height+(#1)\relax]{#2}%
}

\section{Further Discussion}
Our model provides a framework for jointly learning about the nonlinear covariate structure as well as the dependence graph structures for the Gaussian graphical model. We used fixed knot splines for estimation of nonlinear functions. This can be extended for adaptive splines or other adaptive basis functions (\cite{denison1998automatic}). We have introduced our model for decomposable graphs but it can be extended for more general settings. Apart from genomic applications, there are numerous problems that arise in finance, econometrics, and biological sciences where nonlinear graph models can be a useful approach to modeling and therefore, we expect our inference procedure to be effective in those applications. Extension from Gaussian to non-Gaussian models is an interesting topic for future research.

\section*{Acknowledgements}
Nilabja Guha is supported by the NSF award \#2015460.



\begin{table}[H] \label{tab3}
	\centering 
	\scalebox{0.9}{
	\begin{tabular}{c c c}
		\addheight{\includegraphics[width=55mm]{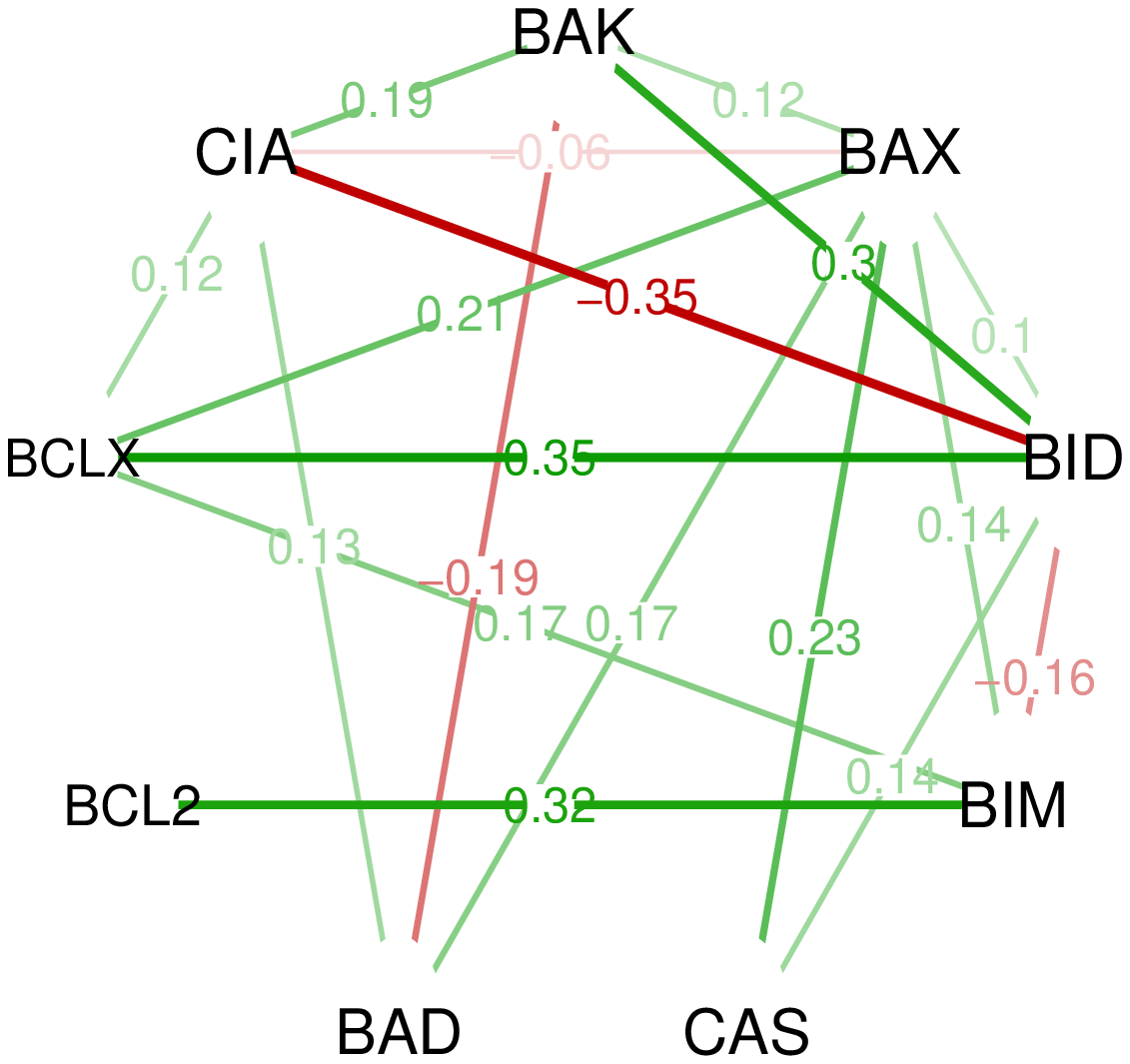}} &
		\addheight{\includegraphics[width=55mm]{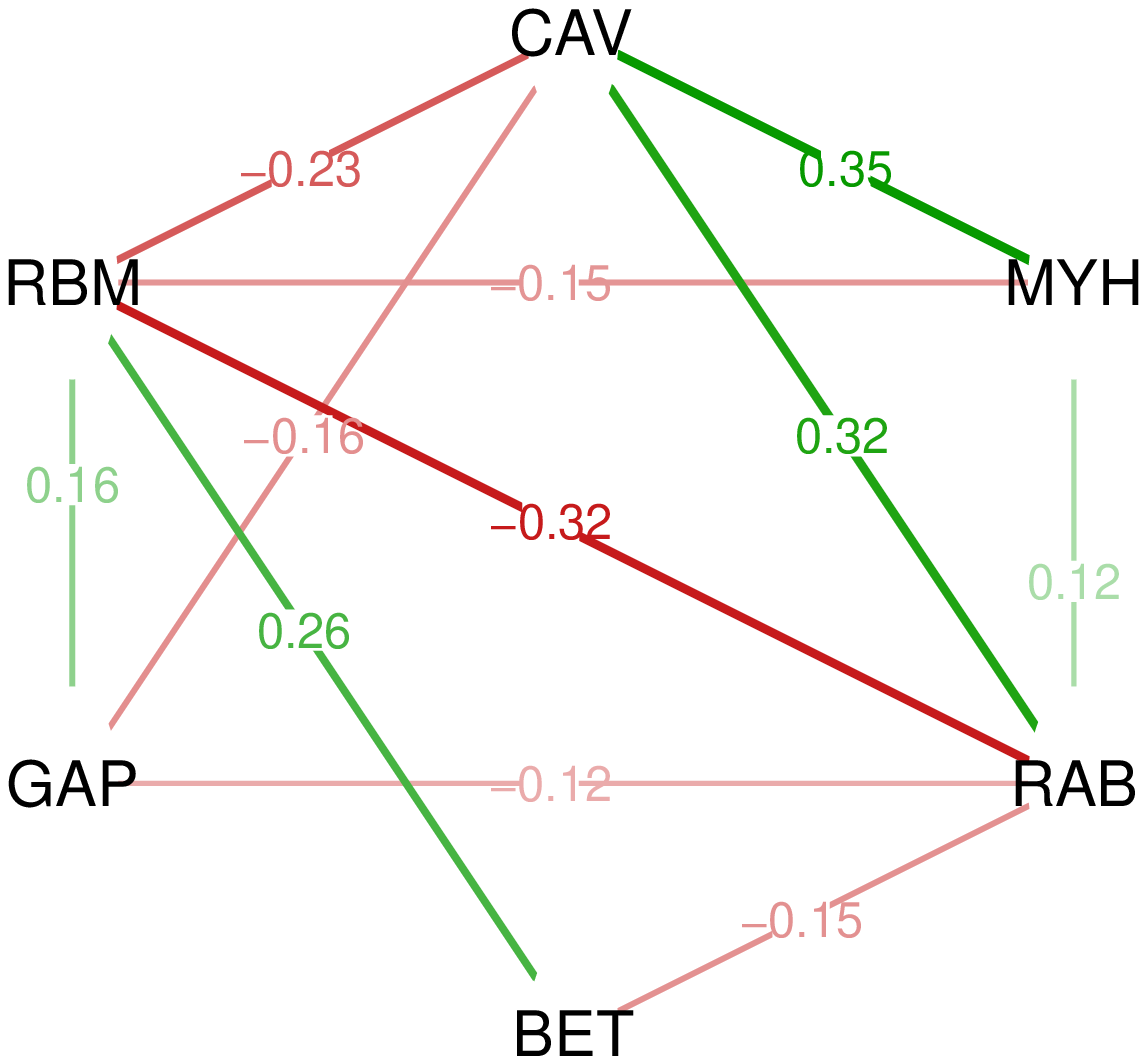}} &
		\addheight{\includegraphics[width=55mm]{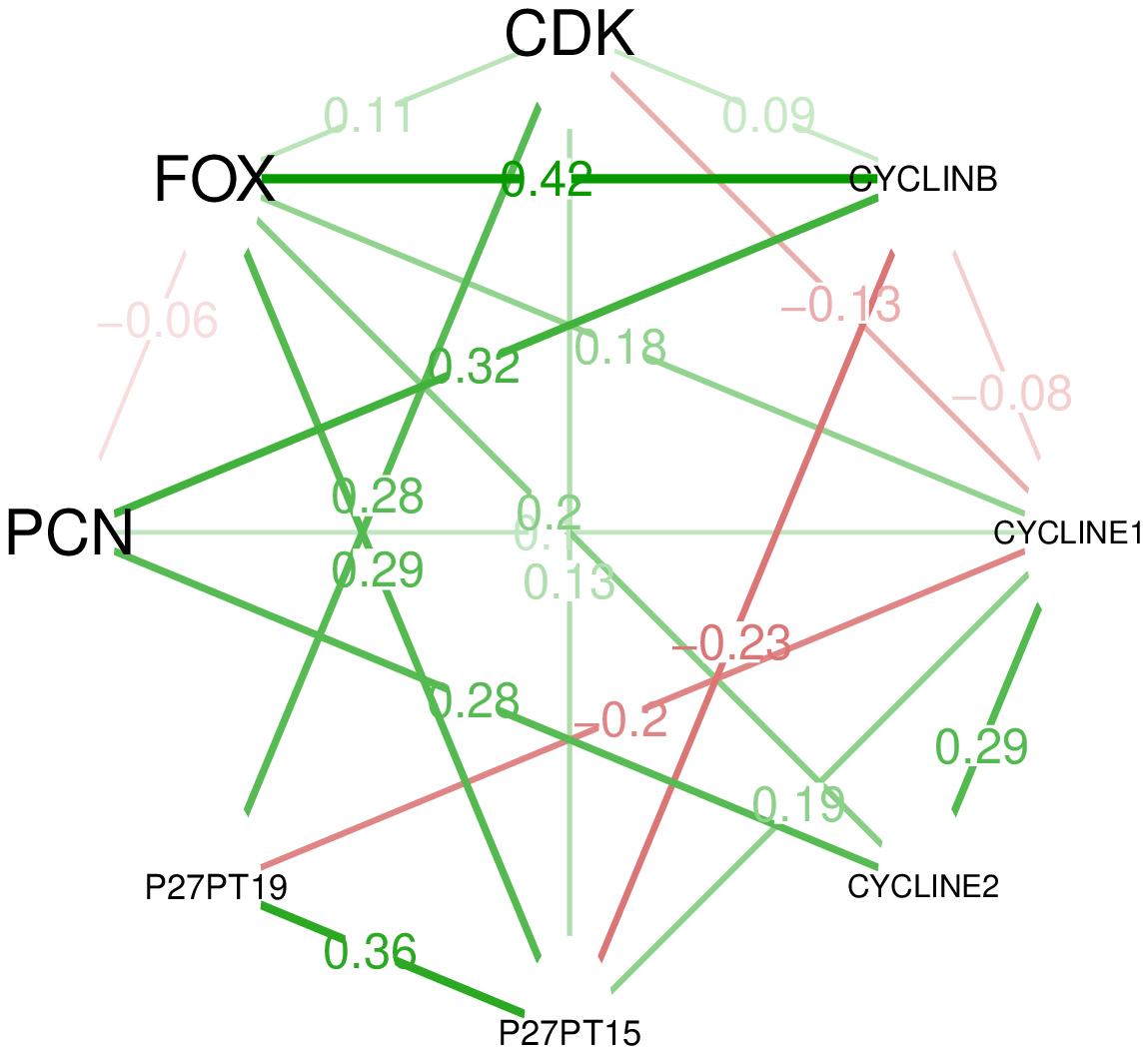}} \\
		\small Apoptosis &  Breast reactive  & Cell cycle \\
		\addheight{\includegraphics[width=55mm]{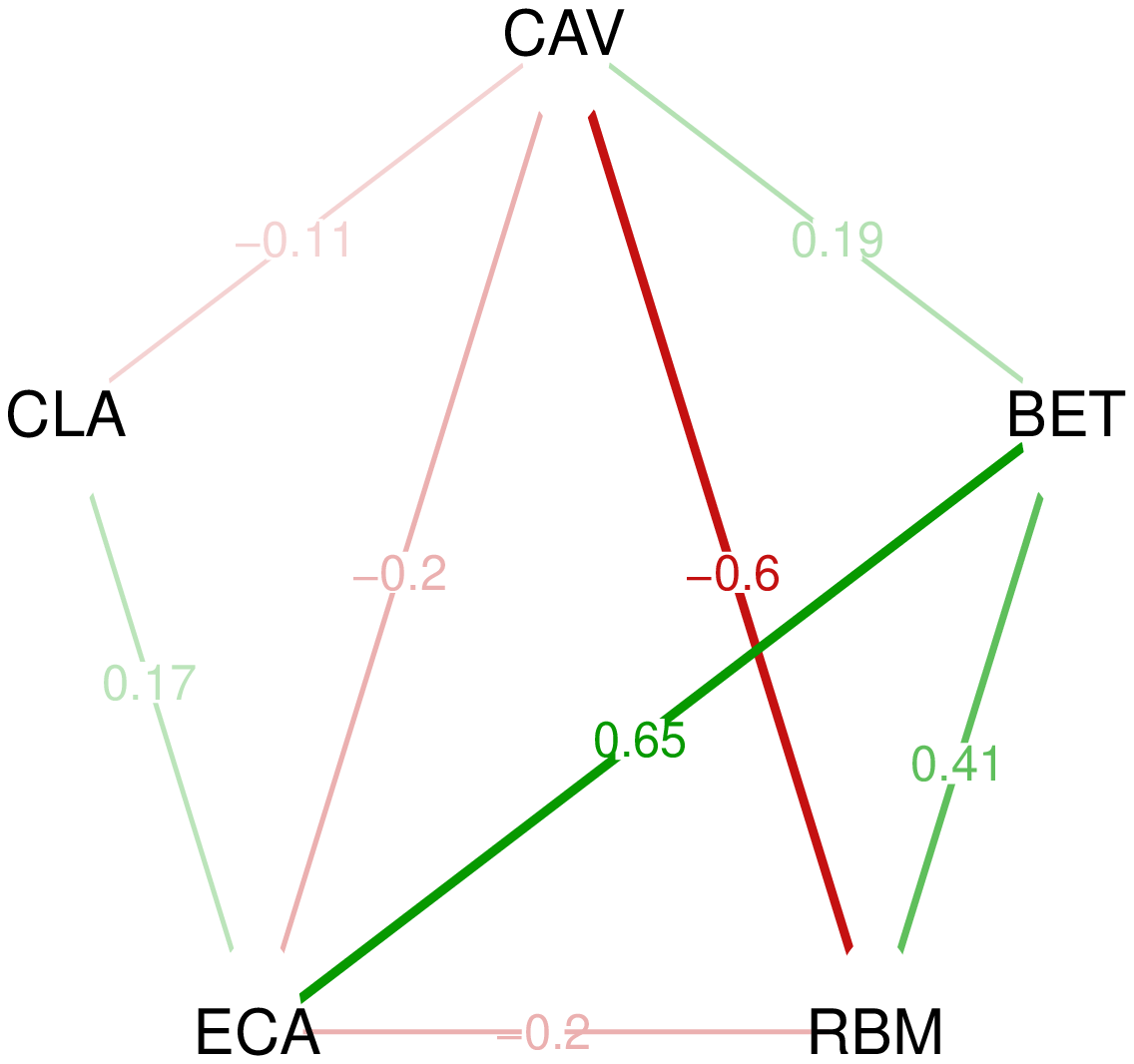}} &
		\addheight{\includegraphics[width=55mm]{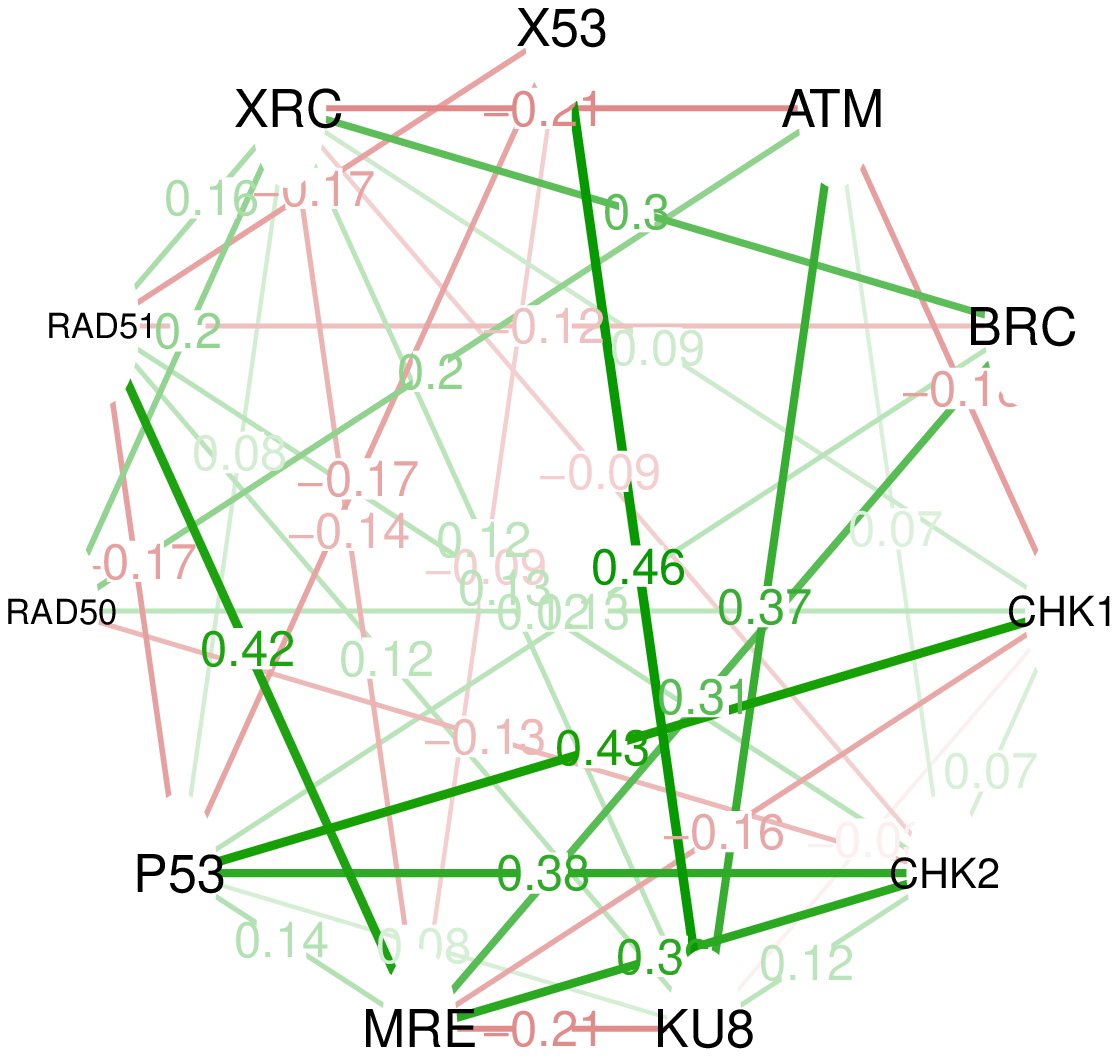}} &
		\addheight{\includegraphics[width=55mm]{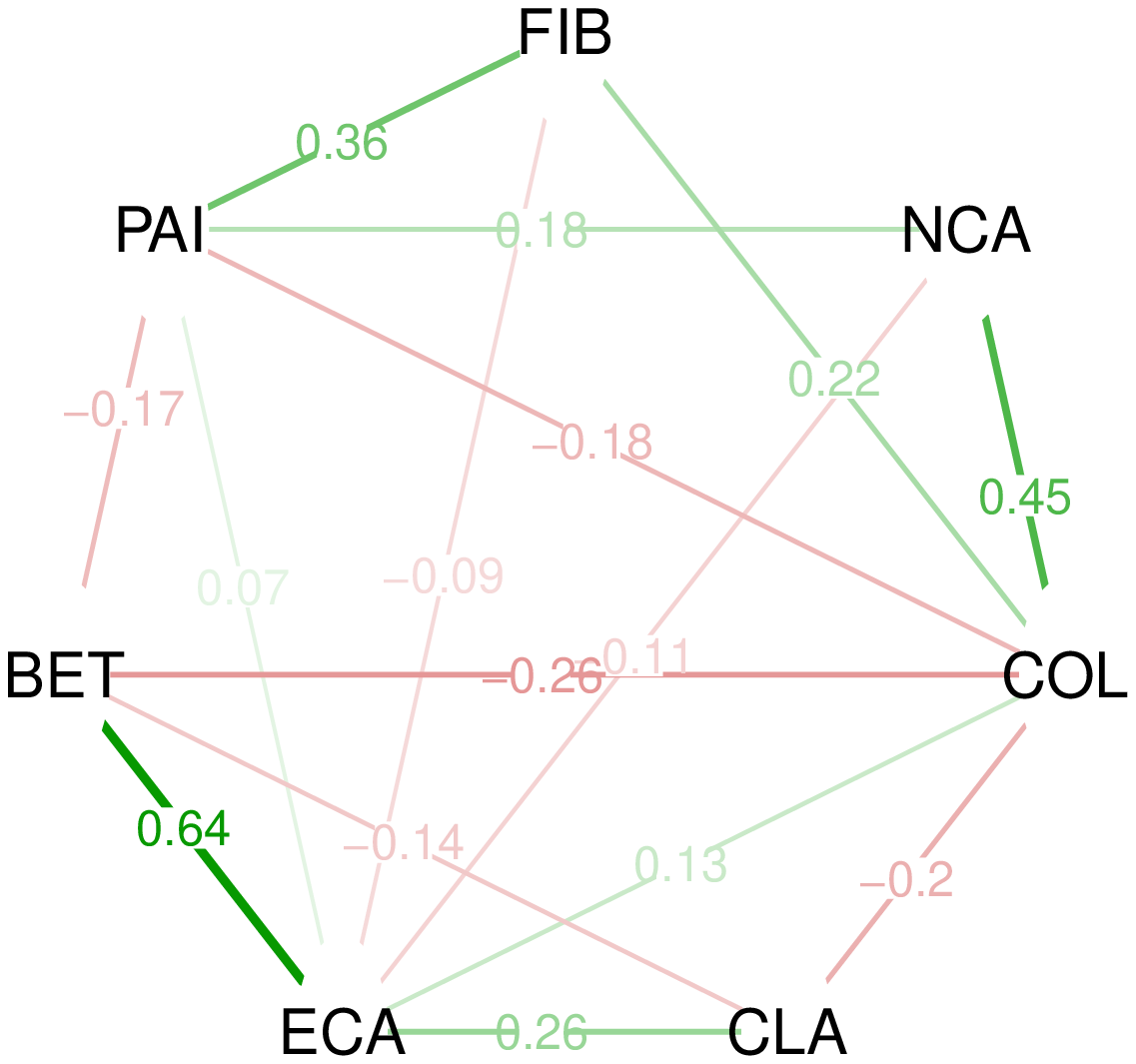}} \\	
		\small Core reactive & DNA damage response & EMT \\
		\addheight{\includegraphics[width=55mm]{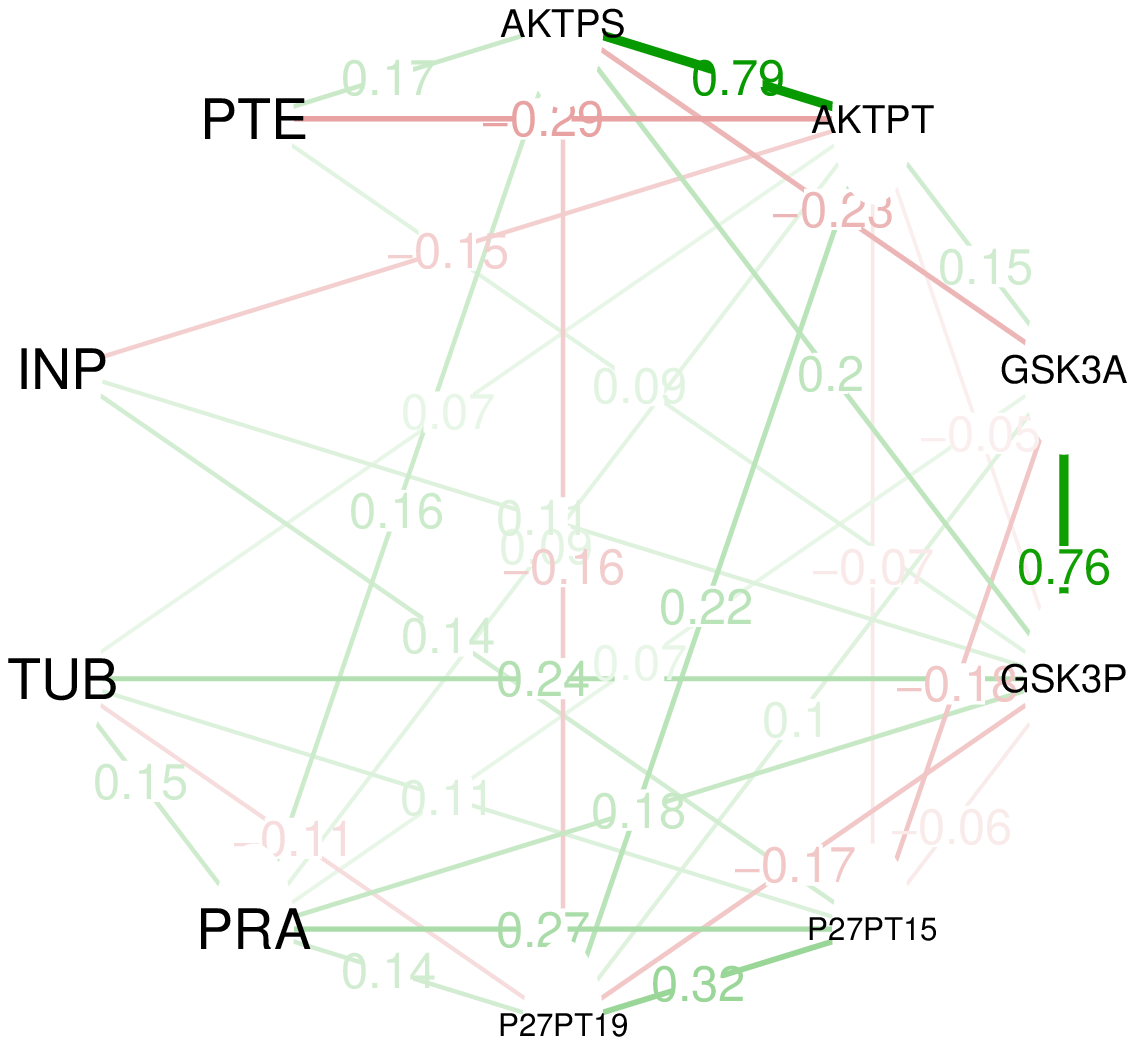}} &
		\addheight{\includegraphics[width=55mm]{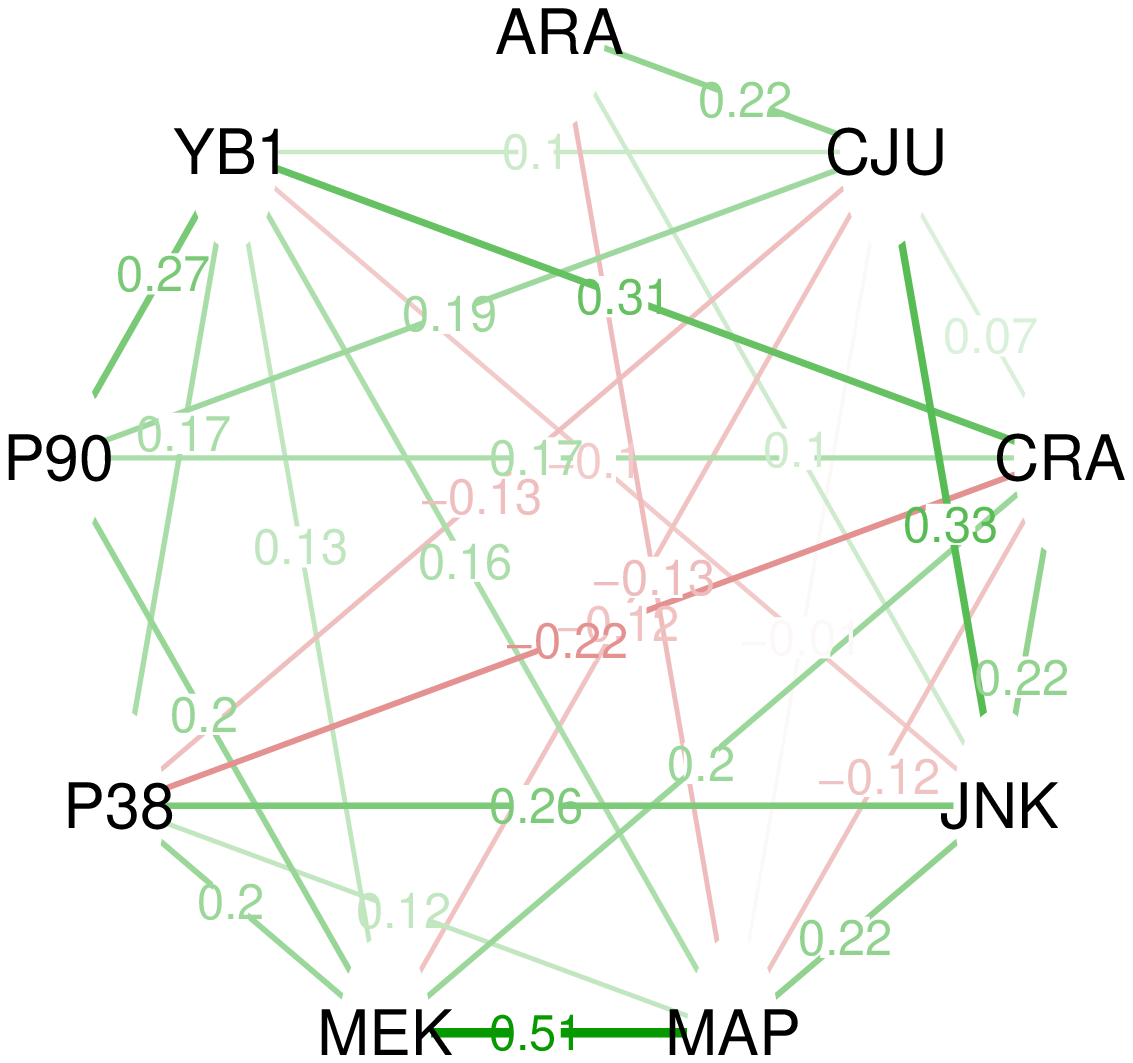}} &
		\addheight{\includegraphics[width=55mm]{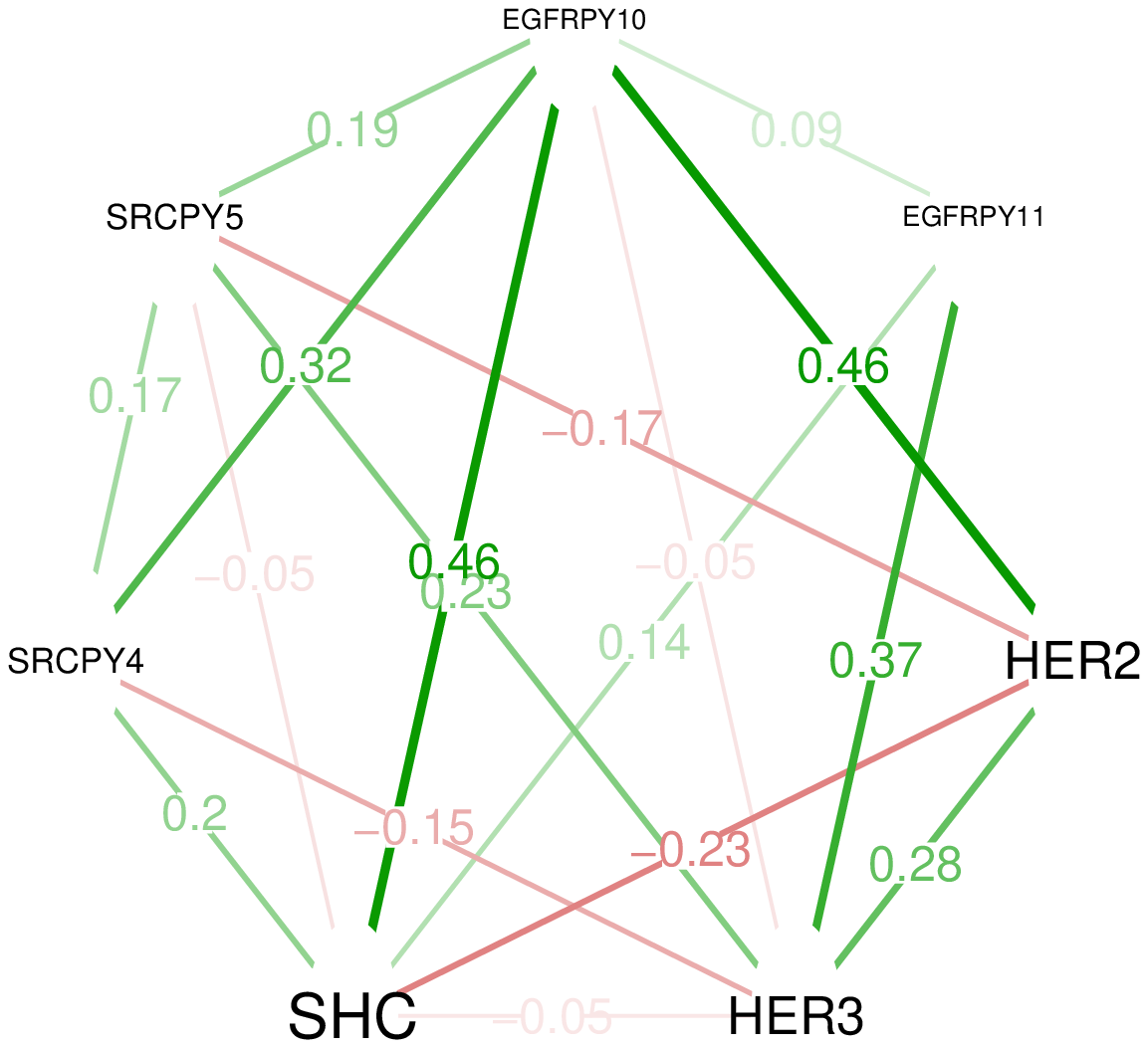}} \\
		\small PI3K/AKT & RAS/MAPK & RTK \\
		\addheight{\includegraphics[width=55mm]{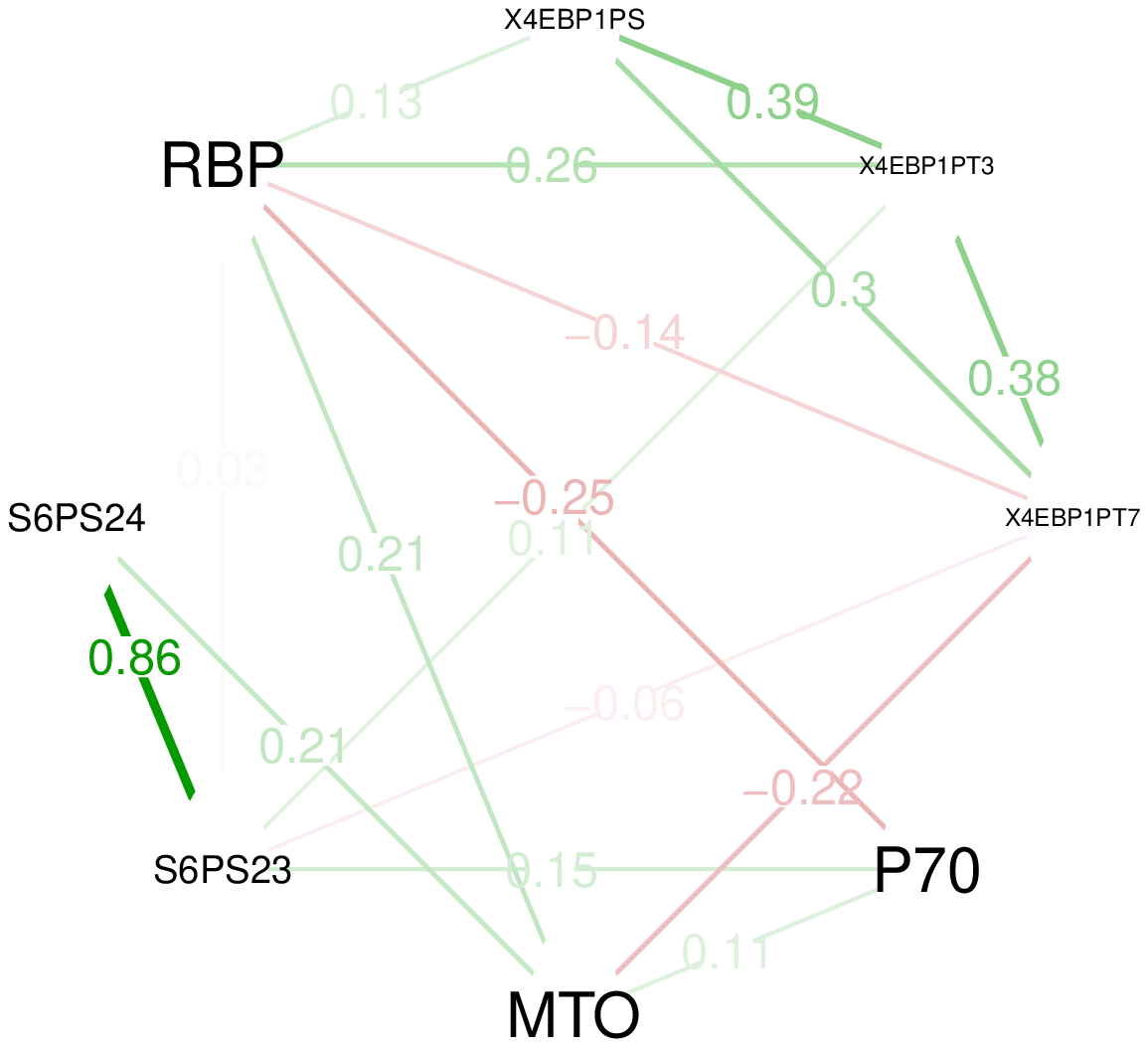}} &
		\addheight{\includegraphics[width=55mm]{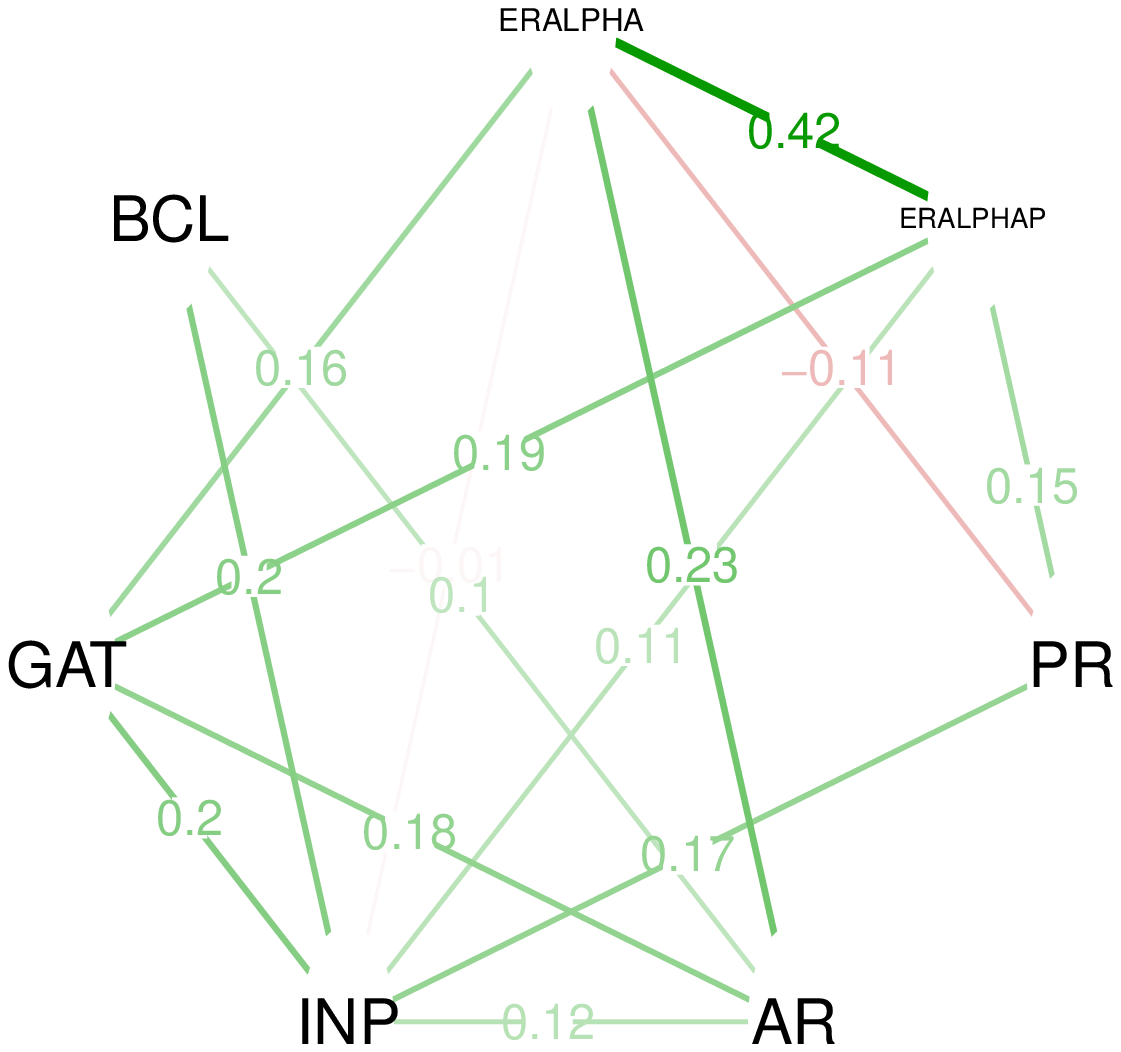}} \\
		\small TSC/mTOR & Hormone receptor\&signaling \\
	\end{tabular}
	}
	\caption*{Figure 3: Protein networks for all 12 pathways}
\end{table}

\bibliography{paper1-nyb}
\bibliographystyle{biometrika}

\newpage
\section*{Appendix 1}
In this appendix, we provide the detail calculation for the conditionally marginal density of $\mathrm{Y}$ given only $\bm{\gamma}$ and graph $G$. Given the hierarchical model in section 2,
\begin{eqnarray*}
(\mathrm{Y}-\mathrm{U}_{\bm\gamma} \mathrm{B}_{\bm{\gamma}, G})| \bm{\gamma}, \Sigma_G & \sim & \mathrm{MN}_{n \times q}(0,  I_n, \Sigma_G), \\
\mathrm{B}_{\bm{\gamma}, G}| \bm\gamma, \Sigma_G & \sim & \mathrm{MN}_{p_{\bm\gamma}(k+1) \times q}(0,g (\mathrm{U}_{\bm\gamma}^T \mathrm{U}_{\bm\gamma})^{-1}_{p_{\bm\gamma}(k+1)}, \Sigma_G), \\
\Sigma_G|G & \sim & \mathrm{HIW}_G (b,d I_q),
\end{eqnarray*}
where $\mathrm{Y}$ is $n \times q$, $\mathrm{U}_{\bm\gamma}$ is $n \times p_{\bm\gamma}(k+1)$, $\mathrm{B}_{\bm{\gamma}, G}$ is $p_{\bm\gamma}(k+1) \times q$, $\Sigma_G$ is $q \times q$. First, we can marginalize out the coefficient matrix $\mathrm{B}_{\bm{\gamma}, G}$ due to the conjugacy of its prior to the likelihood of $\mathrm{Y}$. We have
\begin{equation*}
\mathrm{Y} | \bm{\gamma}, \Sigma_G \sim \mathrm{MN}_{n \times q} (0, I_n + g\mathrm{U}_{\bm\gamma}(\mathrm{U}_{\bm\gamma}^T\mathrm{U}_{\bm\gamma})^{-1}\mathrm{U}_{\bm\gamma}^T, \Sigma_G).
\end{equation*}
Next, we vectorize $\mathrm{Y}$ to prepare for integrating out $\Sigma_G$. So,
\begin{equation}\label{vecY}
vec(\mathrm{Y}^T) | \bm{\gamma}, \Sigma_G \sim N_{nq} (0, (I_n + gP_{\bm\gamma})\otimes \Sigma_G),
\end{equation}
where $P_{\bm\gamma} = \mathrm{U}_{\bm\gamma}(\mathrm{U}_{\bm\gamma}^T\mathrm{U}_{\bm\gamma})^{-1}\mathrm{U}_{\bm\gamma}^T$ and $\otimes$ is the Kronecker product operation.

We use the Sylvester’s determinant theorem to further simplify the density of $vec(\mathrm{Y}^T)$. If $A$ and $B$ are matrices of size $m \times n$ and $n \times m$ respectively, then $|I_m + AB| = |I_n + BA|$. We have
\begin{equation}
\label{Sthm}
|I_n + gP_{\bm\gamma}| = |I_n + g\mathrm{U}_{\bm\gamma}(\mathrm{U}_{\bm\gamma}^T\mathrm{U}_{\bm\gamma})^{-1}\mathrm{U}_{\bm\gamma}^T| = |I_{p_{\bm\gamma}(k+1)} + g(\mathrm{U}_{\bm\gamma}^T\mathrm{U}_{\bm\gamma})^{-1}(\mathrm{U}_{\bm\gamma}^T\mathrm{U}_{\bm\gamma})| = {(g+1)}^{p_{\bm\gamma}(k+1)}.
\end{equation}

By the Sherman-Morrison-Woodbury (SMW) identity, assuming $A$, $C$ and $(C^{-1}+DA^{-1}B)$ to be nonsingular,
\begin{equation*}
(A+BCD)^{-1} = A^{-1} - A^{-1} B (C^{-1}+DA^{-1}B)^{-1} DA^{-1},
\end{equation*}
we have
\begin{equation}\label{smw}
(I_n + gP_{\bm\gamma})^{-1} = \big\{I_n + g\mathrm{U}_{\bm\gamma}(\mathrm{U}_{\bm\gamma}^T\mathrm{U}_{\bm\gamma})^{-1}\mathrm{U}_{\bm\gamma}^T\big\}^{-1} = I_n - \frac{g}{g+1}P_{\bm\gamma}.
\end{equation}

Simplifying the density of (\ref{vecY}) using (\ref{Sthm}) and (\ref{smw}),
\begin{align*}
f(vec(\mathrm{Y}^T)|\bm{\gamma}, \Sigma_G)
&= {(2\pi)}^{-\frac{nq}{2}} {|I_n + gP_{\bm\gamma}|}^{-\frac{q}{2}} {|\Sigma_G|}^{-\frac{n}{2}}\\
&\phantom{{(2\pi)}^{-\frac{nq}{2}} {|I_n + gP_{\bm\gamma}|}^{-\frac{q}{2}}}
exp \bigg[ -\frac{1}{2} {\{vec(\mathrm{Y}^T)\}}^T
\bigg\{{(I_n + gP_{\bm\gamma})}^{-1} \otimes {\Sigma_G}^{-1} \bigg\}
\{vec(\mathrm{Y}^T)\} \bigg] \\
&= {(2\pi)}^{-\frac{nq}{2}} {(g+1)}^{-\frac{p_{\bm\gamma}(k+1)q}{2}} {|\Sigma_G|}^{-\frac{n}{2}} \\
&\phantom{{(2\pi)}^{-\frac{nq}{2}} {(g+1)}^{-\frac{p_{\bm\gamma}(k+1)q}{2}}}
exp \bigg[ -\frac{1}{2} {\{vec(\mathrm{Y}^T)\}}^T
\bigg\{ \bigg(I_n - \frac{g}{g+1}P_{\bm\gamma} \bigg) \otimes {\Sigma_G}^{-1} \bigg\}
\{vec(\mathrm{Y}^T)\} \bigg].
\end{align*}

Matrix vectorization and trace operation have the following relationship. Suppose that $A$ is an $r \times s$ matrix and $B$ is $s \times r$, then $tr(AB) = {\{vec(A)\}}^T vec(B^T) = {\{vec(A^T)\}}^T vec(B)$. So we can further reduce the complexity of the exponential term in the density above.
\begin{align*}
&{\{vec(\mathrm{Y}^T)\}}^T
\bigg\{ \bigg(I_n - \frac{g}{g+1}P_{\bm\gamma} \bigg) \otimes {\Sigma_G}^{-1} \bigg\}
\{vec(\mathrm{Y}^T)\} \\
&= {\{vec(\mathrm{Y}^T)\}}^T \cdot \Bigg\{ \bigg(I_n - \frac{g}{g+1}P_{\bm\gamma} \bigg) \otimes {\Sigma_G}^{-1} \cdot vec(\mathrm{Y}^T) \Bigg\} \\
&= {\{vec(\mathrm{Y}^T)\}}^T \cdot vec \bigg\{\Sigma_G^{-1}\mathrm{Y}^T \bigg(I_n - \frac{g}{g+1}P_{\bm\gamma} \bigg)\bigg\} \\
&= tr \bigg\{ \mathrm{Y}^T \bigg(I_n - \frac{g}{g+1}P_{\bm\gamma} \bigg) \mathrm{Y} \Sigma_G^{-1} \bigg\} \\
&= tr\{S(\bm{\gamma}) \Sigma_G^{-1}\},
\end{align*}
where $S(\bm{\gamma}) = \mathrm{Y}^T \big(I_n - \frac{g}{g+1}P_{\bm\gamma} \big) \mathrm{Y}$. Eventually, by factorizing the density $f(\mathrm{Y}|\bm{\gamma}, \Sigma_G)$ corresponding to the hyper-inverse Wishart prior, we have
\begin{align}
f(\mathrm{Y}|\bm{\gamma}, \Sigma_G)
&= {(2\pi)}^{-\frac{nq}{2}} {(g+1)}^{-\frac{p_{\bm\gamma}(k+1)q}{2}} {|\Sigma_G|}^{-\frac{n}{2}}
etr \bigg\{ -\frac{1}{2} S(\bm{\gamma}) \Sigma_G^{-1} \bigg\} \nonumber \\
&= {(2\pi)}^{-\frac{nq}{2}} {(g+1)}^{-\frac{p_{\bm\gamma}(k+1)q}{2}}
\frac
{\prod_{C\in\mathscr{C}} {|\Sigma_C|}^{-\frac{n}{2}} etr \bigg\{ -\frac{1}{2} S_C(\bm{\gamma}) \Sigma_C^{-1} \bigg\}}
{\prod_{S\in\mathscr{S}} {|\Sigma_S|}^{-\frac{n}{2}} etr \bigg\{ -\frac{1}{2} S_S(\bm{\gamma}) \Sigma_S^{-1} \bigg\}}. \label{YgG}
\end{align}

Since $\Sigma_G | G \sim \mathrm{HIW}_G (b, dI_q)$, then
\begin{align}
f(\Sigma_G | G)
&= \frac
{\prod_{C\in\mathscr{C}}
\frac{{|dI_C|}^{\frac{b+|C|-1}{2}}}
{2^{\frac{(b+|C|-1)|C|}{2}}\Gamma_{|C|}\big(\frac{b+|C|-1}{2}\big)}
{|\Sigma_C|}^{-(\frac{b}{2}+|C|)} etr(-\frac{1}{2}dI_C \Sigma_C^{-1})
}
{\prod_{S\in\mathscr{S}}
\frac{{|dI_S|}^{\frac{b+|S|-1}{2}}}
{2^{\frac{(b+|S|-1)|S|}{2}}\Gamma_{|S|}\big(\frac{b+|S|-1}{2}\big)}
{|\Sigma_S|}^{-(\frac{b}{2}+|S|)} etr(-\frac{1}{2}dI_S \Sigma_S^{-1})
} \nonumber \\
&= \mathscr{H} (b, d I_q, G) \cdot
\frac
{\prod_{C\in\mathscr{C}}
{|\Sigma_C|}^{-(\frac{b}{2}+|C|)} etr(-\frac{1}{2}dI_C \Sigma_C^{-1})
}
{\prod_{S\in\mathscr{S}}
{|\Sigma_S|}^{-(\frac{b}{2}+|S|)} etr(-\frac{1}{2}dI_S \Sigma_S^{-1})
}, \label{hiw}
\end{align}
where
\begin{equation*}
\mathscr{H} (b, d I_q, G) =
\frac
{\prod_{C\in\mathscr{C}}
\frac{{|dI_C|}^{\frac{b+|C|-1}{2}}}
{2^{\frac{(b+|C|-1)|C|}{2}}\Gamma_{|C|}\big(\frac{b+|C|-1}{2}\big)}
}
{\prod_{S\in\mathscr{S}}
\frac{{|dI_S|}^{\frac{b+|S|-1}{2}}}
{2^{\frac{(b+|S|-1)|S|}{2}}\Gamma_{|S|}\big(\frac{b+|S|-1}{2}\big)}
}.
\end{equation*}

Next, we integrate out $\Sigma_G$ by (\ref{YgG}) and (\ref{hiw}),
\begin{align}
&f(\mathrm{Y}|\bm{\gamma}, G)= \int f(\mathrm{Y}|\bm{\gamma}, \Sigma_G)f(\Sigma_G | G) d\Sigma_G \nonumber\\
&= {(2\pi)}^{-\frac{nq}{2}} {(g+1)}^{-\frac{p_{\bm\gamma}(k+1)q}{2}} \mathscr{H} (b, d I_q, G)
\int \frac
{\prod_{C\in\mathscr{C}}
{|\Sigma_C|}^{-(\frac{b+n}{2}+|C|)} etr[-\frac{1}{2}(dI_C+S_C(\gamma)) \Sigma_C^{-1}]
}
{\prod_{S\in\mathscr{S}}
{|\Sigma_S|}^{-(\frac{b+n}{2}+|S|)} etr[-\frac{1}{2}(dI_S+S_S(\gamma)) \Sigma_S^{-1}]
} d\Sigma_G \nonumber\\
&={(2\pi)}^{-\frac{nq}{2}} {(g+1)}^{-\frac{p_{\bm\gamma}(k+1)q}{2}} \mathscr{H} (b, d I_q, G) \mathscr{H}^{-1} (b+n, d I_q + S(\bm{\gamma}), G) \nonumber\\
&= {(2\pi)}^{-\frac{nq}{2}} {(g+1)}^{-\frac{p_{\bm\gamma}(k+1)q}{2}} \mathscr{H} (b, d I_q, G)
\frac
{\prod_{C\in\mathscr{C}}
\frac{{|dI_C+S_C(\bm{\gamma})|}^{-\frac{b+n+|C|-1}{2}}}
{2^{-\frac{(b+n+|C|-1)|C|}{2}}\Gamma_{|C|}^{-1}\big(\frac{b+n+|C|-1}{2}\big)}
}
{\prod_{S\in\mathscr{S}}
\frac{{|dI_S+S_S(\bm{\gamma})|}^{-\frac{b+n+|S|-1}{2}}}
{2^{-\frac{(b+n+|S|-1)|S|}{2}}\Gamma_{|S|}^{-1}\big(\frac{b+n+|S|-1}{2}\big)}
} \nonumber\\
&= M_{n,G} \times {(g+1)}^{-\frac{p_{\bm\gamma}(k+1) q}{2}}
\frac
{\prod_{C\in\mathscr{C}} {|dI_C+S_C(\bm{\gamma})|}^{-\frac{b+n+|C|-1}{2}}}
{\prod_{S\in\mathscr{S}} {|dI_S+S_S(\bm{\gamma})|}^{-\frac{b+n+|S|-1}{2}}}, \label{marginalY}
\end{align}
where
\begin{equation*}
M_{n,G}=
{(2\pi)}^{-\frac{nq}{2}}
\frac
{\prod_{C\in\mathscr{C}}
\frac{{|dI_C|}^{\frac{b+|C|-1}{2}}}
{2^{-\frac{n|C|}{2}}\Gamma_{|C|}\big(\frac{b+|C|-1}{2}\big)\Gamma_{|C|}^{-1}\big(\frac{b+n+|C|-1}{2}\big)}
}
{\prod_{S\in\mathscr{S}}
\frac{{|dI_S|}^{\frac{b+|S|-1}{2}}}
{2^{-\frac{n|S|}{2}}\Gamma_{|S|}\big(\frac{b+|S|-1}{2}\big)\Gamma_{|S|}^{-1}\big(\frac{b+n+|S|-1}{2}\big)}
}
\end{equation*}
is a constant which depends only on $n$ and $G$, but it is the same for all $\bm\gamma$ under the same graph $G$. This makes possible for the cancellation in Metroplis-Hasting step for variable selection, which leads to faster in computation. $S_C(\bm{\gamma})$ and $S_S(\bm{\gamma})$ are the corresponding quadratic form similar to $S(\bm{\gamma})$ but restricted to sub-graphs denoted by $C$ and $S$ (i.e. cliques and separators).

\section*{Appendix 2.}
\newcommand{\RN}[1]{  \textup{\uppercase\expandafter{\romannumeral#1}}  }
We now show the complete proof of Lemma {\ref{lm1}} in this section. The main idea is, if the alternative model $\bm{a}$ does not contain the true model, the likelihood part drives the Bayes factor to zero exponentially, as $\bm{a}$ cannot fit the true mean adequately. If $\bm{a}$ contains true model then the difference in the likelihood becomes negligible but the prior penalizes for the extra dimensions and the Bayes factor goes to zero as $n$ goes to infinity.

By (\ref{marginalY}), the Bayes factor in favor of alternative model $\bm{a}$ under any graph $G$ is
\begin{align}
BF(\bm{a}; \bm{t}| G)
&= \frac{f(\mathrm{Y}|\bm{a}, G)}{f(\mathrm{Y}|\bm{t}, G)} \nonumber \\
&={(g+1)}^{-\frac{(p_{\bm{a}}-p_{\bm{t}})(k+1)q}{2}} \times \frac
{\prod_{C\in\mathscr{C}} {\big(\frac{|dI_C+S_C(\bm{a})|}{|dI_C+S_C(\bm{t})|}\big)}^{-\frac{b+n+|C|-1}{2}} }
{\prod_{S\in\mathscr{S}} {\big(\frac{|dI_S+S_S(\bm{a})|}{|dI_S+S_S(\bm{t})|}\big)}^{-\frac{b+n+|S|-1}{2}} } \nonumber\\
&:= {(g+1)}^{-\frac{(p_{\bm{a}}-p_{\bm{t}})(k+1)q}{2}} \times \frac
{\prod_{C\in\mathscr{C}} {\big\{ \Delta_C (\bm{a}, \bm{t}) \big\}}^{-\frac{b+n+|C|-1}{2}} }
{\prod_{S\in\mathscr{S}} {\big\{ \Delta_S (\bm{a}, \bm{t}) \big\}}^{-\frac{b+n+|S|-1}{2}} } \nonumber\\
&:= \RN{1} \times \RN{2}(\bm{a}, \bm{t}), \label{II}
\end{align}
where $S(\bm{a}) = \mathrm{Y}^T \big(I_n - \frac{g}{g+1}P_{\bm{a}} \big) \mathrm{Y}$, $S(\bm{t}) = \mathrm{Y}^T \big(I_n - \frac{g}{g+1}P_{\bm{t}} \big) Y$,  and $S_C$ and $S_S$ denote the quadratic forms restricted to clique $C \in \mathscr{C}$ and separator $S\in \mathscr{S}$; furthermore, we denote $\Delta_C(\bm{a}, \bm{t})=\frac{|dI_C+S_C(\bm{a})|}{|dI_C+S_C(\bm{t})|}$ and $\Delta_S(\bm{a}, \bm{t})=\frac{|dI_S+S_S(\bm{a})|}{|dI_S+S_S(\bm{t})|}$. Let $\Delta(\bm{a}, \bm{t}) = \frac{|dI_q+S(\bm{a})|}{|dI_q+S(\bm{t})|}$ be the version of $\Delta(\bm{a}, \bm{t})$ for the whole graph $G$.

\begin{lemma}\label{lm8}
$\Delta(\bm{a}, \bm{t}) = \Big|I_q +\frac{1}{d}\frac{g}{g+1} \Big(\frac{1}{n} A_{\bm{t}}\Big)^{-\frac{1}{2}} \Big\{\frac{1}{n}\mathrm{Y}^T \big(P_{\bm{t}} - P_{\bm{a}} \big) \mathrm{Y}\bigg\} \Big(\frac{1}{n} A_{\bm{t}}\Big)^{-\frac{1}{2}}\Big|$, where $A_{\bm{t}}=I_q + \frac{1}{d}\mathrm{Y}^T \big(I_n - \frac{g}{g+1}P_{\bm{t}} \big) Y$.
\end{lemma}
\begin{proof}
\begin{align*}
\Delta(\bm{a}, \bm{t})
& = \frac{|dI_q + S(\bm{a})|}{|dI_q + S(\bm{t})|}
 = \frac{|I_q + \frac{1}{d}S(\bm{a})|}{|I_q + \frac{1}{d}S(\bm{t})|} \\
& = \frac{|I_q + \frac{1}{d}\mathrm{Y}^T \big(I_n - \frac{g}{g+1}P_{\bm{a}} \big) \mathrm{Y}|}
{|I_q + \frac{1}{d}\mathrm{Y}^T \big(I_n - \frac{g}{g+1}P_{\bm{t}} \big) \mathrm{Y}|} \\
& = \frac{|I_q + \frac{1}{d}\mathrm{Y}^T \big(I_n - \frac{g}{g+1}P_{\bm{t}} \big) \mathrm{Y}
+\frac{1}{d}\mathrm{Y}^T \big(I_n - \frac{g}{g+1}P_{\bm{a}} \big) \mathrm{Y} - \frac{1}{d}\mathrm{Y}^T \big(I_n - \frac{g}{g+1}P_{\bm{t}} \big) \mathrm{Y}|}
{|I_q + \frac{1}{d}\mathrm{Y}^T \big(I_n - \frac{g}{g+1}P_{\bm{t}} \big) \mathrm{Y}|} \\
& = \frac{|I_q + \frac{1}{d}\mathrm{Y}^T \big(I_n - \frac{g}{g+1}P_{\bm{t}} \big) \mathrm{Y}
+\frac{1}{d}\frac{g}{g+1}\mathrm{Y}^T \big(P_{\bm{t}} - P_{\bm{a}} \big) \mathrm{Y}|}
{|I_q + \frac{1}{d}\mathrm{Y}^T \big(I_n - \frac{g}{g+1}P_{\bm{t}} \big) \mathrm{Y}|} \\
& = \frac{|A_{\bm{t}} +\frac{1}{d}\frac{g}{g+1}\mathrm{Y}^T \big(P_{\bm{t}} - P_{\bm{a}} \big) \mathrm{Y}|} {|A_{\bm{t}}|} \\
& = \Big|I_q +\frac{1}{d}\frac{g}{g+1} A_{\bm{t}}^{-\frac{1}{2}} \mathrm{Y}^T \big(P_{\bm{t}} - P_{\bm{a}} \big) \mathrm{Y} A_{\bm{t}}^{-\frac{1}{2}}\Big|\\
& = \Big|I_q +\frac{1}{d}\frac{g}{g+1} \Big(\frac{1}{n} A_{\bm{t}}\Big)^{-\frac{1}{2}} \Big\{\frac{1}{n}\mathrm{Y}^T \big(P_{\bm{t}} - P_{\bm{a}} \big) \mathrm{Y}\Big\} \Big(\frac{1}{n} A_{\bm{t}}\Big)^{-\frac{1}{2}}\Big|.
\end{align*}
\end{proof}

\begin{remark}
Lemma \ref{lm8} is with respect to the whole graph $G$, but the same result holds for every clique $C$ and separator $S$. And similarly we have $A^C_{\bm{t}}$ and $A^S_{\bm{t}}$ for clique $C$ and separator $S$, respectively. For simplicity, we will not show the results for cliques and separators. In the next several lemmas, we only show the results with respect to the whole graph $G$, but they all hold for any subgraphs of $G$, i.e. cliques and separators.
\end{remark}

Next we split Bayes factor $\mathrm{BF}(\bm{a};\bm{t}|G)$ into two parts $\mathrm{BF}(\bm{a};\bm{a}\cup\bm{t}|G)$ and $\mathrm{BF}(\bm{a}\cup\bm{t};\bm{t}|G)$ and show them both converge to zero as $n\rightarrow\infty$. But before that, we need to introduce several lemmas.

\begin{lemma}\label{lm8.1}
	Under Condition \ref{cond1}, $\mathrm{p}\lim_{n\rightarrow\infty}\frac{\mathrm{Y}^T(I_n-P_{\bm{t}})\mathrm{Y}}{n} = \Sigma_{G^*}$, where $\Sigma_{G^*}$ is the true covariance matrix with respect to the true graph $G^*$.
\end{lemma}
\begin{proof}
	Since $Y|\bm{t}, \Sigma_{G^*} \sim \mathrm{MN}_{n\times q}(\mathrm{U}_{\bm{t}}\mathrm{B}_{\bm{t},G^*}, I_n, \Sigma_{G^*})$ and $I_n-P_{\bm{t}}$ is symmetric and idempotent, by Corollary 2.1 in \cite{singull2012distribution}, we have $\frac{\mathrm{Y}^T(I_n-P_{\bm{t}})\mathrm{Y}}{n} \sim W_q(n-r_{\bm{t}}, \frac{1}{n}\Sigma_{G^*})$ and the non-central parameter is zero here. Let $\tilde{y}_{ij}(n), i,j = 1,\dots,q$ denote the entries of $\frac{\mathrm{Y}^T(I_n-P_{\bm{t}})\mathrm{Y}}{n}$ and $\tilde{y}_{ij}(n)=\tilde{y}_{ji}(n)$. Further more, let $\sigma_{ij}^*, i,j = 1,\dots,q$ be the entries of $\Sigma_{G^*}$, and $\sigma_{ij}^* = \sigma_{ji}^*$. Since $\mathbb{E}\big\{\frac{\mathrm{Y}^T(I_n-P_{\bm{t}})\mathrm{Y}}{n}\big\} = (1 - \frac{r_{\bm{t}}}{n})\Sigma_{G^*}\rightarrow \Sigma_{G^*}$, then $\mathbb{E}\big\{\tilde{y}_{ij}(n)\big\} =(1 - \frac{r_{\bm{t}}}{n})\sigma_{ij}^*\rightarrow \sigma_{ij}^*$ and $Var\big\{\tilde{y}_{ij}(n) \big\}=\frac{1}{n}(\sigma_{ij}^{*2}+\sigma_{ii}^*\sigma_{jj}^*)\rightarrow 0$. Thus for any $\epsilon>0$, there exist $M_{\epsilon}$, when $n>M_{\epsilon}$, such that $|(1 - \frac{r_{\bm{t}}}{n})\sigma_{ij}^*- \sigma_{ij}^*|<\epsilon/2$, then $Pr(|\tilde{y}_{ij}(n)-\sigma_{ij}^*|>\epsilon)\leq Pr(|\tilde{y}_{ij}(n)-\mathbb{E}\{\tilde{y}_{ij}(n)\}|>\epsilon/2)\leq \frac{4Var\{\tilde{y}_{ij}(n))\}}{\epsilon^2}=\frac{4(\sigma_{ij}^{*2}+\sigma_{ii}^*\sigma_{jj}^*)}{n\epsilon^2}\rightarrow 0$. So $\tilde{y}_{ij}(n)\xrightarrow{p}\sigma_{ij}^*$ in probability, for all $i,j=1,\dots,q$. Therefore, $\frac{\mathrm{Y}^T(I_n-P_{\bm{t}})\mathrm{Y}}{n} \xrightarrow{p} \Sigma_{G^*}$ as $n\rightarrow\infty$.
\end{proof}

\begin{lemma}\label{lm8.2}
	Under Condition \ref{cond1}, \ref{cond2}, \ref{cond4}, $\mathrm{p}\lim_{n\rightarrow\infty}\frac{1}{n}A_{\bm{t}} = \Sigma_{G^*}$, where $\Sigma_{G^*}$ is the true covariance matrix with respect to the true graph $G^*$.
\end{lemma}
\begin{proof}
	First, we show $\mathrm{p}\lim_{n\rightarrow\infty}\frac{1}{n}\frac{1}{g+1}\mathrm{Y}^TP_{\bm{t}}\mathrm{Y} = \bm{\mathrm{0}}_{q \times q}$. Let $\mathrm{y}_i$ be the $i$th column of $\mathrm{Y}$. Note that $\lim_{n\rightarrow\infty}\mathbb{E}(\frac{1}{n}\frac{1}{g+1}\mathrm{y}_i^TP_{\bm{t}}\mathrm{y}_i)=0$ and $\lim_{n\rightarrow\infty}Var(\frac{1}{n}\frac{1}{g+1}\mathrm{y}_i^TP_{\bm{t}}\mathrm{y}_i)=0$, so $\frac{1}{n}\frac{1}{g+1}\mathrm{y}_i^TP_{\bm{t}}\mathrm{y}_i \rightarrow 0$ in probability. Hence, $\sum_{i=1}^q\frac{1}{n}\frac{1}{g+1}\mathrm{y}_i^TP_{\bm{t}}\mathrm{y}_i \rightarrow 0$ in probability. Therefore, the sum of eigenvalues of matrix $\frac{1}{n}\frac{1}{g+1}\mathrm{Y}^TP_{\bm{t}}\mathrm{Y}$ goes to zero in probability. Let $\lambda^t_i$ be the $i$th eigenvalue of $\frac{1}{n}\frac{1}{g+1}\mathrm{Y}^TP_{\bm{t}}\mathrm{Y}$, $i=1,2,\dots,q$. So $\lambda^t_i$ goes to zero in probability as the matrix is non-negative definite.  Using spectral decomposition, $\frac{1}{n}\frac{1}{g+1}\mathrm{Y}^TP_{\bm{t}}\mathrm{Y}=\sum_{i=1}^q\lambda^t_i u_i u^T_i$, where $u_i$'s are orthonormal eigenvectors, each of the entries of  $\frac{1}{n}\frac{1}{g+1}\mathrm{Y}^TP_{\bm{t}}\mathrm{Y}$ goes to zero in probability and our claim follows. 

	Therefore, 
	\begin{align*}
	\mathrm{p}\lim_{n\rightarrow\infty}\frac{1}{n}A_{\bm{t}} & = \mathrm{p}\lim_{n\rightarrow\infty}\frac{1}{n} \bigg\{I_q + \frac{1}{d}\mathrm{Y}^T \big(I_n - \frac{g}{g+1}P_{\bm{t}} \big) \mathrm{Y}\bigg\} \\
	& = \mathrm{p}\lim_{n\rightarrow\infty}\frac{1}{n}I_q + \mathrm{p}\lim_{n\rightarrow\infty}\frac{1}{d}\frac{\mathrm{Y}^T \big(I_n - P_{\bm{t}} \big) \mathrm{Y}}{n} + \mathrm{p}\lim_{n\rightarrow\infty}\frac{1}{d}\frac{1}{g+1}\frac{\mathrm{Y}^T P_{\bm{t}} \mathrm{Y}}{n}  \\
	& = \bm{\mathrm{0}}_{q \times q} + \Sigma_{G^*} + \bm{\mathrm{0}}_{q \times q} = \Sigma_{G^*}.
	\end{align*}
\end{proof}

\begin{lemma}\label{lm8.3}
	Let $\widetilde{\lambda}^{\bm{a}}_i, i = 1, \dots, q$ be the eigenvalues of $\frac{1}{n}S(\bm{a})$ and $\widetilde{\lambda}^{\bm{a}\cup\bm{t}}_i, i = 1, \dots, q$ be the eigenvalues of $\frac{1}{n}S(\bm{a}\cup\bm{t})$, where $S(\bm{a}) = \mathrm{Y}^T \big(I_n - \frac{g}{g+1}P_{\bm{a}} \big) \mathrm{Y}$ and $S(\bm{a}\cup\bm{t}) = \mathrm{Y}^T \big(I_n - \frac{g}{g+1}P_{\bm{a}\cup\bm{t}} \big) \mathrm{Y}$. Under Condition \ref{cond1}, \ref{cond2}, \ref{cond4}, \ref{cond6}, $Pr(\widetilde{\lambda}^{\bm{a}}_i > \bar{C}) \rightarrow 0$ and $Pr(\widetilde{\lambda}^{\bm{a}\cup\bm{t}}_i > \bar{C}) \rightarrow 0$, $i = 1, \dots, q$, as $n\rightarrow \infty$, where $\bar{C}$ is some fixed positive constant.
\end{lemma}
\begin{proof}
	Let $\mathrm{y}_i$ be the $i$th column of $\mathrm{Y}$ and $\mathrm{b}_i$ be the $i$th column of $\mathrm{B}_{\bm{t},G^*}$. Then $\mathrm{v}_i := \mathrm{U}_{\bm{t}}\mathrm{b}_i$ is the $i$th column of $\mathrm{U}_{\bm{t}}\mathrm{B}_{\bm{t},G^*}$, $i=1,\dots, q$. Next, we have
	\begin{align*}
	tr\Big\{ \frac{S(\bm{a})}{n} \Big\} &= tr\Big\{ \frac{1}{n}\mathrm{Y}^T \Big(I_n - \frac{g}{g+1}P_{\bm{a}} \Big) \mathrm{Y} \Big\} \\
	&= tr\Big\{ \frac{1}{n}\mathrm{Y}^T \Big(I_n - P_{\bm{a}} \Big) \mathrm{Y} \Big\} +tr \Big\{ \frac{1}{n} \frac{1}{g+1}\mathrm{Y}^T P_{\bm{a}}\mathrm{Y} \Big\} \\
	&= \sum_{i=1}^{q} tr\Big\{ \frac{1}{n}\mathrm{y}_i^T \Big(I_n - P_{\bm{a}} \Big) \mathrm{y}_i \Big\} + \sum_{i=1}^{q}tr\Big\{ \frac{1}{n}\frac{1}{g+1}\mathrm{y}_i^T P_{\bm{a}}\mathrm{y}_i \Big\}.
	\end{align*}

	Let $r_{\bm{a}} = rank (\mathrm{U}_{\bm{a}})$. Since $\mathrm{y}_i\sim N_n(\mathrm{v}_i, I_n)$, $i = 1, \dots, q$, we have
	\begin{align*}
	\mathrm{y}_i^T \Big(I_n - P_{\bm{a}} \Big) \mathrm{y}_i &\sim \chi^2_{n-r_{\bm{a}}}(\phi^{n-\bm{a}}_i), \\
	\mathrm{y}_i^T P_{\bm{a}} \mathrm{y}_i &\sim \chi^2_{r_{\bm{a}}}(\phi^{\bm{a}}_i),
	\end{align*}
	where $\phi^{n-\bm{a}}_i = \frac{1}{2}\mathrm{v}_i^T(I_n - P_{\bm{a}})\mathrm{v}_i$, $\phi^{\bm{a}}_i = \frac{1}{2}\mathrm{v}_i^TP_{\bm{a}}\mathrm{v}_i$. By Condition \ref{cond2}, we have $\frac{1}{n}\phi^{n-\bm{a}}_i = \frac{1}{2n}\mathrm{v}_i^T(I_n - P_{\bm{a}})\mathrm{v}_i \leq \frac{1}{n}\mathrm{v}_i^T\mathrm{v}_i = \mathrm{b}_i^T\frac{\mathrm{U}_{\bm{t}}^T\mathrm{U}_{\bm{t}}}{n} \mathrm{b}_i   \leq \frac{\lambda_{max}}{n}\| \mathrm{b}_i \|^2_2 < b_Md_\mathrm{U}$, where $b_M = max\{  \| \mathrm{b}_i \|^2_2, i = 1, \dots, q  \} < \infty$. Similarly, $\phi^{\bm{a}}_i = \frac{1}{2}\mathrm{v}_i^TP_{\bm{a}}\mathrm{v}_i\leq b_Md_\mathrm{U}$. Next, 
	\begin{align*}
	\mathbb{E}\Big[\frac{1}{n}\mathrm{y}_i^T \Big(I_n - P_{\bm{a}} \Big) \mathrm{y}_i \Big] &= \frac{1}{n}(n-r_{\bm{a}}+\phi^{n-\bm{a}}_i) \leq 1 + \frac{1}{n}\phi^{n-\bm{a}}_i < 1 + b_Md_\mathrm{U}, \\
	Var\Big[\frac{1}{n}\mathrm{y}_i^T \Big(I_n - P_{\bm{a}} \Big) \mathrm{y}_i \Big] &= \frac{1}{n^2}(2n-2r_{\bm{a}}+4\phi^{n-\bm{a}}_i) \leq \frac{1}{n}\Big(2 + \frac{4}{n}\phi^{n-\bm{a}}_i\Big) < \frac{1}{n}\Big(2 + 4b_Md_\mathrm{U}\Big)\rightarrow 0.
	\end{align*}
	Analogously, by Condition \ref{cond4} and \ref{cond6},
	\begin{align*}
	\mathbb{E}\Big[\frac{1}{g+1}\frac{1}{n}\mathrm{y}_i^T P_{\bm{a}} \mathrm{y}_i \Big] &= \frac{1}{g+1}\frac{1}{n}(r_{\bm{a}}+\phi^{\bm{a}}_i) < \frac{1}{g+1}b_Md_\mathrm{U} \rightarrow 0, \\
	Var\Big[\frac{1}{g+1}\frac{1}{n}\mathrm{y}_i^T P_{\bm{a}} \mathrm{y}_i \Big] &= \frac{1}{(g+1)^2}\frac{1}{n^2}(2r_{\bm{a}}+4\phi^{\bm{a}}_i) < \frac{1}{(g+1)^2}\frac{1}{n}\Big(2 + 4b_Md_\mathrm{U}\Big) \rightarrow 0,
	\end{align*}
	So for any $\bar{\epsilon}>0$, we have $Pr\Big\{\frac{1}{n}\mathrm{y}_i^T \big(I_n - P_{\bm{a}} \big) \mathrm{y}_i > 1 + b_Md_\mathrm{U} + \bar{\epsilon} \Big\} \rightarrow 0$ and $Pr\Big\{\frac{1}{g+1}\frac{1}{n}\mathrm{y}_i^T P_{\bm{a}} \mathrm{y}_i > \bar{\epsilon} \Big\} \rightarrow 0$, $i = 1. \dots, q$, as $n\rightarrow 0$. By combining the two results together,
	\begin{align*}
	& Pr\Big\{\frac{1}{n}\mathrm{y}_i^T \Big(I_n - \frac{g}{g+1}P_{\bm{a}} \Big) \mathrm{y}_i > 1 + b_Md_\mathrm{U} + 2\bar{\epsilon} \Big\} \\
	= & Pr\Big\{ \frac{1}{n}\mathrm{y}_i^T \Big(I_n - P_{\bm{a}} \Big) \mathrm{y}_i + \frac{1}{g+1}\frac{1}{n}\mathrm{y}_i^T P_{\bm{a}} \mathrm{y}_i > 1 + b_Md_\mathrm{U} + \bar{\epsilon} + \bar{\epsilon} \Big\} \\
	\leq & Pr\Big\{\frac{1}{n}\mathrm{y}_i^T \big(I_n - P_{\bm{a}} \big) \mathrm{y}_i > 1 + b_Md_\mathrm{U} + \bar{\epsilon} \Big\}
	+ Pr\Big\{\frac{1}{g+1}\frac{1}{n}\mathrm{y}_i^T P_{\bm{a}} \mathrm{y}_i > \bar{\epsilon} \Big\} \rightarrow 0.
	\end{align*}
	Therefore,
	\begin{align*}
	& Pr\Big\{ \widetilde{\lambda}^{\bm{a}}_i > q(1 + b_Md_\mathrm{U} + 2\bar{\epsilon}) \Big\} \leq Pr\Big\{ \sum_{i=1}^{q}\widetilde{\lambda}^{\bm{a}}_i > q(1 + b_Md_\mathrm{U} + 2\bar{\epsilon})   \Big\} \\
	= & Pr\Big\{ tr\Big( \frac{S(\bm{a})}{n} \Big) > q(1 + b_Md_\mathrm{U} + 2\bar{\epsilon})   \Big\} \\
	= & Pr\Big\{ \sum_{i=1}^{q}\frac{1}{n}\mathrm{y}_i^T \Big(I_n - \frac{g}{g+1}P_{\bm{a}} \Big) \mathrm{y}_i > q(1 + b_Md_\mathrm{U} + 2\bar{\epsilon}) \Big\} \\
	\leq & Pr\Big\{ \cup_{i=1}^{q}\Big(\frac{1}{n}\mathrm{y}_i^T \Big(I_n - \frac{g}{g+1}P_{\bm{a}} \Big) \mathrm{y}_i > 1 + b_Md_\mathrm{U} + 2\bar{\epsilon}\Big) \Big\} \\
	\leq & \sum_{i=1}^{q} Pr\Big\{\frac{1}{n}\mathrm{y}_i^T \Big(I_n - \frac{g}{g+1}P_{\bm{a}} \Big) \mathrm{y}_i > 1 + b_Md_\mathrm{U} + 2\bar{\epsilon} \Big\} \rightarrow 0.
	\end{align*}
	Let $\bar{\epsilon} = 0.5$ and $\bar{C} = q(2 + b_Md_\mathrm{U})$, we have $Pr(\widetilde{\lambda}^{\bm{a}}_i > \bar{C}) \rightarrow 0$, $i = 1, \dots, q$, as $n\rightarrow 0$. Same as the proof above, we can show $Pr(\widetilde{\lambda}^{\bm{a}\cup\bm{t}}_i > \bar{C}) \rightarrow 0$, $i = 1, \dots, q$, as $n\rightarrow 0$.
\end{proof}

\begin{lemma}\label{lm8.4}
	Under Condition \ref{cond1}, \ref{cond2} and \ref{cond6}, when $\bm{a}\nsubseteq\bm{t}$,
\begin{enumerate}
	\itemsep0em
	\item If $p_{\bm{a}}$ is bounded, the largest eigenvalue of $\mathrm{Y}^T \big(P_{\bm{a}\cup\bm{t}} - P_{\bm{t}} \big) \mathrm{Y}$ is $O_p(1)$;
	\item If $p_{\bm{a}}$ is unbounded, the largest eigenvalue of $\mathrm{Y}^T \big(P_{\bm{a}\cup\bm{t}} - P_{\bm{t}} \big) \mathrm{Y}$ is at most $O_p(r_{\bm{a} \cup \bm{t}})$.
\end{enumerate}
\end{lemma}
\begin{proof}
	As we know $P_{\bm{a} \cup \bm{t}}-P_{\bm{t}}$ is idempotent, then follow the same notations as in Lemma \ref{lm8.3}, we have $tr\{ \mathrm{Y}^T \big(P_{\bm{a}\cup\bm{t}} - P_{\bm{t}} \big) \mathrm{Y} \} = \sum_{i=1}^q \mathrm{y}_i^T(P_{\bm{a} \cup \bm{t}}-P_{\bm{t}})\mathrm{y}_i$, and $\mathrm{y}_i^T(P_{\bm{a} \cup \bm{t}}-P_{\bm{t}})\mathrm{y}_i \sim \chi^2_{r_{\bm{a} \cup \bm{t}}-r_{\bm{t}}}$. If $p_{\bm{a}}$ is bounded, then $r_{\bm{a} \cup \bm{t}}-r_{\bm{t}} = O(1)$. In this case, $tr\{ \mathrm{Y}^T \big(P_{\bm{a}\cup\bm{t}} - P_{\bm{t}} \big) \mathrm{Y} \} = O_p(1)$, which means the largest eigenvalue of $\mathrm{Y}^T \big(P_{\bm{a}\cup\bm{t}} - P_{\bm{t}} \big) \mathrm{Y}$ is $O_p(1)$. By Condition \ref{cond2}, $U_{\bm{a}\cup\bm{t}}$ has full column rank, then $r_{\bm{a} \cup \bm{t}}=p_{\bm{a} \cup \bm{t}}(k+1)$. If $p_{\bm{a}}$ is unbounded, then $r_{\bm{a} \cup \bm{t}} - r_{\bm{t}} \preceq O_p(r_{\bm{a} \cup \bm{t}})$. So $tr\{ \mathrm{Y}^T \big(P_{\bm{a}\cup\bm{t}} - P_{\bm{t}} \big) \mathrm{Y} \} \preceq O_p(r_{\bm{a} \cup \bm{t}})$, which means the largest eigenvalue of $\mathrm{Y}^T \big(P_{\bm{a}\cup\bm{t}} - P_{\bm{t}} \big) \mathrm{Y}$ is at most $O_p(r_{\bm{a} \cup \bm{t}})$. By Condition \ref{cond6} we know $\frac{r_{\bm{a} \cup \bm{t}}}{n}=o_p(n)$.
\end{proof}

\begin{lemma}\label{lm8.5}
	Let $\widetilde{\lambda}^{\bm{a}\cup\bm{t} - \bm{a}}_M$ denote the largest eigenvalue of $\frac{1}{n}\mathrm{Y}^T \big(P_{\bm{a}\cup\bm{t}} - P_{\bm{a}} \big) \mathrm{Y}$. Under Condition \ref{cond1}, \ref{cond3}, $Pr(\widetilde{\lambda}^{\bm{a}\cup\bm{t} - \bm{a}}_M > \bar{\bar{C}}) \rightarrow 1$, as $n\rightarrow\infty$, where $\bar{\bar{C}}$ is some fixed positive constant.
\end{lemma}
\begin{proof}
	Follow the same notations as in Lemma \ref{lm8.3}, $tr\big\{\frac{1}{n}\mathrm{Y}^T \big(P_{\bm{a} \cup \bm{t}} - P_{\bm{a}} \big) \mathrm{Y}\big\} = \sum_{i=1}^{q} \frac{1}{n}\mathrm{y}_i^T \big(P_{\bm{a} \cup \bm{t}} - P_{\bm{a}} \big) \mathrm{y}_i$.	Then,
	\begin{equation*}
	\mathrm{y}_i^T \big(P_{\bm{a} \cup \bm{t}} - P_{\bm{a}} \big) \mathrm{y}_i \sim \chi^2_{r_{\bm{a}\cup\bm{t}} - r_{\bm{a}}}(\phi^{\bm{a}\cup\bm{t} - \bm{a}}_i),
	\end{equation*}	
	where $\phi^{\bm{a}\cup\bm{t} - \bm{a}}_i = \frac{1}{2}\mathrm{v}_i^T(P_{\bm{a}\cup\bm{t}} - P_{\bm{a}})\mathrm{v}_i = \frac{1}{2}\mathrm{v}_i^T(I_n - P_{\bm{a}})P_{\bm{t}}\mathrm{v}_i = \frac{1}{2}\mathrm{v}_i^T(I_n - P_{\bm{a}})\mathrm{v}_i$ and $r_{\bm{a}\cup\bm{t}} - r_{\bm{a}} \leq r_{\bm{t}} < \infty$. As in Lemma \ref{lm8.3}, we know $\frac{1}{n}\phi^{\bm{a}\cup\bm{t} - \bm{a}}_i \leq \frac{1}{n}\mathrm{v}_i^T \mathrm{v}_i \leq b_M d_{\mathrm{U}}$.
	Next, by Condotion \ref{cond3},
	\begin{align*}
	\mathbb{E}\Big[tr\Big\{\frac{1}{n}\mathrm{Y}^T \big(P_{\bm{a} \cup \bm{t}} - P_{\bm{a}} \big) \mathrm{Y}\Big\}\Big] 
	&= \sum_{i=1}^{q} \mathbb{E}\Big[ \frac{1}{n}\mathrm{y}_i^T \big(P_{\bm{a} \cup \bm{t}} - P_{\bm{a}} \big) \mathrm{y}_i \Big]
	= \sum_{i=1}^{q} \frac{1}{n}(r_{\bm{a}\cup\bm{t}} - r_{\bm{a}} + \phi^{\bm{a}\cup\bm{t} - \bm{a}}_i) \\
	&\geq \frac{1}{2n}\sum_{i=1}^{q}\mathrm{v}_i^T(I_n - P_{\bm{a}})\mathrm{v}_i = \frac{1}{2n} tr\{\mathrm{E}_y^T(I_n-P_{\bm{a}})\mathrm{E}_y\} > C_0/2,
	\end{align*}
	\begin{align*}
	Var\Big[\frac{1}{n}\mathrm{y}_i^T \big(P_{\bm{a} \cup \bm{t}} - P_{\bm{a}} \big) \mathrm{y}_i\Big] 
	&= \frac{1}{n^2}(2r_{\bm{a}\cup\bm{t}} - 2r_{\bm{a}} + 4\phi^{\bm{a}\cup\bm{t} - \bm{a}}_i )
	\leq \frac{1}{n}\Big( \frac{1}{n}2r_{\bm{t}} + \frac{1}{n}4\phi^{\bm{a}\cup\bm{t} - \bm{a}}_i\Big) \\
	&\leq \frac{2r_{\bm{t}}}{n^2} + \frac{b_Md_{\mathrm{U}}}{n} \rightarrow 0, i = 1, \dots, q.
	\end{align*}
	Then, $Var\Big[tr\Big\{\frac{1}{n}\mathrm{Y}^T \big(P_{\bm{a} \cup \bm{t}} - P_{\bm{a}} \big) \mathrm{Y}\Big\}\Big] \leq \sum_{i=1}^{q}\sum_{j=1}^{q} \sqrt{Var\Big[\frac{1}{n}\mathrm{y}_i^T \big(P_{\bm{a} \cup \bm{t}} - P_{\bm{a}} \big) \mathrm{y}_i\Big] Var\Big[\frac{1}{n}\mathrm{y}_j^T \big(P_{\bm{a} \cup \bm{t}} - P_{\bm{a}} \big) \mathrm{y}_j\Big]} \rightarrow 0$,
	as $n\rightarrow\infty$. So $Pr\Big\{tr\Big\{\frac{1}{n}\mathrm{Y}^T \big(P_{\bm{a} \cup \bm{t}} - P_{\bm{a}} \big) \mathrm{Y}\Big\} > C_0/4 \Big\} \rightarrow 1$.
	Let $\widetilde{\lambda}^{\bm{a}\cup\bm{t} - \bm{a}}_i, i=1,\dots,q$ be the eigenvalues of $\frac{1}{n}\mathrm{Y}^T \big(P_{\bm{a}\cup\bm{t}} - P_{\bm{a}} \big) \mathrm{Y}$, therefore
	\begin{align*}
	Pr\Big(\widetilde{\lambda}^{\bm{a}\cup\bm{t} - \bm{a}}_M > \frac{C_0}{4q}\Big) \geq Pr(\sum_{i=1}^{q} \widetilde{\lambda}^{\bm{a}\cup\bm{t} - \bm{a}}_i > C_0/4) =
	Pr\Big\{tr\Big\{\frac{1}{n}\mathrm{Y}^T \big(P_{\bm{a} \cup \bm{t}} - P_{\bm{a}} \big) \mathrm{Y}\Big\} > C_0/4 \Big\} \rightarrow 1.
	\end{align*}
	Let $\bar{\bar C} = C_0/4q$, then we have $Pr(\widetilde{\lambda}^{\bm{a}\cup\bm{t} - \bm{a}}_M > \bar{\bar{C}}) \rightarrow 1$, as $n\rightarrow\infty$.
\end{proof}

\begin{lemma}\label{lm8.6}
	Under Condition \ref{cond1}, \ref{cond2}, \ref{cond4} and \ref{cond6}, $\mathrm{p}\lim_{n\rightarrow\infty}\mathrm{BF}(\bm{a}\cup\bm{t};\bm{t}|G)=0$, when $\bm{a}\nsubseteq\bm{t}$.
\end{lemma}
\begin{proof}
	Case 1: If $p_{\bm{a}}$ is bounded, by Lemma \ref{lm8.2}, we know $\frac{1}{n} A_{\bm{t}}$ converges in probability to a positive definite constant matrix. So for large $n$, all eigenvalues of $\big(\frac{1}{n} A_{\bm{t}}\big)^{-\frac{1}{2}}$ are positive and $O_p(1)$. By Lemma \ref{lm8.4}, the largest eigenvalue of $\mathrm{Y}^T \big(P_{\bm{a}\cup\bm{t}} - P_{\bm{t}} \big) \mathrm{Y}$ is positive and $O_p(1)$. Since $g=O(n)$ and $d=O(1)$, we have the largest eigenvalue of $\frac{1}{d}\frac{g}{g+1} \big(\frac{1}{n} A_{\bm{t}}\big)^{-\frac{1}{2}} \big\{\frac{1}{n}\mathrm{Y}^T \big(P_{\bm{a} \cup \bm{t}} - P_{\bm{t}} \big) \mathrm{Y} \big\} \big(\frac{1}{n} A_{\bm{t}}\big)^{-\frac{1}{2}}$ is positive $O_p(\frac{1}{n})$. Therefore,
	\begin{equation*}
	\Delta(\bm{a}\cup\bm{t}, \bm{t})
	= \bigg|I_q +\frac{1}{d}\frac{g}{g+1} \bigg(\frac{1}{n} A_{\bm{t}}\bigg)^{-\frac{1}{2}} \bigg\{\frac{1}{n}\mathrm{Y}^T \big(P_{\bm{t}} - P_{\bm{a} \cup \bm{t}} \big) \mathrm{Y}\bigg\} \bigg(\frac{1}{n} A_{\bm{t}}\bigg)^{-\frac{1}{2}}\bigg| \succeq \bigg\{1-O_p\bigg(\frac{1}{n}\bigg)\bigg\}^h,
	\end{equation*}
	where $h$ is the number of nonzero eigenvalues of matrix $\frac{1}{d}\frac{g}{g+1}\Big(\frac{1}{n} A_{\bm{t}}\Big)^{-\frac{1}{2}} \Big\{\frac{1}{n}\mathrm{Y}^T \big(P_{\bm{t}} - P_{\bm{a} \cup \bm{t}} \big) \mathrm{Y}\Big\} \Big(\frac{1}{n} A_{\bm{t}}\Big)^{-\frac{1}{2}}$. Since the result also holds for every clique and separator, we have
	\begin{equation*}
	\RN{2}(\bm{a}\cup\bm{t}, \bm{t})
	=\frac
	{\prod_{C\in\mathscr{C}} {\big\{ \Delta_C (\bm{a}\cup\bm{t}, \bm{t}) \big\}}^{-\frac{b+n+|C|-1}{2}} }
	{\prod_{S\in\mathscr{S}} {\big\{ \Delta_S (\bm{a}\cup\bm{t}, \bm{t}) \big\}}^{-\frac{b+n+|S|-1}{2}} }
	\preceq \frac{\prod_{C\in\mathscr{C}}  \big\{1-O_p\big(\frac{1}{n}\big)\big\}^{-O(n)}}{\big\{1-O_p\big(\frac{1}{n}\big)\big\}^{O(n)}}
	= \bigg\{1-O_p\bigg(\frac{1}{n}\bigg)\bigg\}^{-O(n)},
	\end{equation*}
	where $\theta$ is a constant. So $\RN{2}(\bm{a}\cup\bm{t}, \bm{t}) \rightarrow$ some constant, as $n\rightarrow\infty$. Then
	\begin{equation*}
	\mathrm{BF}(\bm{a}\cup\bm{t},\bm{t}|G) = \RN{1} \times \RN{2}(\bm{a}\cup\bm{t}, \bm{t}) = \{O(n)+1\}^{-\frac{(p_{\bm{a}\cup \bm{t}}-p_{\bm{t}})(k+1)q}{2}} \times \text{constant} \rightarrow 0.
	\end{equation*}

	Case 2: If $p_{\bm{a}}$ is unbounded, similarly, we have
	\begin{equation*}
	\Delta(\bm{a}\cup\bm{t}, \bm{t})
	= \bigg|I_q +\frac{1}{d}\frac{g}{g+1} \bigg(\frac{1}{n} A_{\bm{t}}\bigg)^{-\frac{1}{2}} \bigg\{\frac{1}{n}\mathrm{Y}^T \big(P_{\bm{t}} - P_{\bm{a} \cup \bm{t}} \big) \mathrm{Y}\bigg\} \bigg(\frac{1}{n} A_{\bm{t}}\bigg)^{-\frac{1}{2}}\bigg| \succeq \bigg\{1-O_p\bigg(\frac{r_{\bm{a} \cup \bm{t}}}{n}\bigg)\bigg\}^h
	\end{equation*}
	and
	\begin{equation*}
	\RN{2}(\bm{a}\cup\bm{t}, \bm{t})
	=\frac
	{\prod_{C\in\mathscr{C}} {\big\{ \Delta_C (\bm{a}\cup\bm{t}, \bm{t}) \big\}}^{-\frac{b+n+|C|-1}{2}} }
	{\prod_{S\in\mathscr{S}} {\big\{ \Delta_S (\bm{a}\cup\bm{t}, \bm{t}) \big\}}^{-\frac{b+n+|S|-1}{2}} }
	\preceq \frac{\prod_{C\in\mathscr{C}}  \big\{1-O_p\big(\frac{r_{\bm{a} \cup \bm{t}}}{n}\big)\big\}^{-O(n)}}{\big\{1-O_p\big(\frac{r_{\bm{a} \cup \bm{t}}}{n}\big)\big\}^{O(n)}}
	= \bigg\{1-O_p\bigg(\frac{r_{\bm{a} \cup \bm{t}}}{n}\bigg)\bigg\}^{-O(n)}.
	\end{equation*}
	Then,
	\begin{align*}
	log\{\mathrm{BF}(\bm{a}\cup\bm{t},\bm{t}|G)\} 
	& \preceq -\frac{(p_{\bm{a}\cup \bm{t}}-p_{\bm{t}})(k+1)q}{2} log\{O(n)+1\} - O(n) log \bigg\{1-O_p\bigg(\frac{r_{\bm{a} \cup \bm{t}}}{n}\bigg)\bigg\}\\
	& \preceq -O_p(r_{\bm{a} \cup \bm{t}}) log\{O(n)+1\} - O_p(r_{\bm{a} \cup \bm{t}}) O_p\bigg(\frac{n}{r_{\bm{a} \cup \bm{t}}}\bigg)log \bigg\{1-O_p\bigg(\frac{r_{\bm{a} \cup \bm{t}}}{n}\bigg)\bigg\}\\
	& = -O_p(r_{\bm{a} \cup \bm{t}}) log\{O(n)+1\} \rightarrow -\infty (\text {as } log(1+x)/x \rightarrow 1, \text { as } x \rightarrow 0).
	\end{align*}
	Therefore, $\mathrm{BF}(\bm{a}\cup\bm{t};\bm{t}|G) \rightarrow 0$, as $n\rightarrow\infty$.
\end{proof}

\begin{lemma}\label{lm8.7}
Under Condition \ref{cond1}-\ref{cond6}, $\mathrm{p}\lim_{n\rightarrow\infty} \mathrm{BF}(\bm{a};\bm{a}\cup\bm{t}|G) = 0$, when $\bm{t}\nsubseteq\bm{a}$.
\end{lemma}
\begin{proof}
Let $\hat{\Sigma}_G^{-1}$ be the MLE of $\Sigma_G^{-1}$ under model $\bm{a}$, then $f(\mathrm{Y}|\bm{a}, \Sigma_G)\leq f(\mathrm{Y}|\bm{a}, \hat{\Sigma}_G)$ for any positive definite matrix $\Sigma_G$ under the given graph $G$. The explicit calculation of the MLE can be done by calculating the MLEs of cliques, separators and combining them. We assume that the MLE converges to a positive definite matrix ${\Sigma_G^0}^{-1}$. For the true graph $G^*$, this statement holds trivially. Under supremum norm for each clique and separator, given $0<\epsilon<1$, we have a $\epsilon'$-neighborhood $Nb(\epsilon')$ of ${\Sigma_G^0}^{-1}$, where $0<\epsilon'<\epsilon$, which satisfies $Nb(\epsilon') = \{ \Sigma_G^{-1} : \|\Sigma_G^{-1}-{\Sigma_G^0}^{-1}\|_{\infty}<\epsilon'\}$ and $Pr\{Nb(\epsilon')\}>\delta'>0$ under HIW  prior, such that $|{\Sigma}_G^{-1}\Sigma_G^0|<1+\epsilon$, $|\Sigma_G{\Sigma_G^0}^{-1}|<1+\epsilon$. For large $n$, we also have $\hat{\Sigma}_G^{-1}\in Nb(\epsilon')$. So $|\hat{\Sigma}_G^{-1}\Sigma_G^0|<1+\epsilon$, $|\hat{\Sigma}_G{\Sigma_G^0}^{-1}|<1+\epsilon$ and $\|\hat{\Sigma}_G^{-1}-{\Sigma_G^0}^{-1}\|_{\infty}<\epsilon'$. 

Now dividing numerator and denominator of $\mathrm{BF}(\bm{a};\bm{a}\cup\bm{t}|G)$ by $f(\mathrm{Y}|\bm{a}, \hat{\Sigma}_G)$, the likelihood at MLE under model $\bm{a}$,
\begin{align*}
\mathrm{BF}(\bm{a};\bm{a}\cup\bm{t}|G)
& = \frac{\int f(\mathrm{Y}|\bm{a}, \Sigma_G)f(\Sigma_G|G)d\Sigma_G}{\int f(\mathrm{Y}|\bm{a}\cup\bm{t}, \Sigma_G)f(\Sigma_G|G)d\Sigma_G} \\
& = \frac{\int\frac{f(\mathrm{Y}|\bm{a}, \Sigma_G)}{f(\mathrm{Y}|\bm{a}, \hat{\Sigma}_G)} f(\Sigma_G|G)d\Sigma_G}{\int\frac{f(\mathrm{Y}|\bm{a}\cup\bm{t}, \Sigma_G)}{f(\mathrm{Y}|\bm{a}, \hat{\Sigma}_G)} f(\Sigma_G|G)d\Sigma_G} \\
& < \frac{\int f(\Sigma_G|G)d\Sigma_G}{\int_{Nb(\epsilon')}\frac{f(\mathrm{Y}|\bm{a}\cup\bm{t}, \Sigma_G)}{f(\mathrm{Y}|\bm{a}, \hat{\Sigma}_G)} f(\Sigma_G|G)d\Sigma_G} \\
& = \frac{(g+1)^{\frac{(p_{\bm{a}\cup\bm{t}}-p_{\bm{a}})(k+1)q}{2}}}
{\int_{Nb(\epsilon')}|\Sigma_G{\hat{\Sigma}_G}^{-1}|^{-\frac{n}{2}}
exp\big[-\frac{1}{2}tr\{S(\bm{a}\cup\bm{t})\Sigma_G^{-1}-S(\bm{a})\hat{\Sigma}_G^{-1} \}\big]
 f(\Sigma_G|G)d\Sigma_G} \\ 
& = \frac{(g+1)^{\frac{(p_{\bm{a}\cup\bm{t}}-p_{\bm{a}})(k+1)q}{2}}}
{\int_{Nb(\epsilon')}|\Sigma_G{\Sigma_G^0}^{-1}|^{-\frac{n}{2}} |\Sigma_G^0{\hat{\Sigma}_G}^{-1}|^{-\frac{n}{2}}
exp\big[-\frac{1}{2}tr\{S(\bm{a}\cup\bm{t})\Sigma_G^{-1}-S(\bm{a})\hat{\Sigma}_G^{-1} \}\big]
 f(\Sigma_G|G)d\Sigma_G} \\
& = \frac{(g+1)^{\frac{(p_{\bm{a}\cup\bm{t}}-p_{\bm{a}})(k+1)q}{2}} (1+\epsilon)^n}
{\int_{Nb(\epsilon')}
exp\big[-\frac{1}{2}tr\{S(\bm{a}\cup\bm{t})\Sigma_G^{-1}-S(\bm{a})\hat{\Sigma}_G^{-1} \}\big]
f(\Sigma_G|G)d\Sigma_G}.
\end{align*}

Next, let $\alpha$ be a $q \times 1$ vector, where $\alpha\in\mathbb{R}^q$, such that $\alpha^T \frac{S(\bm{a})-S(\bm{a}\cup\bm{t})}{n} \alpha = \widetilde{\lambda}^{\bm{a}\cup\bm{t} - \bm{a}}_M$. Let $\beta = {\Sigma_G^0}^{-1/2}\alpha$ and $b_{\beta} = \| \beta \|_2^2 < \infty$. Denote $\lambda_M'$ be the largest eigenvalue of ${\Sigma_G^0}^{-1/2}\frac{S(\bm{a})-S(\bm{a}\cup\bm{t})}{n}{\Sigma_G^0}^{-1/2}$, then $\beta^T{\Sigma_G^0}^{-1/2}\frac{S(\bm{a})-S(\bm{a}\cup\bm{t})}{n}{\Sigma_G^0}^{-1/2}\beta = \widetilde{\lambda}^{\bm{a}\cup\bm{t} - \bm{a}}_M \leq \lambda_M'\| \beta \|_2^2$. By Lemma \ref{lm8.5}, $Pr(\lambda_M' > \bar{\bar{C}}/b_{\beta}) \rightarrow 1$, as $n\rightarrow\infty$.


By Lemma \ref{lm8.3} and Lemma \ref{lm8.5}, we have
\begin{equation*}
Pr\Big(\lambda_M'
-\Big|tr\Big\{\frac{S(\bm{a}\cup\bm{t})}{n}(\hat{\Sigma}_G^{-1}-{\Sigma_G^0}^{-1})\Big\}\Big|
-\Big|tr\Big\{ \frac{S(\bm{a})}{n}({\Sigma_G^{0}}^{-1}-\hat{\Sigma}_G^{-1}) \Big\}\Big| > \bar{\bar{C}}/b_{\beta} - 2q\epsilon\bar{C} \Big) \rightarrow 1.
\end{equation*}
Then, by choosing $\epsilon < \frac{\bar{\bar{C}}}{2qb_{\beta}\bar{C}}$, we know $\lambda_M'
-\Big|tr\Big\{\frac{S(\bm{a}\cup\bm{t})}{n}(\hat{\Sigma}_G^{-1}-{\Sigma_G^0}^{-1})\Big\}\Big|
-\Big|tr\Big\{ \frac{S(\bm{a})}{n}({\Sigma_G^{0}}^{-1}-\hat{\Sigma}_G^{-1}) \Big\}\Big| > \bar{\bar{C}}/(2b_{\beta})$ in probability.

So, we have
\begin{align*}
& -\frac{1}{2}tr\{S(\bm{a}\cup\bm{t})\Sigma_G^{-1}-S(\bm{a})\hat{\Sigma}_G^{-1} \}\\
& = \frac{n}{2}\bigg[tr\Big\{{\Sigma_G^0}^{-1}\frac{S(\bm{a})-S(\bm{a}\cup\bm{t})}{n}\Big\}
-tr\Big\{\frac{S(\bm{a}\cup\bm{t})}{n}(\hat{\Sigma}_G^{-1}-{\Sigma_G^0}^{-1})\Big\}
-tr\Big\{ \frac{S(\bm{a})}{n}({\Sigma_G^{0}}^{-1}-\hat{\Sigma}_G^{-1}) \Big\}\bigg]\\
& \geq \frac{n}{2}\bigg[tr\Big\{{\Sigma_G^0}^{-1/2}\frac{S(\bm{a})-S(\bm{a}\cup\bm{t})}{n}{\Sigma_G^0}^{-1/2}\Big\}
-\Big|tr\Big\{\frac{S(\bm{a}\cup\bm{t})}{n}(\hat{\Sigma}_G^{-1}-{\Sigma_G^0}^{-1})\Big\}\Big|
-\Big|tr\Big\{ \frac{S(\bm{a})}{n}({\Sigma_G^{0}}^{-1}-\hat{\Sigma}_G^{-1}) \Big\}\Big|\bigg]\\
& \geq \frac{n}{2}\bigg[\lambda_M'
-\Big|tr\Big\{\frac{S(\bm{a}\cup\bm{t})}{n}(\hat{\Sigma}_G^{-1}-{\Sigma_G^0}^{-1})\Big\}\Big|
-\Big|tr\Big\{ \frac{S(\bm{a})}{n}({\Sigma_G^{0}}^{-1}-\hat{\Sigma}_G^{-1}) \Big\}\Big|\bigg].
\end{align*}
Then $Pr\Big\{ -\frac{1}{2}tr\{S(\bm{a}\cup\bm{t})\Sigma_G^{-1}-S(\bm{a})\hat{\Sigma}_G^{-1} \} > \widetilde{C}n \Big\} \rightarrow 1$, as $n\rightarrow\infty$, where $\widetilde{C}$ is some fixed constant.

Since $BF_{\bm{a}, \bm{a}\cup\bm{t}}< \frac{(g+1)^{\frac{(p_{\bm{a}\cup\bm{t}}-p_{\bm{a}})(k+1)q}{2}} (1+\epsilon)^n}{\int_{Nb(\epsilon')}	exp\big[-\frac{1}{2}tr\{S(\bm{a}\cup\bm{t})\Sigma_G^{-1}-S(\bm{a})\hat{\Sigma}_G^{-1} \}\big]f(\Sigma_G|G)d\Sigma_G}$, then $Pr\big\{ BF_{\bm{a}, \bm{a}\cup\bm{t}} < (g+1)^{\frac{p_{\bm{t}}(k+1)q}{2}} (1+\epsilon)^n e^{-\widetilde{C}n}\big\} \rightarrow 1$. Therefore, $\mathrm{p}\lim_{n\rightarrow\infty} \mathrm{BF}(\bm{a};\bm{a}\cup\bm{t}|G) = 0$.
\end{proof}

By combining the results from Lemma \ref{lm8.6} and \ref{lm8.7}, we have 
\begin{equation*}
\mathrm{p}\lim_{n\rightarrow\infty} \mathrm{BF}(\bm{a};\bm{t}|G) = \mathrm{p}\lim_{n\rightarrow\infty} \mathrm{BF}(\bm{a};\bm{a}\cup\bm{t}|G) \cdot \mathrm{p}\lim_{n\rightarrow\infty} \mathrm{BF}(\bm{a}\cup\bm{t};\bm{t}|G) = 0,
\end{equation*}
for any model $\bm{a}\neq\bm{t}$.

\section*{Appendix 3.}

\begin{figure}[H] \label{sApop}
	\centering
	\includegraphics[width=1\textwidth]{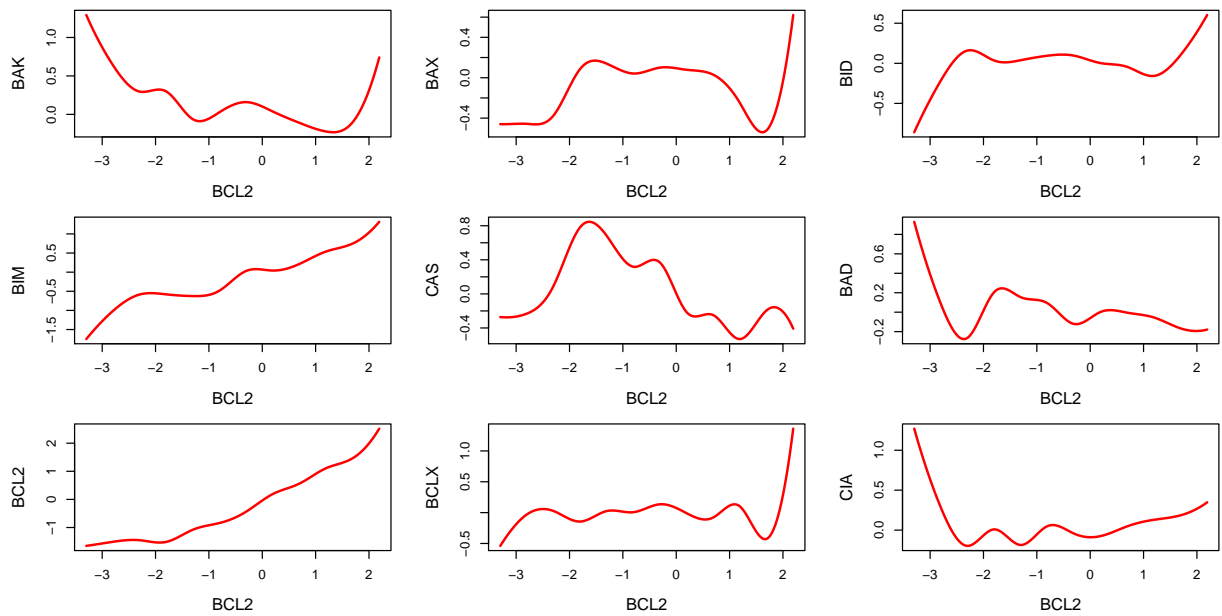}
	\caption{posterior mean of the nonlinear functions for proteins in apoptosis pathway, mRNA selected is BCL2.}
\end{figure}

\begin{figure}[H] \label{sCell}
	\centering
	\includegraphics[width=1\textwidth]{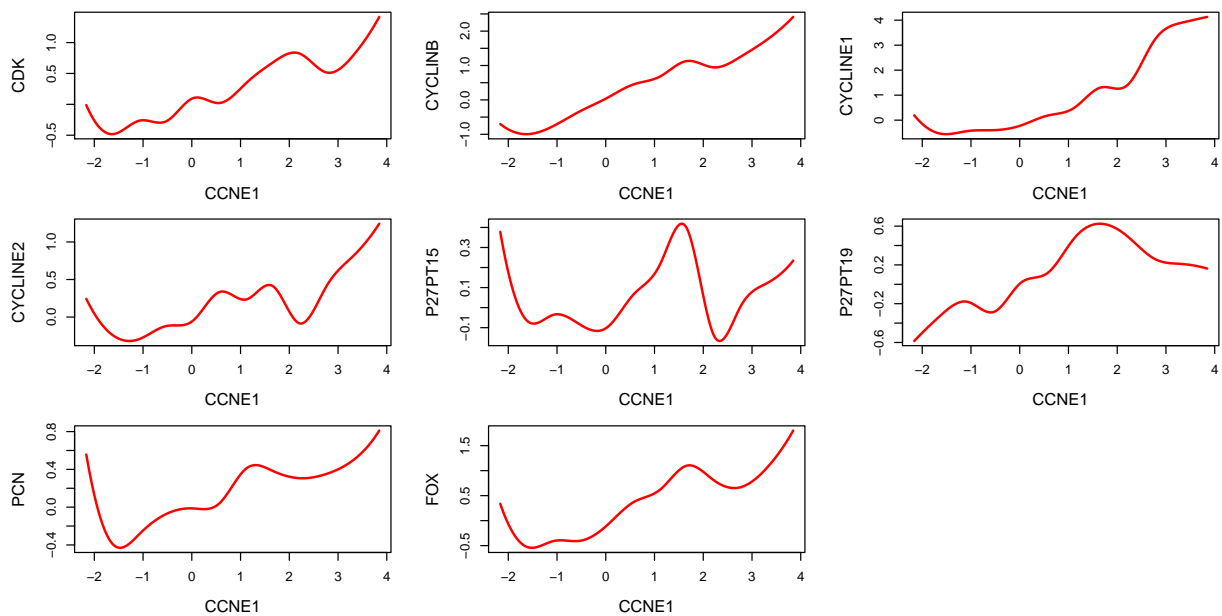}
	\caption{posterior mean of the nonlinear functions for proteins in cell cycle pathway, mRNA selected is CCNE1.}
\end{figure}

\begin{figure}[H] \label{sCore}
	\centering
	\includegraphics[width=1\textwidth]{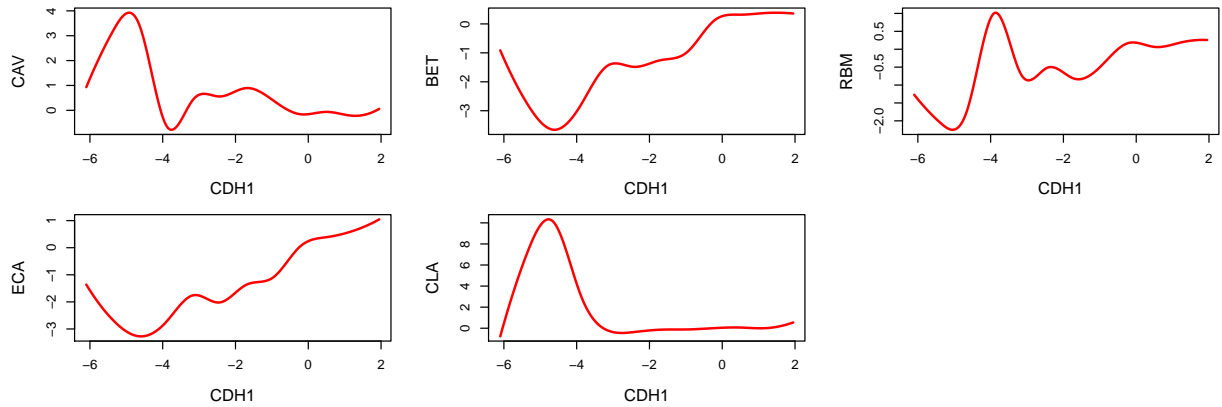}
	\caption{posterior mean of the nonlinear functions for proteins in core reactive pathway, mRNA selected is CDH1.}
\end{figure}

\begin{figure}[H] \label{sEMT}
	\centering
	\includegraphics[width=1\textwidth]{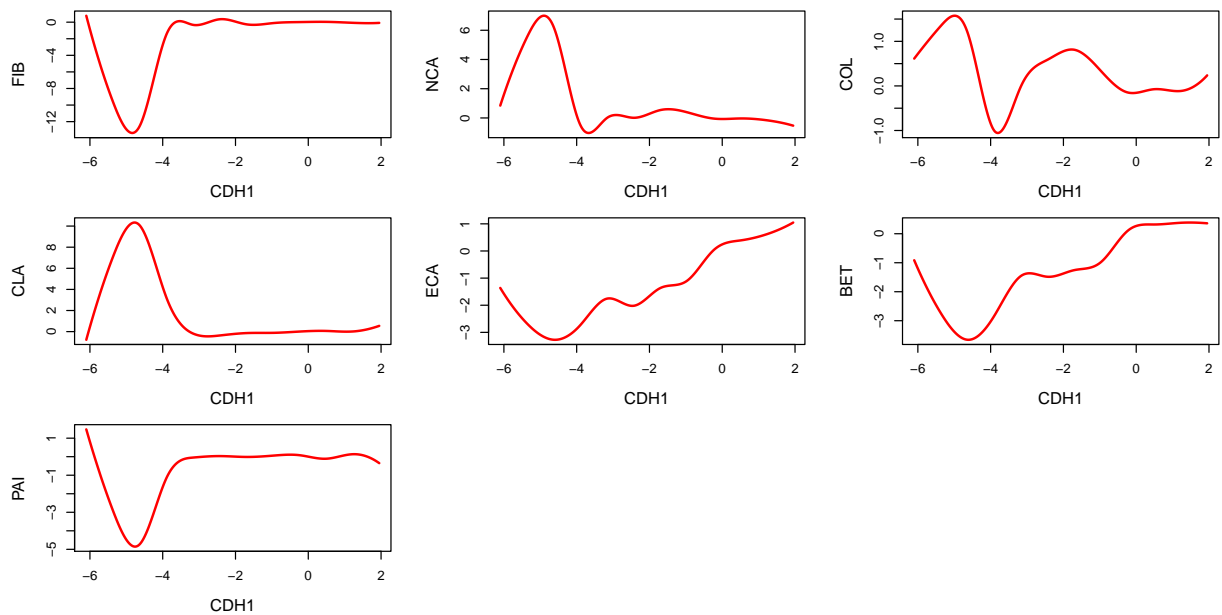}
	\caption{posterior mean of the nonlinear functions for proteins in EMT pathway, mRNA selected is CDH1.}
\end{figure}

\begin{figure}[H] \label{sPI3K}
	\centering
	\includegraphics[width=1\textwidth]{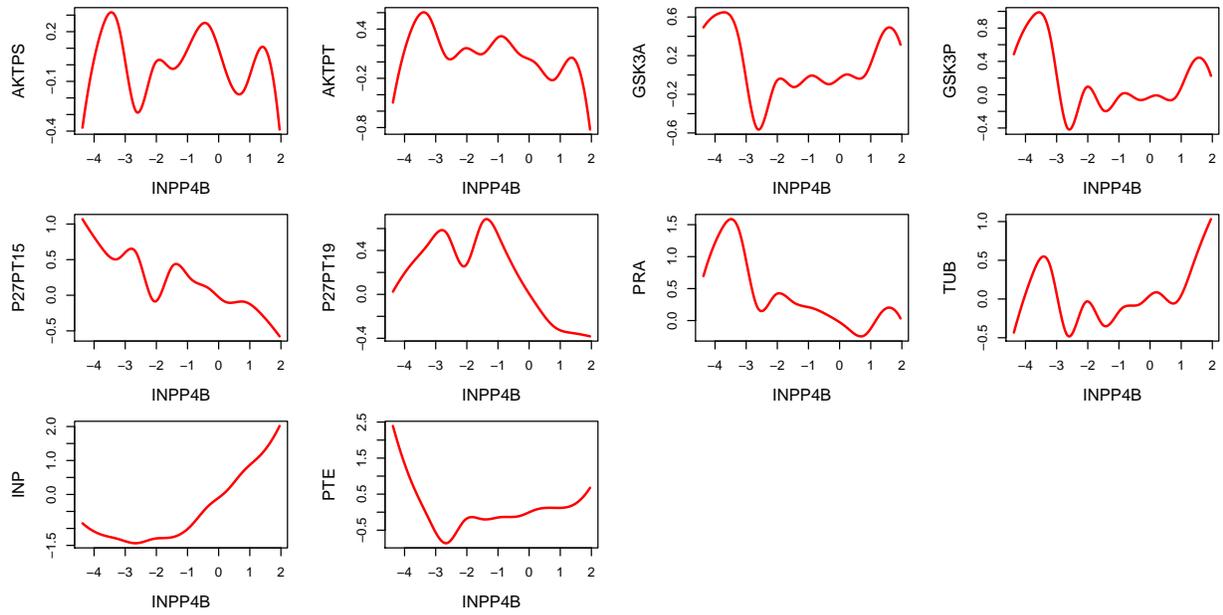}
	\caption{posterior mean of the nonlinear functions for proteins in PI3K/AKT pathway, mRNA selected is INPP4B.}
\end{figure}

\begin{figure}[H] \label{sRTK}
	\centering
	\includegraphics[width=1\textwidth]{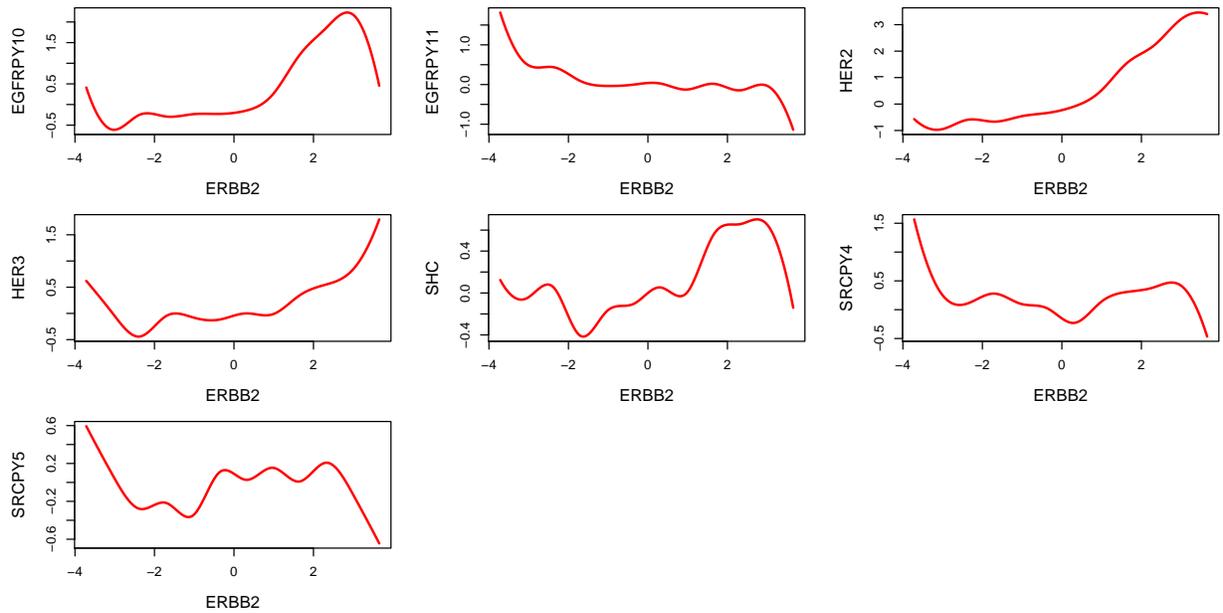}
	\caption{posterior mean of the nonlinear functions for proteins in RTK pathway, mRNA selected is ERBB2.}
\end{figure}

\begin{figure}[H] \label{sHorm1}
	\centering
	\includegraphics[width=1\textwidth]{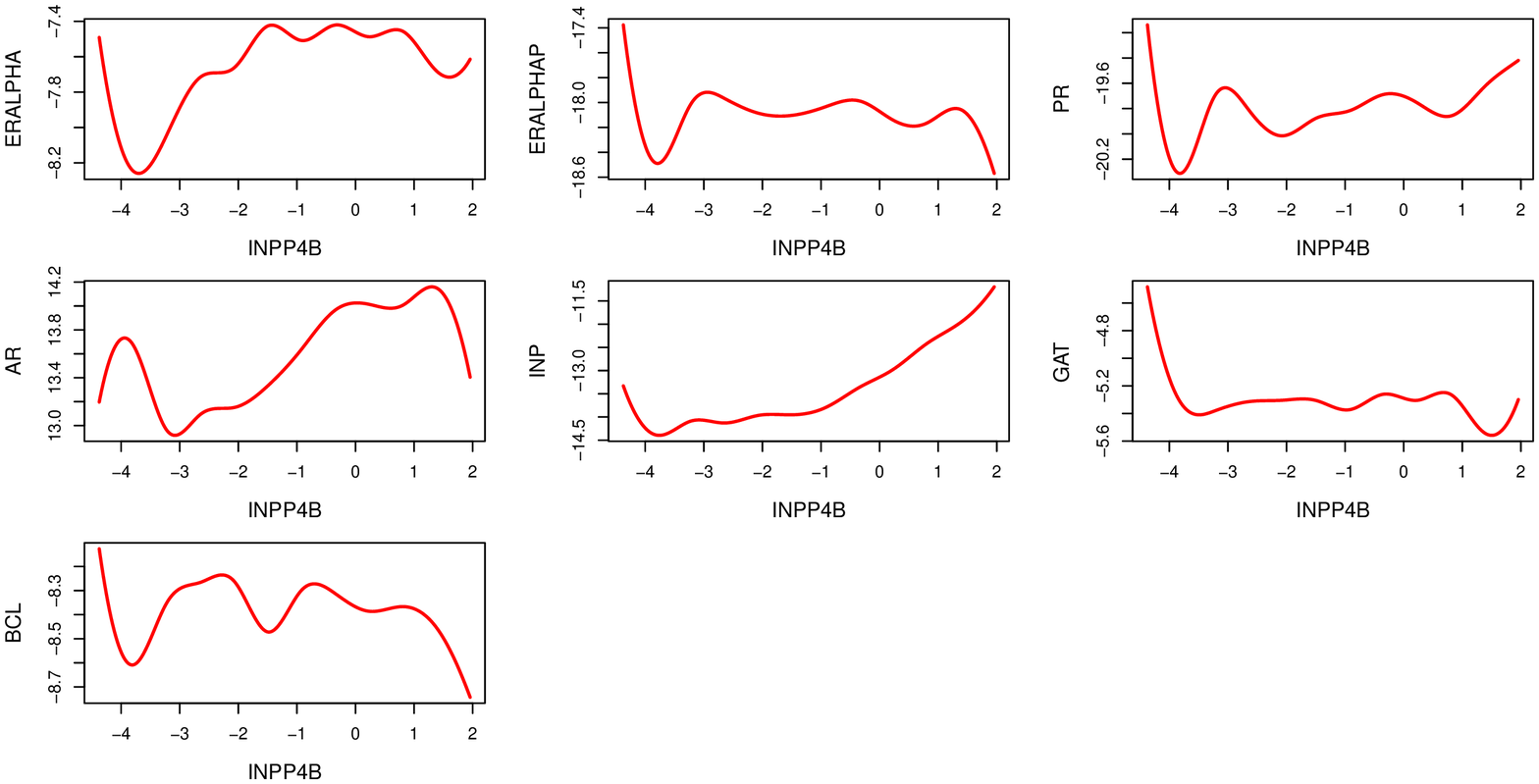}
	\caption{posterior mean of the nonlinear functions for proteins in hormone receptor\&signaling pathway, mRNA selected is INPP4B.}
\end{figure}

\begin{figure}[H] \label{sHorm2}
	\centering
	\includegraphics[width=1\textwidth]{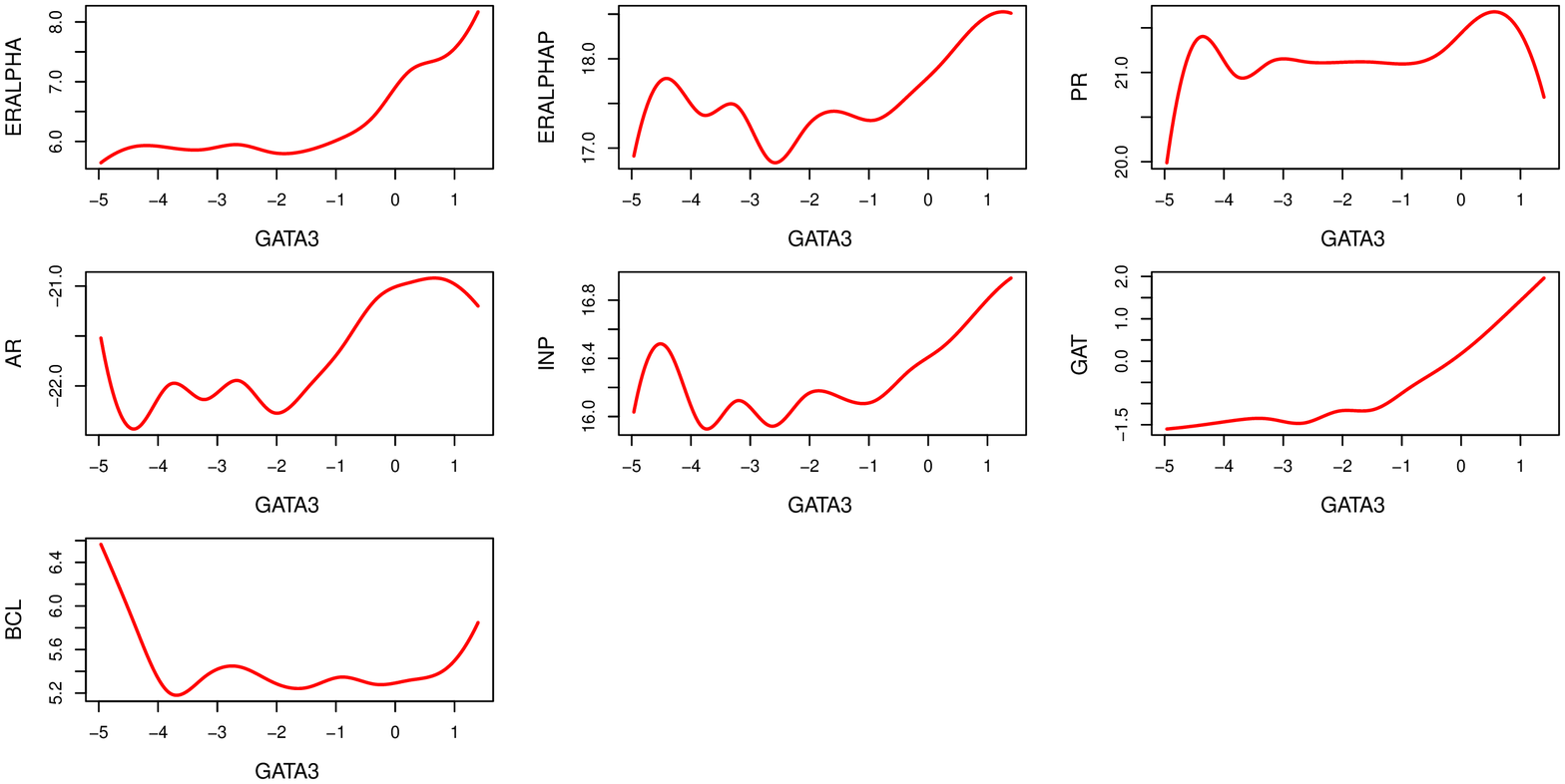}
	\caption{posterior mean of the nonlinear functions for proteins in hormone receptor\&signaling pathway, mRNA selected is GATA3.}
\end{figure}

\begin{figure}[H] \label{sHorm3}
	\centering
	\includegraphics[width=1\textwidth]{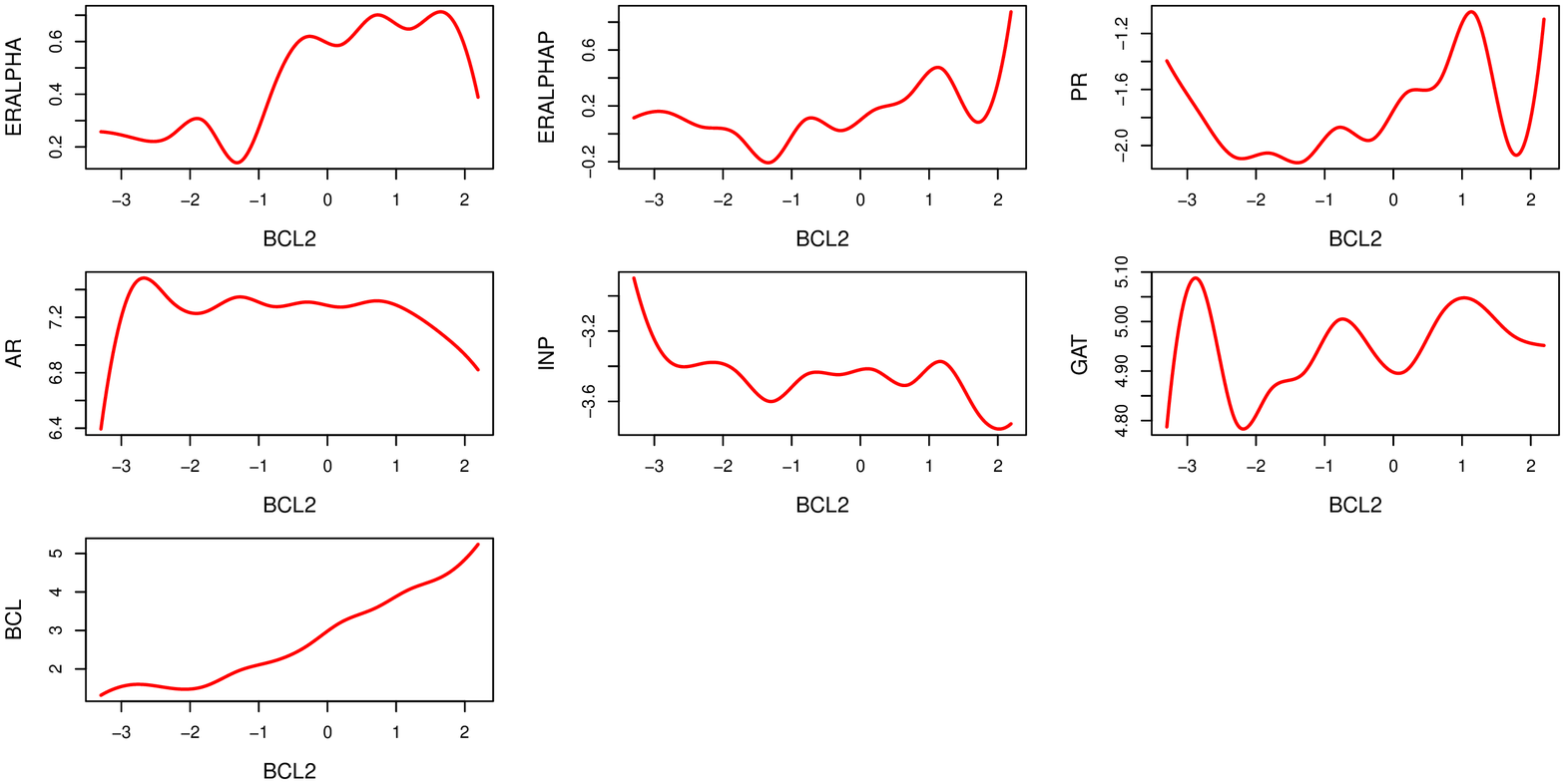}
	\caption{posterior mean of the nonlinear functions for proteins in hormone receptor\&signaling pathway, mRNA selected is BCL2.}
\end{figure}

\end{document}